\documentclass[aps,prb,twocolumn,superscriptaddress,10pt]{revtex4-2}

\usepackage{geometry}

\usepackage{amsmath}
\usepackage{amssymb}
\usepackage{amsthm}
\usepackage{amsfonts}
\usepackage{mathrsfs}
\usepackage{latexsym}

\usepackage[table]{xcolor}
\usepackage{booktabs}
\usepackage{multirow}

\usepackage{graphicx}
\usepackage{float}          
\usepackage[font=footnotesize,labelfont=bf]{caption}

\usepackage{hyperref}
\hypersetup{
    colorlinks=true,
    linkcolor=blue,
    citecolor=red,
    urlcolor=blue
}

\usepackage{enumerate}
\usepackage[normalem]{ulem}
\usepackage{comment}
\usepackage{ragged2e}

\usepackage{adjustbox}      
\usepackage{tabularx}
\usepackage{array}

\usepackage{tikz}
\usetikzlibrary{shapes.geometric, arrows.meta, positioning, calc, backgrounds}

\newtheorem{lemma}{Lemma}

\begin{document}

\title{Majorana Constellations: A Geometric Lens on Multipartite Entanglement and Geometric Phases}

\author{Chon-Fai Kam}
\email{dubussygauss@gmail.com}
\affiliation{Quantum Theory Group, Dipartimento di Fisica e Chimica Emilio Segr\`e,\\
Universit\`a degli Studi di Palermo, Via Archirafi 36, I-90123 Palermo, Italy}
\affiliation{University Paris City \& University of Reunion, DSIMB, Inserm, BIGR U1134, 75015, Paris, France}

\date{\today}

\newgeometry{left=2cm,right=2cm}

\begin{abstract}
The Majorana stellar representation maps a pure spin-$S$ state to $2S$ points on a sphere. This review develops it with entanglement as the organising principle, and two objects recur throughout: the constellation, and the permanent of the Gram matrix of its stars. The degeneracy pattern of the constellation is invariant under stochastic local operations and classical communication, so the integer partitions of $N$ label a finite set of families of symmetric $N$-qubit states. That labelling is a coarse-graining rather than a classification, since from four distinct stars onwards each family carries continuous M\"obius moduli. The permanent supplies what the pattern omits. Normalised by it, inter-star chordal distances give the concurrence and the three-tangle in closed form, and the same permanent governs the anomalous contribution to the Berry phase acquired under adiabatic cyclic evolution, so that a single quantity links static correlations to dynamical holonomy. We also fix the computational reach of the geometry. Overlaps of symmetric states are permanents of matrices of rank at most two and are polynomially computable, whereas measures defined by an optimisation over the sphere are not reached by that argument. Interest in stellar representations has resurged, but the literature remains dispersed, and no existing treatment develops the link between constellation geometry, multipartite entanglement, and geometric phases within a single framework. The same two objects organise the applications reviewed here, from extremal states in metrology and permutation-invariant codes to collective spin models and photonic constellations, together with extensions to mixed states and to continuous-variable systems through the stellar rank. Whether the anomalous phase admits a bound in terms of any entanglement monotone remains open.
\end{abstract}

\maketitle

\section{Introduction}
\label{sec:introduction}

Quantum entanglement remains one of the most profound features of quantum mechanics \cite{Nielsen2010, Wilde2017, Chitambar2019, Bengtsson2017}. It underpins applications ranging from quantum algorithms and secure key distribution to fault-tolerant computing \cite{Preskill2018, Bharti2022, Portmann2022}, while also driving foundational tests such as Bell inequality violations and Kochen-Specker contextuality tests \cite{BIGBellTest2018, Budroni2022} that probe the boundary between local realism and quantum non-locality. While algebraic tools such as concurrence \cite{Wootters1998}, three-tangle \cite{Coffman2000}, negativity \cite{Vidal2002}, and entanglement witnesses \cite{Guhne2009} provide quantitative measures, they often obscure the intuitive geometric structures of entanglement. This geometric perspective becomes particularly valuable for symmetric multi-qubit states and for bosonic ensembles, where the wave function is invariant under arbitrary particle permutations. For such systems---which occur naturally in bosonic ensembles or collective spin systems---the Majorana stellar representation offers an elegant visual framework. 

Proposed by Ettore Majorana in 1932 \cite{Majorana1932}, the Majorana representation describes arbitrary spin-$S$ states as symmetrized combinations of $2S$ spin-$1/2$ particles. It establishes an isomorphism between the spin-$S$ Hilbert space (the $(2S+1)$-dimensional irreducible representation of $\text{SU}(2)$) and the totally symmetric subspace of $N=2S$ qubits, on which the collective action $U^{\otimes N}$ of $\mathrm{SU}(2)$ carries precisely the spin-$N/2$ irreducible representation \cite{UshaDevi2012, Bengtsson2017}, bridging angular momentum theory with multi-qubit entanglement. In this framework, a pure spin-$S$ state is mapped to a unique constellation of $2S$ points (Majorana stars) on the unit sphere. Their positions are determined as the complex roots $z_i$ of the associated polynomial of degree $2S$ \cite{Majorana1932},
\begin{equation}
P(z) = \sum_{m=-S}^{S} (-1)^{S+m} \sqrt{\binom{2S}{S+m}} \, c_m \, z^{S+m},
\label{eq:intro_majorana_polynomial}
\end{equation}
where $c_m$ is the amplitude of $|S,m\rangle$, equivalently of the Dicke state with $S+m$ excitations. This convention is used throughout and is restated in Eq.~\eqref{eq:standard_Majorana_polynomial}. These roots are then mapped to spherical coordinates via stereographic projection. This geometric picture transforms abstract Hilbert-space vectors into arrangements on a sphere, exposing symmetries, degeneracies, and correlations hidden in coefficient expansions.

Following its inception in 1932, Majorana's geometric approach attracted little attention. Written in Italian, and offering only a few lines of justification, the construction appeared at a moment when atomic and nuclear spectroscopy were moving toward the algebraic tensor methods of Racah \cite{racah1942theory} and Wigner \cite{wigner1940matrices}. Its early champion was Schwinger, who came upon the paper as a teenager, worked the construction out for himself, and published the verification in 1937 \cite{schwinger1937nonadiabatic}; yet the 1952 monograph in which he developed the boson representation that now carries his name makes no mention of the Majorana formula at all \cite{schwinger1952angular, sanchezsoto2026quantum}. The geometric picture was therefore less displaced than never taken up. However, the same symmetric-spinor structure resurfaced in Penrose's spinor approach to general relativity \cite{penrose1960spinor}, developed systematically in the monograph with Rindler \cite{penrose1984spinors, penrose1986spinors}; Penrose later recognised the principal null directions of that formalism as Majorana constellations and used them, for spin-$3/2$ states, to construct a probability-free geometric argument of Kochen--Specker type against non-contextual hidden variables \cite{zimba1993bell}. In parallel, and independently of Majorana's construction, spin coherent states were developed as the $\mathrm{SU}(2)$ analogue of the canonical coherent states \cite{radcliffe1971some, arecchi1972atomic}; they enter the stellar picture as the fully degenerate constellations. A renaissance followed in the 1990s. This revival was driven both by its application to quantum chaos \cite{leboeuf1990chaos} and by the work of Hannay linking stellar trajectories to geometric phases \cite{Hannay1998a, Hannay1998b}. This work demonstrated that the Berry phase acquired during the adiabatic cyclic evolution of a spin-$S$ state equals minus half the total solid angle subtended by the closed paths of the $2S$ stars, augmented by anomalous contributions arising from the constellation's internal relative motions and geometry.

At the turn of the twenty-first century, the framework branched into many-body physics and quantum information science, providing a powerful bridge between algebraic structures and geometric intuition. It proved instrumental in mapping quantum phase transitions and thermodynamic limits in spinor Bose-Einstein condensates \cite{makela2007inert} and the Lipkin-Meshkov-Glick (LMG) model \cite{ribeiro2007thermodynamical}. In the realm of entanglement, degeneracy patterns within the constellation---the number, multiplicity, and grouping of coincident stars---were found to be invariant under Stochastic Local Operations and Classical Communication (SLOCC), and hence to label \emph{families} of pure symmetric $N$-qubit states \cite{bastin2009operational, aulbach2010maximally, markham2011entanglement}. This mapping furnishes a visual and computationally efficient taxonomy at the level of families; as we discuss in Sec.~\ref{sec:entanglement}, it does not exhaust the SLOCC orbits, since from four distinct stars onwards each family carries continuous M{\"o}bius moduli. Furthermore, beyond discrete classification, the continuous spatial features of the constellation directly encode entanglement metrics. Ganczarek \textit{et al.}~\cite{ganczarek2012barycentric} built an entanglement measure directly from the barycentre of the constellation, its distance from the origin quantifying how far the state is from a spin coherent state, and configurations with vanishing barycentre reappear later as the optimal states for high-precision metrology \cite{toth2014quantum}.

The last decade has witnessed a profound theoretical and practical expansion of this framework, transitioning the Majorana picture from a foundational curiosity into a valuable tool for modern quantum technologies. Significant strides have been made in addressing long-standing foundational limitations. Notably, initial frameworks generalizing the stellar representation from pure to mixed states \cite{serrano2020majorana} and the direct extraction of physical multipoles from constellation geometry \cite{romero2024multipoles} have broadened its applicability. Because these geometric features inherently capture multipartite entanglement, these theoretical advances have translated into applied contexts, enabling the identification of extremal states that maximize the quantum Fisher information for high-precision sensors \cite{goldberg2020extremal} and the design of permutation-invariant quantum error-correcting codes \cite{ouyang2014permutation}. Concurrently, the characterization of anticoherent spin states has provided foundational resources for rotation sensing \cite{giraud2015tensor}. Extremal constellations have been realised and used for metrology in a multiport photonic interferometer \cite{bouchard2017quantum}, and, combined with the emergence of Rydberg atom arrays \cite{bernien2017probing, ebadi2021quantum} and cavity-coupled atomic ensembles \cite{hosten2016measurement}, these developments have brought Majorana's 1932 vision into the realm of laboratory reality; a concise survey of the field's present state is given in Ref.~\cite{sanchezsoto2026quantum}.

The literature on the Majorana representation, though now extensive, remains dispersed. The framework is introduced in textbook chapters (such as that of Bengtsson and {\.Z}yczkowski \cite{Bengtsson2017}) and outlined within individual research papers \cite{UshaDevi2012}; substantial reviews exist of particular facets, notably the extremal and maximally unpolarized states \cite{goldberg2020extremal} and, more recently, a concise survey of the representation as a whole \cite{sanchezsoto2026quantum}. What these treatments do not attempt is a sustained development of the connections between constellation geometry, genuine multipartite entanglement, and geometric phases within a single framework. That is the gap this review addresses. By organising the material around entanglement, we emphasize how the constellation reveals structure in quantum correlations and how these static geometric insights intertwine with dynamical geometric phases.

The review is structured as follows, and its overall architecture is summarised in Fig.~\ref{fig:roadmap}. Section~\ref{sec:history} outlines the historical development and early geometric insights. Section~\ref{sec:foundations} reviews the mathematical foundations of Majorana stars, including polynomial representations and the fundamental isomorphism for symmetric multi-qubit states. Section~\ref{sec:tools} introduces the geometric quantities derived from constellations, such as inter-star distances, multipole moments, stellar rank, and quasi-probability distributions. Section~\ref{sec:entanglement} delves into entanglement in the Majorana picture, bridging discrete SLOCC classification with continuous geometric measures like the three-tangle and genuine multipartite invariants. Section~\ref{sec:phases} explores unitary dynamics, focusing on the geometric phase anomalies and semiclassical Hannay angles. Section~\ref{sec:applications} discusses applications spanning quantum metrology, symmetric error-correcting codes, condensed-matter analogs, and photonic implementations. Section~\ref{sec:challenges} highlights open challenges, including universal mixed-state extensions, computational scalability and machine-learning solutions, and experimental and thermodynamic horizons. We conclude in Sec.~\ref{sec:conclusions} with perspectives on the representation's role in future quantum technologies. Five appendices supply the supporting derivations: the group-theoretic isomorphism (Appendix~\ref{app:isomorphism}), its spinor-calculus formulation (Appendix~\ref{app:spinor_calculus}), the Bargmann-Fock construction (Appendix~\ref{app:bargmann}), the permanent-determinant complexity analysis (Appendix~\ref{app:complexity}), and the antipodal orthogonality algorithm (Appendix~\ref{app:orthogonality}).

\begin{figure*}[t]
\centering
\resizebox{\textwidth}{!}{
\begin{tikzpicture}[
    node distance=1.3cm and 1.0cm,
    box/.style={
        rectangle,
        rounded corners,
        minimum width=6.2cm,
        minimum height=1.3cm,
        align=center, 
        draw=black,
        fill=black!5,
        font=\small\sffamily\bfseries,
        inner sep=6pt
    },
    mergebox/.style={
        box,
        dashed,
        fill=white,
        draw=black!70
    },
    arrow/.style={
        -{Stealth[scale=1.2]},
        thick,
        draw=black!70
    },
    label/.style={
        font=\large\bfseries\itshape,
        align=center,
        minimum width=6.2cm
    }
]

\node (label_a) [label] {\uppercase{Foundations}};

\node (ch1) [box, below=of label_a, yshift=0.3cm] {
    Sec. I. Introduction \\ 
    \scriptsize{(Background, Motivation \& Roadmap)}
};

\node (ch2) [box, below=of ch1] {
    Sec. II. Historical Development \\ 
    \scriptsize{(Majorana 1932, Penrose Spinors,} \\
    \scriptsize{Hannay \& the 1990s Revival)}
};

\node (ch3) [box, below=of ch2] {
    Sec. III. Mathematical Foundations \\ 
    \scriptsize{(Majorana Polynomials, Bloch Sphere} \\
    \scriptsize{\& Coherent States)}
};

\node (ch4) [box, below=of ch3] {
    Sec. IV. Geometric Tools \\ 
    \scriptsize{(Symmetries, Metrics \&} \\
    \scriptsize{Quasi-probability Distributions)}
};

\node (label_b) [label, right=of label_a] {\uppercase{Core Applications}};

\node (ch5) [box, below=of label_b, yshift=0.3cm] {
    Sec. V. Entanglement \& Classification \\ 
    \scriptsize{(SLOCC Families, Tangles} \\
    \scriptsize{\& Geometric Measures)}
};

\node (ch6) [box, below=of ch5] {
    Sec. VI. Geometric Phases \\ 
    \scriptsize{(Berry Anomalies, Hannay Angles \& Dynamics)}
};

\node (ch7) [box, below=of ch6] {
    Sec. VII. Applications \\ 
    \scriptsize{(Metrology, Error Correction,} \\
    \scriptsize{Many-Body Models \& Photonics)}
};

\node (label_c) [label, right=of label_b] {\uppercase{Reality \& Horizons}};

\node (exp) [mergebox, below=of label_c, yshift=0.3cm] {
    Experimental Implementations \\ 
    \scriptsize{(within Sec. VII: cold atoms, trapped ions,} \\
    \scriptsize{photonics \& state tomography)}
};

\node (ch8) [box, below=of exp] {
    Sec. VIII. Challenges \& Future Directions \\ 
    \scriptsize{(Mixed States, Scalability} \\
    \scriptsize{\& Machine Learning)}
};

\node (ch9) [box, below=of ch8] {
    Sec. IX. Conclusions \\ 
    \scriptsize{(Synthesis \& Outlook)}
};

\draw [arrow] (ch1) -- (ch2);
\draw [arrow] (ch2) -- (ch3);
\draw [arrow] (ch3) -- (ch4);

\coordinate (trans1) at ($(ch4.east) + (0.5cm, 0)$);
\draw [thick, draw=black!70] (ch4.east) -- (trans1);
\draw [arrow] (trans1) |- (ch5.west);

\draw [arrow] (ch5) -- (ch6);
\draw [arrow] (ch6) -- (ch7);

\coordinate (trans2) at ($(ch7.east) + (0.5cm, 0)$);
\draw [thick, draw=black!70] (ch7.east) -- (trans2);
\draw [arrow] (trans2) |- (exp.west);

\draw [arrow] (exp) -- (ch8);
\draw [arrow] (ch8) -- (ch9);

\begin{scope}[on background layer]
    \draw [fill=blue!3!white, draw=none, rounded corners=10pt] 
        ($(label_a.north west)+(-0.3cm, 0.2cm)$) rectangle ($(ch4.south east)+(0.3cm, -0.3cm)$);
    
    \draw [fill=orange!3!white, draw=none, rounded corners=10pt] 
        ($(label_b.north west)+(-0.3cm, 0.2cm)$) rectangle ($(ch7.south east)+(0.3cm, -0.3cm)$);
    
    \draw [fill=green!3!white, draw=none, rounded corners=10pt] 
        ($(label_c.north west)+(-0.3cm, 0.2cm)$) rectangle ($(ch9.south east)+(0.3cm, -0.3cm)$);
\end{scope}

\end{tikzpicture}%
}
\caption{Overall roadmap of this review. The flow follows an S-shaped trajectory from the Foundations (Secs.~I--IV) through the Core Applications (Secs.~V--VII), passing through the experimental material distributed across Sec.~VII, and culminating in the open problems of Sec.~VIII and the conclusions of Sec.~IX.}
\label{fig:roadmap}
\end{figure*}
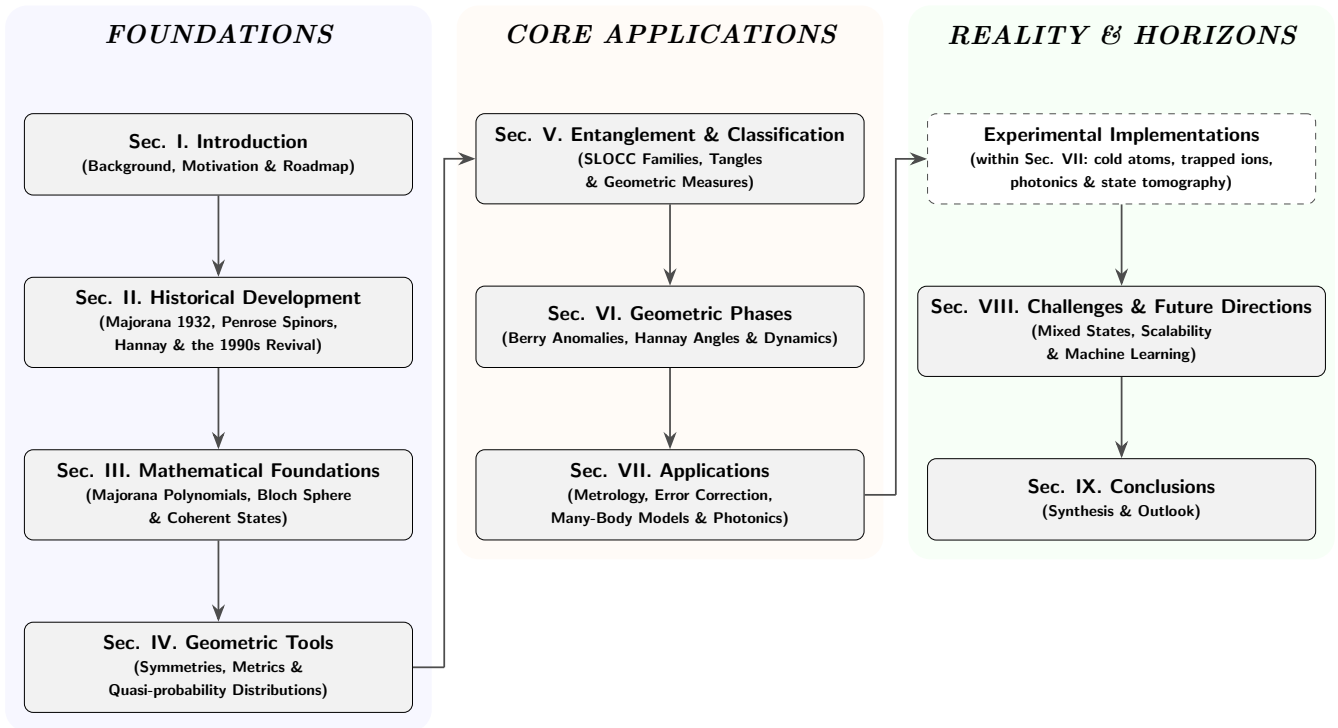

\section{Historical Development and Early Insights}
\label{sec:history}
\subsection{Ettore Majorana's original 1932 proposal}

In 1932, Ettore Majorana published a seminal paper titled ``Atomi orientati in campo magnetico variabile'' (``Oriented atoms in a variable magnetic field'') in \textit{Il Nuovo Cimento} \cite{Majorana1932}. He was exploring extensions of the Stern--Gerlach experiment to more general magnetic fields and higher multipole moments, aiming to generalize the description of atomic magnetic moments beyond the simple dipole case.

The paper's most enduring contribution, however, appears almost in passing: Majorana proposed that any pure quantum state of a system with angular momentum (spin) $S$ can be represented as a symmetrized combination of $2S$ spin-$1/2$ particles (or ``oriented atoms'' in his terminology). He stated that ``ogni stato sar\`a rappresentato da $2S$ punti sulla sfera unitaria'' (``every state will be represented by $2S$ points on the unit sphere''), providing only a brief justification as if the construction were self-evident \cite{Majorana1932}.

Consider a general pure state in the $(2S+1)$-dimensional Hilbert space of a spin-$S$ system,
\begin{equation}
|\psi\rangle = \sum_{m=-S}^{S} c_m |S, m\rangle,
\label{eq:general_spin_state}
\end{equation}
with $\sum |c_m|^2 = 1$. Majorana associated this state with a polynomial of degree $2S$ in a complex variable $z$,
\begin{equation}
P(z) = \sum_{k=0}^{2S} a_k z^k,
\label{eq:Majorana_polynomial}
\end{equation}
where the coefficients $a_k$ are determined from $c_m$. One common modern convention expresses this as
\begin{equation}
P(z) = \sum_{m=-S}^{S} (-1)^{S+m} \sqrt{\binom{2S}{S+m}} c_{m} z^{S+m}.
\label{eq:standard_polynomial}
\end{equation}
The phase factor $(-1)^{S+m}$ and the binomial prefactor are explicitly chosen to ensure covariance under $\text{SU}(2)$ rotations.

By the fundamental theorem of algebra, the polynomial $P(z)$ has exactly $2S$ roots (counting multiplicity) in the extended complex plane, including roots at infinity if the leading coefficients vanish. These roots $z_i$ are mapped to the unit sphere via the inverse stereographic projection,
\begin{equation}
\mathbf{n} = \left( \frac{z + z^*}{|z|^2 + 1}, \frac{z - z^*}{i(|z|^2 + 1)}, \frac{1 - |z|^2}{1 + |z|^2} \right),
\label{eq:stereographic}
\end{equation}
yielding a unique constellation of $2S$ points---the \emph{Majorana constellation} or \emph{Majorana stars}.

Majorana viewed this decomposition as expressing the spin-$S$ state as a symmetric product of $2S$ spin-$1/2$ states, each corresponding to a direction on the sphere. Mathematically, each root $z_i$ defines a constituent spin-$1/2$ state $|\phi_i\rangle \propto |1\rangle + z_i |0\rangle$, with the labelling $|1\rangle \equiv |{\uparrow}\rangle$ fixed by Eq.~\eqref{eq:dicke_to_spin} so that $z_i = 0$ is the north pole, allowing the total spin-$S$ state to be reconstructed as
\begin{equation}
|\psi\rangle \propto \mathcal{S} \left( |\phi_1\rangle \otimes |\phi_2\rangle \otimes \dots \otimes |\phi_{2S}\rangle \right),
\label{eq:symmetric_product}
\end{equation}
where $\mathcal{S}$ is the symmetrization operator. For $S = 1/2$, the representation reduces to a single point. For integer or higher half-integer $S$, the $2S$ points provide a complete geometric characterization of the state, including its symmetries and degeneracies. Although Majorana provided only a qualitative justification in his short note---and the work received little immediate attention due to its publication in Italian and the dominance of the Dirac and Pauli formalisms---this brief proposal successfully laid the geometric foundation for mapping abstract Hilbert-space vectors to spherical geometry.

\subsection{Spinor calculus and Penrose's geometric contextuality}

During the mid-twentieth century, the Majorana representation was little used, while the algebraic tensor methods of Racah \cite{racah1942theory} and Wigner \cite{wigner1940matrices} and the Schwinger boson formalism \cite{schwinger1952angular} became the working tools of the field. The relation between Majorana and Schwinger is more curious than a simple displacement. S\'anchez-Soto \textit{et al.} record that Schwinger encountered Majorana's paper in 1935, at seventeen, and verified the construction himself shortly afterwards \cite{schwinger1937nonadiabatic}; he returned to the subject in 1947, prompted by the review of Bloch and Rabi \cite{bloch1945atoms}. The unpublished 1952 monograph that introduced the Schwinger map nevertheless cites Weyl, G{\"u}ttinger, van der Waerden, Racah and Wigner while omitting the Majorana formula entirely, an omission recorded in Ref.~\cite{sanchezsoto2026quantum}. However, the geometric essence of Majorana's approach was preserved and generalized in the 1960s by Roger Penrose through the development of twistor theory and spinor calculus \cite{penrose1984spinors}. 

Penrose recognized that the algebraic decomposition of a spin-$S$ state is naturally expressed in the language of symmetric spinors. In this formalism, a pure quantum state of spin $S$ corresponds to a totally symmetric spinor of rank $n = 2S$, denoted as $\Psi_{A_1 A_2 \dots A_n}$, where each spinor index $A_i$ takes values in $\{0, 1\}$. By invoking the fundamental theorem of algebra within the spinor framework, Penrose demonstrated that any such totally symmetric rank-$n$ spinor can be uniquely factored into a symmetrized product of $n$ principal rank-$1$ spinors,
\begin{equation}
\Psi_{A_1 A_2 \dots A_n} = \mathcal{S} \left( \alpha_{A_1}^{(1)} \alpha_{A_2}^{(2)} \dots \alpha_{A_n}^{(n)} \right),
\label{eq:spinor_decomposition}
\end{equation}
where $\mathcal{S}$ denotes the total symmetrization over all indices. Each constituent principal spinor $\alpha^{(k)}_A$ has two complex components. Its projective class uniquely determines a direction, mapping it to a single point on the Riemann sphere (referred to by Penrose as the celestial sphere) via the ratio $z_k = \alpha^{(k)}_1 / \alpha^{(k)}_0$. This formulation provided a rigorous mathematical scaffolding for Majorana's original intuition, explicitly equating the $2S$ Majorana stars with the principal null directions of spacetime associated with the quantum state.

While initially developed for relativistic fields, this spinor-geometric perspective later proved instrumental in quantum foundations. In 1993, Zimba and Penrose utilized this framework to formulate a geometric proof of the Kochen-Specker theorem, establishing quantum contextuality without probabilistic inequalities \cite{zimba1993bell}. They focused specifically on spin-$3/2$ systems, where $n=3$, meaning every state is faithfully represented by a triad of Majorana stars on the sphere.

By analyzing the geometric conditions under which two spin-$3/2$ states are orthogonal---which translates to specific antipodal relationships among their respective star constellations---Zimba and Penrose identified finite sets of orthogonal rays that defy any non-contextual hidden-variable assignment. They leveraged the symmetric constellation corresponding to the vertices of a regular dodecahedron. This translation of quantum orthogonality into spherical geometry demonstrated the physical utility of the Majorana constellation prior to the advent of modern quantum information science.

Beyond quantum foundations, the underlying symmetric tensor algebra establishes a direct mathematical isomorphism between the geometric classification of multipartite entanglement and the Petrov classification of gravitational fields \cite{petrov2000classification, penrose1986spinors}. Because the Weyl curvature tensor describing massless spin-$2$ gravitational fields is equivalent to a rank-$4$ totally symmetric spinor, its spacetime geometry is dictated by the degeneracy patterns of four principal null directions. Consequently, the SLOCC families of four-qubit symmetric states mirror relativistic spacetime classifications. For instance, the symmetric Greenberger-Horne-Zeilinger (GHZ) family corresponds to algebraically general Type I spacetimes. Furthermore, the generalized four-qubit W-state (the Dicke state with a single excitation, characterized by a $\mathcal{D}_{3,1}$ constellation) maps precisely to Type III radiation, and states with two pairs of degenerate stars (such as the Dicke state $|D_4^{(2)}\rangle = |S{=}2, m{=}0\rangle$ in the $\mathcal{D}_{2,2}$ class) map strictly to the Type D geometries characteristic of rotating black holes. This exact correspondence highlights the Majorana representation as a unified geometric language connecting quantum correlations with spacetime curvature. A rigorous mathematical derivation of this equivalence is detailed in Appendix~\ref{app:spinor_calculus}.

\subsection{Revival in the 1990s: Quantum chaos and Hannay's geometric phase}

While Penrose's spinor calculus connected the Majorana representation to relativistic spacetime, a parallel revival occurred in the 1990s driven by non-relativistic quantum dynamics, specifically quantum chaos and geometric phases. This resurgence centered on understanding the spatial distribution and temporal trajectories of the Majorana stars.

The geometric characterization of spin states naturally begins with spin coherent states \cite{radcliffe1971some, arecchi1972atomic}. In the Majorana picture, a spin coherent state---representing the most classical quantum configuration, localized along a specific direction $\mathbf{n}$---corresponds to the maximum geometric degeneracy where all $2S$ stars perfectly coincide at a single point. Deviations from this full degeneracy provide a geometric measure of nonclassicality. In the 1990s, this property was leveraged to study quantum chaos. Leboeuf and Voros \cite{leboeuf1990chaos}, and later Hannay \cite{hannay1996chaotic}, demonstrated that for chaotic or random spin ensembles, the stars strongly repel each other, distributing themselves quasi-uniformly over the sphere like the roots of random polynomials. This established the constellation's spatial distribution as a direct visual signature of quantum chaoticity.

Beyond static distributions, Hannay's seminal 1998 work \cite{Hannay1998a} connected the dynamical trajectories of the stars to the Berry phase. For a spin-$S$ state undergoing adiabatic cyclic evolution, Hannay proved that the acquired Berry phase $\gamma$ can be geometrically decomposed as
\begin{equation}
\gamma = -\frac{1}{2} \sum_{k=1}^{2S} \Omega_k + \Delta\gamma,
\label{eq:Hannay_berry_majorana}
\end{equation}
where $\Omega_k$ is the solid angle subtended by the closed trajectory of the $k$-th star on the unit sphere. The first term represents the independent geometric contributions of the constituent spin-$1/2$ states. The second term, $\Delta\gamma$, is an anomalous geometric phase arising from the internal relative motions and pairwise correlations between the stars; it is resolved into pairwise contributions in Eq.~\eqref{eq:berry_decomposition} and the same symbol is used for it throughout. 

For a spin coherent state, where the constellation remains rigidly collapsed to a single point, the anomalous term vanishes, and the formula gracefully recovers the standard solid-angle result for a single macro-spin. However, for generalized states, the anomalous contribution directly encodes the entanglement-like correlations among the underlying symmetrized spin-$1/2$ particles. Hannay subsequently extended this geometric phase framework to optical polarization \cite{Hannay1998b}, illustrating how paraxial light fields acquire geometric phases governed by stellar trajectories.

By the early 2000s, this dual understanding---the static spread of stars indicating nonclassicality and the dynamic relative motions encoding phase anomalies---set the stage for modern quantum information. The states exhibiting maximal stellar repulsion, initially introduced in the context of chaos, were later formalized as ``anticoherent'' states \cite{giraud2015tensor}, serving as optimal resources for rotation sensing. The geometric formalism developed in this decade thus successfully transformed the Majorana representation from a static visual aid into a robust analytical framework for quantum kinematics and dynamics.

\subsection{Transition to quantum information: symmetric multi-qubit states and permutation symmetry}

At the turn of the twenty-first century, the Majorana representation transitioned from a descriptive tool for spin geometry into a framework for quantum information science. This shift was primarily driven by the need to analyze multipartite entanglement within symmetric multi-qubit systems. 

The analytical utility of this transition stems from the mathematical isomorphism between the spin-$S$ Hilbert space and the totally symmetric subspace of $N = 2S$ qubits. Any pure state $|\psi\rangle_{\rm sym}$ invariant under all qubit permutations can be mapped bijectively to a spin-$S$ state,
\begin{align}
&|\psi\rangle_{\rm sym} = \sum_{k=0}^{N} c_k \, |D_N^{(k)}\rangle \nonumber\\
&\leftrightarrow |\psi\rangle_{S=N/2} = \sum_{m=-S}^{S} c_{S+m} |S, m\rangle,
\label{eq:symmetric_equivalence}
\end{align}
where $|D_N^{(k)}\rangle$ denotes the normalized symmetric Dicke state with $k$ excitations. While this equivalence was implicit in Schwinger's angular-momentum formalism \cite{schwinger1952angular}, its application to quantum information provided a geometric language for complex many-body phenomena, such as collective spins and Bose-Einstein condensates \cite{makela2007inert, stamper2013spinor}.

Permutation symmetry structurally simplifies multipartite entanglement classification. For generic $N$-qubit states, classification under stochastic local operations and classical communication (SLOCC) yields an infinite continuum of orbits for $N \geq 4$. Restricting to the symmetric subspace does not remove that continuum, but it organises it: in 2009, Bastin \textit{et al.} \cite{bastin2009operational} proved that the degeneracy configuration of the Majorana stars is a SLOCC invariant, and therefore partitions the symmetric states into a finite set of \emph{families} labelled by the integer partitions of $N$. Within a family of $k$ distinct stars a $(k-3)$-dimensional space of M{\"o}bius moduli survives, so the families coincide with the orbits only for $k \leq 3$. Because local filtering operations correspond to M\"obius transformations on the Riemann sphere, the topological degeneracy of the stars (represented by integer partitions of $N$) remains a SLOCC invariant. A generic constellation without degenerate stars indicates states akin to the generalized GHZ family, whereas fully degenerate constellations correspond to separable spin coherent states.

Beyond qualitative classification, the stellar geometry facilitates the quantitative evaluation of entanglement. Markham and collaborators demonstrated that genuine multipartite entanglement measures can be formulated using geometric invariants derived directly from the spatial distribution of the stars \cite{markham2011entanglement}. Instead of computing high-degree polynomial invariants in the computational basis, entanglement monotones can be extracted from the spherical distances and inner products between the constituent Majorana vectors. 

This geometric approach provided analytical expressions for bipartite concurrence in symmetric two-qubit reductions and the tripartite three-tangle for symmetric three-qubit states, effectively writing entanglement measures as functions of inter-star angular separations. The physical intuition is clear: maximal multiparticle entanglement corresponds to constellations where stars are maximally distributed across the sphere, whereas a loss of entanglement manifests geometrically as stars collapsing toward a common point. 

By the mid-2010s, this permutation-invariant geometric toolkit was well-established. It simplified mathematical proofs in entanglement theory and facilitated direct connections to quantum metrology, where the angular spread of the Majorana constellation dictates the precision bounds for optimal sensing and spin squeezing. This geometric quantification solidified the Majorana representation as a mainstream analytical method in modern quantum information science.

\subsection{Geometric quantification of nonclassicality: Anticoherence and stellar rank}

The Majorana constellation provides a natural geometric metric for intrinsic nonclassicality by contrasting the localized nature of coherent states with the isotropic spread of highly quantum states. In this geometric perspective, classicality is defined by clustering, while quantumness is defined by spatial dispersion.

For a pure spin coherent state $|\theta, \phi\rangle$, all $2S$ Majorana stars are completely degenerate at a single point on the Riemann sphere. Because it is generated by a continuous $\text{SU}(2)$ rotation from the highest weight state, it possesses a universal explicit expansion in the standard angular momentum basis $|S, m\rangle$. This rigid clustering maximizes the collective dipole moment, yielding $|\langle \mathbf{J} \rangle| = S$, which embodies the quantum analogue of a classical macroscopic spin. 

In stark contrast, Zimba introduced the concept of \emph{anticoherent states} to identify quantum states that are as unpolarized as possible \cite{Zimba2006}. A pure state is defined as $t$-anticoherent if the expectation values of its spherical multipole operators vanish identically up to rank $t$, $\langle T_{kq}\rangle = 0$ for $1 \leq k \leq t$; equivalently, $\langle (\mathbf{n}\cdot\mathbf{J})^k\rangle$ is independent of the direction $\mathbf{n}$ for all $k \leq t$. Note that this is a statement about the multipoles, not about $\langle \mathbf{J}^k\rangle$ itself, which for $k=2$ is fixed at $S(S+1)$ for every state. Unlike coherent states, anticoherent states do not form a continuous group orbit. In the standard algebraic framework, finding a general explicit formula for a $t$-anticoherent state becomes algebraically unwieldy, as it requires solving highly nonlinear systems of coupled multipole equations for the basis coefficients $c_m$.

This algebraic bottleneck highlights the distinct geometric advantage of the Majorana representation. Instead of blindly searching for complex coefficients in a vast Hilbert space, the stellar picture reframes the problem entirely: one simply needs to find a constellation of $2S$ points on the sphere that minimizes its low-order geometric moments. The search for highly quantum states is thus reduced to identifying discrete rotational symmetries. 

The simplest manifestation of this geometric intuition is the $1$-anticoherent state, defined by the vanishing of the dipole moment, $\langle \mathbf{J} \rangle = 0$. It is tempting to read this as the statement that the ``center of mass'' of the Majorana stars rests at the origin, and the two conditions do coincide for the highly symmetric configurations considered below; they are not equivalent in general, however, and the reason---the unequal permanental weighting of the stars---is set out following Eq.~\eqref{eq:average_spin}. What remains true without qualification is that a constellation invariant under a discrete symmetry group with no invariant vector has both a vanishing barycenter and a vanishing dipole. For any spin $S \geq 1$ this is realised, for instance, by the spin cat state, whose $2S$ stars are equally spaced around the equator:
\begin{equation}
|\psi_{\text{cat}}\rangle = \frac{1}{\sqrt{2}} \left( |S, S\rangle + |S, -S\rangle \right).
\label{eq:cat_state}
\end{equation}
It is also satisfied, in a different way, by $|S,0\rangle$, whose polynomial $P(z) \propto z^{S}$ places $S$ stars at the north pole and $S$ at the south pole; both configurations have their centroid at the origin, and in both cases the discrete symmetry independently forces $\langle \mathbf{J} \rangle = 0$.

Achieving higher-order anticoherence algebraically is highly non-trivial, yet geometrically it strictly requires higher discrete symmetries. For instance, finding the explicit $t=2$ maximally anticoherent state for a spin-$2$ system corresponds precisely to placing four stars at the vertices of a regular tetrahedron. Once this geometric constellation is identified, the state is uniquely determined. In the $|S, m\rangle$ basis, this extremal state takes the explicit form:
\begin{equation}
|\psi_{\text{tetrahedron}}\rangle = \frac{1}{2} \left( |2, 2\rangle + i\sqrt{2}|2, 0\rangle + |2, -2\rangle \right).
\label{eq:tetrahedron_state}
\end{equation}
Direct evaluation confirms that for this explicit superposition, the quadrupole moments are perfectly isotropic (e.g., $\langle J_z^2 \rangle = \langle J_x^2 \rangle = \langle J_y^2 \rangle = 2$). For higher spins, maximally nonclassical states correspond to larger Platonic solids or spherical designs \cite{romero2024multipoles, giraud2015tensor}, the connection between the two notions being the subject of a dedicated literature \cite{crann2010spherical, bannai2011note}. The Majorana representation provides a robust geometric alternative to algebraic complexity, demonstrating that the absence of multipoles is fundamentally dictated by discrete geometric symmetry.

The mathematical framework of counting complex roots naturally extends beyond finite-dimensional spin systems to infinite-dimensional continuous-variable (CV) bosonic systems. In 2020, Chabaud \textit{et al.} introduced the \emph{stellar representation} to classify non-Gaussianity in single-mode CV states \cite{chabaud2020stellar}. In this infinite-dimensional regime, a pure bosonic state $|\psi\rangle$ is rigorously mapped to an entire analytic function in the Bargmann-Fock space, $F(\alpha^*) = e^{|\alpha|^2/2} \langle \alpha | \psi \rangle$. By Hadamard's factorization theorem, this analytic function is determined by its complex roots together with a non-vanishing exponential factor, which for a Gaussian state carries all of the content. Because the phase-space Husimi $Q$-function is a non-negative real-valued function defined by the modulus square of the Bargmann function, $Q(\alpha) = \frac{1}{\pi} e^{-|\alpha|^2} |F(\alpha^*)|^2$, the roots of $F(\alpha^*)$ uniquely dictate the specific phase-space points where the $Q$-function strictly vanishes. These vanishing points act as the exact CV counterparts to the Majorana stars.

Within this framework, the \emph{stellar rank} $r^\star$ of a CV state is defined as the number of these discrete vanishing points, counted with multiplicity. It should be distinguished from the spin-space quantity $r_\psi$ of Eq.~\eqref{eq:stellar_rank_definition}, which counts \emph{distinct} Majorana roots and therefore ignores multiplicity. For pure Gaussian states---such as the vacuum or squeezed states---the Bargmann functions are purely exponential without any polynomial roots, naturally yielding $r^\star = 0$. Conversely, applying nonclassical operations, such as adding $n$ single photons to a Gaussian state, introduces a polynomial of degree $n$, generating exactly $n$ complex roots and elevating the stellar rank to $r^\star = n$. By reducing the classification of non-Gaussian quantum resources to the discrete counting of these phase-space vanishing points, the stellar formalism establishes a unified geometric language that connects the nonclassicality of finite-dimensional spins with that of continuous bosonic fields.

\section{Mathematical Foundations of Majorana Stellar Representations}
\label{sec:foundations}
\subsection{Mathematical framework: The Majorana polynomial and stereographic projection}

The Majorana stellar representation translates any pure quantum state of a spin-$S$ system (or, equivalently, the totally symmetric subspace of $2S$ qubits) into a unique geometric constellation of $2S$ points on the unit sphere \cite{Majorana1932, Bengtsson2017}. This mapping is mathematically realized through a characteristic polynomial of degree $2S$, whose complex roots uniquely determine the spatial coordinates of these stars.

Consider a general pure spin-$S$ state expanded in the standard angular-momentum basis,
\begin{equation}
|\psi\rangle = \sum_{m=-S}^{S} c_m \, |S, m\rangle,
\label{eq:spin_state_expansion}
\end{equation}
where $c_m \in \mathbb{C}$ and $\sum_{m=-S}^{S} |c_m|^2 = 1$. The state is bijectively mapped to a homogeneous polynomial of degree $2S$ in a complex variable $z$. Adopting the standard convention that ensures proper $\mathrm{SU}(2)$ transformation properties, the Majorana polynomial is defined as \cite{Majorana1932}
\begin{equation}
P_\psi(z) = \sum_{m=-S}^{S} (-1)^{S+m} \sqrt{\binom{2S}{S+m}} \, c_m \, z^{S+m}.
\label{eq:standard_Majorana_polynomial}
\end{equation}

By the fundamental theorem of algebra, $P_\psi(z)$ admits exactly $2S$ roots $\{z_1, z_2, \dots, z_{2S}\}$ in the extended complex plane $\mathbb{C} \cup \{\infty\}$, where a vanishing leading coefficient corresponds to a root at infinity. These roots completely determine the quantum state up to a global phase. The geometric constellation is formalized by mapping each complex root $z_k = \tan(\theta_k/2) e^{i \phi_k}$ onto the unit sphere via inverse stereographic projection \cite{Bengtsson2017}. The spatial coordinate of the $k$-th Majorana star is given by the unit vector
\begin{align}
\mathbf{n}_k &= \left( \sin\theta_k \cos\phi_k, \sin\theta_k \sin\phi_k, \cos\theta_k \right) \nonumber\\
&= \left( \frac{2 \operatorname{Re}(z_k)}{|z_k|^2 + 1}, \frac{2 \operatorname{Im}(z_k)}{|z_k|^2 + 1}, \frac{1 - |z_k|^2}{1 + |z_k|^2} \right).
\label{eq:star_position}
\end{align}

This bijective mapping translates unitary $\mathrm{SU}(2)$ rotations of $|\psi\rangle$ into rigid global rotations of the $2S$ Majorana stars $\{\mathbf{n}_k\}$ on the sphere. Furthermore, the multiplicity of the roots visually encodes the state's classicality: a spin coherent state corresponds to a degenerate constellation where all $2S$ stars coincide at a single point, whereas highly nonclassical states exhibit distributed constellations.

Beyond visual representation, the polynomial carries an algebraic structure that allows for the direct evaluation of physical observables. By factorizing the polynomial into its monic form,
\begin{equation}
\tilde{P}_\psi(z) \propto \prod_{k=1}^{2S} (z - z_k) = \sum_{j=0}^{2S} (-1)^j e_j(z_1, \dots, z_{2S}) z^{2S-j},
\label{eq:monic_expansion}
\end{equation}
the expansion coefficients are rigorously dictated by the elementary symmetric polynomials $e_j$. This algebraic structure geometrically links the constellation to the multipole expansion of the spin density operator \cite{romero2024multipoles}. 

It is tempting to suppose that the geometric center of mass of the Majorana stars is the average spin vector of the state, that is, $\langle \mathbf{S} \rangle = \frac{1}{2} \sum_k \mathbf{n}_k$. This identity is exact for a spin coherent state, where all $2S$ stars coincide, but it fails for every other constellation, and it fails in direction as well as in magnitude. The reason is the same one encountered repeatedly below: the symmetrized product of non-orthogonal coherent states is overcomplete, and its normalization---the permanent of the stellar Gram matrix $G_{kl} = \langle \mathbf{n}_k | \mathbf{n}_l \rangle$---weights the stars unequally. The correct statement follows from the single-qubit reduced state $\rho_1$ of the equivalent symmetric $N = 2S$ qubit description,
\begin{equation}
\langle \mathbf{S} \rangle = \frac{N}{2} \operatorname{Tr} \! \left( \rho_1 \boldsymbol{\sigma} \right),
\label{eq:average_spin}
\end{equation}
with $\boldsymbol{\sigma}$ the vector of Pauli matrices and $\rho_1$ given explicitly in terms of permanental minors of $G$ in Eq.~\eqref{eq:reduced_state_permanental} below. The barycenter $\frac{1}{2S}\sum_k \mathbf{n}_k$ remains a legitimate and useful rotationally covariant geometric descriptor of the constellation; it is simply not the dipole moment of the state, except in the coherent limit.

Higher-order multipole moments and spin variances are similarly governed by higher power-sum symmetric polynomials, which can be computed recursively from the elementary symmetric sums via Newton's identities. Consequently, the algebraic properties of the roots elevate the Majorana representation from a descriptive geometric aid into a rigorous calculational framework. Complex entanglement invariants, geometric phases, and multipole moments can thus be extracted through the spatial separations and symmetric algebraic products of the Majorana roots.

\subsection{Symmetric multi-qubit equivalence and symmetry-induced entanglement classification}

The mathematical machinery of the spin-$S$ Majorana representation gains operational power when mapped to quantum information systems. There exists a canonical unitary isomorphism between a single spin-$S$ particle and the totally symmetric subspace of $N = 2S$ identical qubits. The dimension of this symmetric subspace, $\dim \mathcal{H}_{\rm sym}^{(N)} = N + 1$, perfectly matches the $2S + 1$ dimension of the spin-$S$ Hilbert space. This equivalence is explicitly realized by identifying the symmetric $N$-qubit Dicke states---which are normalized equal superpositions of all permutations containing exactly $S+m$ excitations---with the standard angular-momentum basis:
\begin{equation}
\binom{2S}{S+m}^{-1/2} \sum_{\sigma \in S_N / \mathrm{stab}} \sigma \bigl( |1\rangle^{\otimes (S+m)} \otimes |0\rangle^{\otimes (S-m)} \bigr) \longleftrightarrow |S, m\rangle,
\label{eq:dicke_to_spin}
\end{equation}
where $m = -S, -S+1, \dots, S$ and the sum runs over the $\binom{2S}{S+m}$ distinct arrangements of the string, i.e.\ over cosets of the stabilizer rather than over all of $S_N$; summing over the full symmetric group would overcount each arrangement $(S+m)!(S-m)!$ times (see Appendix~\ref{app:isomorphism}). Consequently, any pure symmetric $N$-qubit state shares the exact same Majorana polynomial and geometric constellation as its spin-$S$ counterpart. 

Under this isomorphism, collective local unitary operations on the qubits, of the form $U^{\otimes N}$ with $U \in \mathrm{SU}(2)$, correspond strictly to a rigid global rotation of the entire Majorana constellation. In the complex plane, such a rotation acts uniformly on all roots through the M\"obius transformation, $z \mapsto (az + b)/(-b^*z + a^*)$, where $a$ and $b$ are the Cayley-Klein parameters defining $U$. Beyond continuous symmetries, discrete quantum symmetries also imprint structural constraints on the constellation. For instance, time-reversal invariance in integer-spin systems requires the state to be symmetric under the anti-unitary operator $\Theta = \exp(-i \pi S_y) K$. In the complex representation, this mandates that the roots must map to themselves under the transformation $z \mapsto -1/z^*$, forcing the Majorana stars to form perfect antipodal pairs on the unit sphere.

From an operational standpoint, this isomorphism enables the classification of multipartite entanglement while bypassing severe computational bottlenecks. For generic multi-qubit systems, the algebraic characterization of entanglement is fundamentally intractable due to two distinct algorithmic walls. On one hand, decision problems---most directly, deciding whether a given bipartite density matrix is separable---are NP-hard, both in Gurvits's original formulation and in its strengthening to inverse-polynomial distance from the separable set \cite{gurvits2003classical, gharibian2010strong}. On the other hand, the exact quantitative evaluation of entanglement, which requires computing invariant polynomials or exactly contracting multi-body tensor networks, represents a counting problem that falls into the even more demanding \#P-hard complexity class \cite{schuch2007computational}. The Majorana representation circumvents these dual barriers for permutation-symmetric states by replacing the algebraic complexity with a direct geometric classification.

In quantum information theory, two $N$-qubit states share the same entanglement structure if they can be transformed into each other via Stochastic Local Operations and Classical Communication (SLOCC). For symmetric states, such local filtering operations correspond to collective operations driven by the invertible local group $(\mathrm{SL}(2, \mathbb{C}))^{\otimes N}$. The action of an arbitrary non-unitary element on the Majorana polynomial extends the rigid spherical rotations into general M\"obius transformations of the complex plane. A fundamental algebraic property of these transformations is that, while they distort the relative distances between roots, they strictly preserve the algebraic multiplicities of the polynomial roots.

Consequently, the SLOCC equivalence classes are uniquely dictated by the algebraic degeneracy configuration of the Majorana roots \cite{bastin2009operational}. This transforms the classification of multipartite entanglement into a direct geometric assessment of coincident stars. When all $N$ stars coincide at a single point, the state possesses maximal degeneracy and reduces to a fully separable multi-qubit coherent state. Conversely, states with no degeneracies, featuring $N$ distinct stars, represent generic and genuinely multipartite entangled configurations. Between these two extremes lie the restricted SLOCC families, such as the generalized $W$-states. A root configuration exhibiting partial degeneracies (with multiplicities $m \ge 2$) indicates a reduction in the effective number of continuous parameters characterizing the state, mathematically defining a distinct orbit of symmetric entanglement. By merging the algebraic properties of M\"obius transformations with discrete root multiplicities, the Majorana representation visualizes the invariant structure of multipartite entanglement without relying on computationally intractable polynomial invariants, such as the hyperdeterminant \cite{miyake2003classification}.

This geometric approach also fundamentally alters the landscape of computational complexity for symmetric states. While identifying the Majorana stars for an $N$-qubit symmetric state is equivalent to root-finding---a task reducible to an eigenvalue problem solvable in $\mathcal{O}(N^3)$ time---the exact quantitative evaluation of many-body overlaps or invariants from these stellar coordinates is governed by the matrix permanent. Although computing the permanent of a general matrix is \#P-hard \cite{valiant1979complexity}, the Gram matrices derived from the Majorana stellar representation inherently possess a low rank of at most 2, since they are constructed from the inner products of 2-dimensional single-qubit state vectors. By leveraging Barvinok's algorithm \cite{barvinok1996} for fixed-rank matrices, these permanents can be efficiently evaluated in polynomial time (specifically $\mathcal{O}(N^3)$), effectively ``dequantizing'' the complexity. The Majorana representation thus not only visually classifies entanglement but also provides a fully tractable computational framework for exact algebraic evaluations (see Appendix~\ref{app:complexity} for the algebraic derivation).

\subsection{Generalizations: Breaking Purity, Symmetry, and Local Dimensionality}
\label{subsec:generalizations}

The elegance of the standard Majorana stellar representation lies in its strict mathematical bijection: every pure, permutation-symmetric state of $N$ qubits corresponds uniquely to $N$ discrete points on the Riemann sphere, defined by the roots of a single univariate polynomial. However, this exact algebraic correspondence relies heavily on the idealizations of absolute purity, global permutation symmetry, and two-level local systems. When extending the framework to encompass broader quantum phenomena---such as the decoherence of mixed states, the unstructured correlations of non-symmetric ensembles, or the expanded local Hilbert spaces of higher-dimensional qudits---the standard univariate polynomial description inevitably breaks down. Generalizing the stellar picture to these regimes therefore requires moving beyond simple root-finding on $S^2$. Instead, the mathematical formalism must be elevated to accommodate quasi-probability distributions, effective subspace projections, and continuous algebraic varieties in higher dimensions. This structural elevation allows the framework to handle exponentially larger or classical-mixed state spaces while successfully preserving the core geometric intuition of entanglement and state overlaps.

When a symmetric system undergoes decoherence, the loss of phase information between basis components prevents the state from being described by a single holomorphic polynomial. For a mixed state characterized by a density operator $\rho$, the exact stellar bijection is often explored through two complementary geometric generalizations. On a continuous level, the state can be mapped to a quasi-probability distribution via the Husimi Q-function $Q(\mathbf{n}) = \langle \mathbf{n} | \rho | \mathbf{n} \rangle$, where $|\mathbf{n}\rangle$ denotes the $SU(2)$ coherent state parameterized by the unit vector $\mathbf{n}$ on the sphere \cite{Bengtsson2017}. Because $Q(\mathbf{n})$ for a mixed state is a convex sum of strictly positive terms, the exact zeros characteristic of pure states wash out into local minima, creating a continuous topography of quantum fluctuations.

On a discrete level, to rigorously preserve the notion of individual stars, the density matrix must be expanded into its Cartesian multipole moments. By applying Sylvester's theorem to the highest-rank symmetric traceless tensor of $\rho$, the state is factorized into a set of effective real vectors \cite{serrano2020majorana}. As the state mixes, the norms of these factorized vectors decrease, meaning these effective ``stars'' are no longer constrained to the spherical surface but physically migrate into the interior of the Bloch ball. The radial distance of these interior stars from the origin becomes a geometric proxy for local purity, allowing one to visualize open-system dynamics, such as Lindblad evolution, as smooth continuous trajectories of a collapsing constellation.

The geometric clarity becomes significantly more intricate when the assumption of permutation invariance is lifted. For a general non-symmetric multi-qubit state, the Hilbert space dimension explodes from the linearly growing symmetric sector of size $N+1$ to the full tensor product space of $2^N$. Consequently, the state can no longer be faithfully compressed onto a single two-dimensional sphere. While the global bijective mapping is lost, the stellar geometry remains a powerful diagnostic tool. By projecting a general state onto the symmetric subspace $\rho_{\rm sym} = \mathcal{P}_{\rm sym} \rho \mathcal{P}_{\rm sym}$, one can construct an effective constellation that captures the permutation-averaged properties of the system \cite{aulbach2010maximally}. This projection serves as a geometric witness for genuine multipartite entanglement, identifying how much macroscopic coherence survives the breakdown of global symmetry. To fully capture the residual unstructured correlations, the geometric framework must be expanded into multiple interacting constellations or mapped onto high-dimensional entanglement polytopes. This geometric fragmentation fundamentally reflects the computational transition from the tractable symmetric sector---where low-rank permanent algorithms render exact evaluations efficient---to the exponentially complex, \#P-hard domain of generic Hilbert spaces.

Beyond relaxing the constraints of purity and symmetry, the most profound generalization of the stellar framework involves elevating the local degrees of freedom from qubits to arbitrary $d$-dimensional qudits. As the system transitions from $\mathrm{SU}(2)$ to $\mathrm{SU}(d)$, the underlying geometric space for a pure local state expands from the familiar Riemann sphere $S^2$ to the complex projective space $\mathbb{CP}^{d-1}$ \cite{sanchezsoto2026quantum}. Consequently, the traditional univariate Majorana polynomial is replaced by a single multivariate homogeneous polynomial. Rather than yielding discrete point-like stars, the zero locus of this generalized polynomial traces out continuous algebraic varieties---such as curves or hypersurfaces---embedded within $\mathbb{CP}^{d-1}$. This continuous geometric lift allows high-dimensional entanglement classifications and nonclassicality measures, such as the stellar rank \cite{chabaud2020stellar}, to be rigorously formulated using algebraic geometry. Ultimately, this algebro-geometric perspective provides a powerful analytical tool for qudit-based quantum computing and high-dimensional metrology. By characterizing complex multi-qudit entanglement through geometric invariants rather than full state vector reconstruction, and by leveraging the polynomial-time tractability of fixed-rank matrix permanents within the symmetric subspace, this framework sidesteps, for the quantities expressible through those permanents, the \#P-hard barriers that otherwise limit the analysis of generic asymmetric many-body systems. Quantities requiring optimisation over the sphere, such as the geometric measure, are not covered by this argument.

\section{Geometric Tools and Quantities Derived from Constellations}
\label{sec:tools}

\subsection{Inter-star distances, dot products, and pair correlations}

The discrete nature of the Majorana constellation makes geometric quantities---such as inter-star distances, dot products between star position vectors, and pairwise correlation functions---particularly natural and powerful observables. Each Majorana star is represented by a unit vector $\mathbf{n}_k \in \mathbb{R}^3$ on the Bloch sphere, obtained from the corresponding polynomial root $z_k$ via the stereographic projection (Eq.~\eqref{eq:star_position}) \cite{Majorana1932, Bengtsson2017}. The most fundamental relative geometric quantity is the dot product between two stars, $\mathbf{n}_k \cdot \mathbf{n}_l = \cos \vartheta_{kl}$, where $\vartheta_{kl}$ is their great-circle angular separation. Expressed explicitly in terms of the complex stereographic coordinates, this relation takes the form
\begin{equation}
\mathbf{n}_k \cdot \mathbf{n}_l = 1 - \frac{ |z_k - z_l|^2 }{ (1 + |z_k|^2)(1 + |z_l|^2) }.
\label{eq:dot_product_z}
\end{equation}
This quantity is central to many entanglement-related expressions because it relies solely on the relative positions of the stars, ensuring strict invariance under global $\mathrm{SU}(2)$ rotations \cite{ribeiro2011entanglement}.

Building directly upon the dot product, one can define the Euclidean, or chordal, distance between two stars in the embedding space $\mathbb{R}^3$. While the geodesic distance on the sphere is simply the angle $\vartheta_{kl}$, the chordal distance, given by
\begin{align}
d_{kl} &= \| \mathbf{n}_k - \mathbf{n}_l \| = \sqrt{ 2 (1 - \mathbf{n}_k \cdot \mathbf{n}_l) } \nonumber\\
&= \sqrt{ \frac{ 2 |z_k - z_l|^2 }{ (1 + |z_k|^2)(1 + |z_l|^2) } },
\label{eq:chordal_distance}
\end{align}
is frequently more convenient for algebraic manipulations and the construction of entanglement invariants \cite{markham2011entanglement}. It is often useful to employ a normalized version of this distance, defined as $\tilde{d}_{kl} = d_{kl}/2 = \sin(\vartheta_{kl}/2)$. This normalized metric is strictly bounded within the interval $[0,1]$, vanishing entirely when two stars coincide ($\vartheta_{kl}=0$) and reaching its maximum when they are diametrically opposed, or antipodal ($\vartheta_{kl}=\pi$).

Beyond individual pairs, the collective geometric structure of the constellation is captured by statistical averages and higher-order correlation functions. A fundamental global metric is the average pairwise dot product (excluding self-terms), which is fixed entirely by the length of the constellation barycenter:
\begin{equation}
C = \frac{1}{2S(2S-1)} \sum_{k \neq l} \mathbf{n}_k \cdot \mathbf{n}_l = \frac{ \left| \sum_k \mathbf{n}_k \right|^2 - 2S }{ 2S(2S-1) }.
\label{eq:average_pairwise_dot}
\end{equation}
This is an identity of Euclidean geometry, obtained by expanding $\left( \sum_k \mathbf{n}_k \right)^2$ and applying $\mathbf{n}_k^2 = 1$; it involves no quantum input, and in particular it should not be rewritten in terms of $\langle \mathbf{S} \rangle$, since the barycenter is not the dipole moment (cf.\ the discussion following Eq.~\eqref{eq:average_spin}). The average pairwise correlation establishes strict bounds on the state's geometry: for fully coherent states where all stars align identically, $\left| \sum_k \mathbf{n}_k \right| = 2S$ and $C = 1$, whereas for any constellation with a vanishing barycenter the value is exactly $C = -1/(2S-1)$, for every $S$ and not merely asymptotically \cite{Zimba2006, giraud2015tensor}. This is a statement about the point set alone. It should not be conflated with $1$-anticoherence, which is the vanishing of $\langle \mathbf{J} \rangle$: the two conditions coincide for the symmetric configurations usually quoted as examples, but neither implies the other in general, again because the barycenter is not the dipole moment. This statistical perspective can be readily extended to higher-order geometrical structures, such as the three-point function $T_{klm} = \mathbf{n}_k \cdot (\mathbf{n}_l \times \mathbf{n}_m)$. Such scalar triple products quantify the chiral volume spanned by triplets of stars and naturally emerge in the evaluation of internal-geometry contributions to dynamical phases.

Ultimately, these pairwise and multi-star geometric observables serve as the critical bridge connecting the discrete constellation to continuous algebraic measures of quantum correlations. In the context of entanglement, bipartite concurrence in symmetric reductions is governed by the average or minimal inter-star distances within the reduced constellation \cite{aulbach2010maximally}. Moving to tripartite systems, the genuine entanglement of symmetric three-qubit states---quantified algebraically by the Coffman-Kundu-Wootters three-tangle---can be expressed elegantly as a symmetric product of the normalized inter-star distances, $\prod_{k<l} \tilde{d}_{kl}^2$ \cite{markham2011entanglement}. Beyond static correlations, these geometries dictate the magnitude of dynamical phase anomalies. When a non-coherent state undergoes adiabatic cyclic evolution, the resulting Berry phase acquires anomalous pairwise corrections $\beta_{kl}$ that depend explicitly on the relative motion and dot-product history of the stars \cite{Hannay1998a}. Furthermore, these metrics act as robust nonclassicality witnesses, where minimal inter-star dot products directly certify maximal anticoherence. By translating abstract algebraic invariants into rotationally invariant spatial relations, these geometric tools provide an intuitive framework that is utilized throughout the remainder of this review.

\subsection{Multipole moments and expansion in spherical harmonics}

While the Majorana constellation provides a discrete geometric picture of a spin-$S$ state, this point-like framework is intertwined with the continuous, rotationally covariant description offered by multipole expansions \cite{Bengtsson2017}. In the standard algebraic formulation, a quantum state is characterized by the expectation values of irreducible tensor operators $T^{(l)}_m$ (where $l = 0,1,\dots,2S$ and $m = -l,\dots,l$). These operators are constructed from the spin angular momentum generators $\mathbf{S} = (S_x, S_y, S_z)$ as normalized spherical harmonics \cite{romero2024multipoles}:
\begin{equation}
T^{(l)}_m = \sqrt{\frac{4\pi}{2l+1}} \, Y_{lm}(\hat{\mathbf{S}} / S).
\label{eq:multipole_tensor}
\end{equation}
The corresponding multipole moments are then defined via the trace with the density matrix $\rho$, capturing the angular anisotropies of the state:
\begin{equation}
\langle T^{(l)}_m \rangle = \operatorname{Tr} \bigl[ \rho \, T^{(l)}_m \bigr] = \sqrt{\frac{4\pi}{2l+1}} \, \langle Y_{lm} \rangle.
\label{eq:multipole_moment}
\end{equation}
In this hierarchy, the $l=0$ term enforces trace normalization, the $l=1$ moments define the classical-like Bloch vector (the dipole moment $\langle \hat{\mathbf{n}} \rangle$), and the higher-order terms ($l \ge 2$) encode quantum properties such as quadrupole and octupole deformations \cite{giraud2015tensor}.

The bridge between these continuous multipole moments and the discrete Majorana stars relies on the state's quasi-probability distributions. Because any pure state admits a continuous expansion over the overcomplete basis of coherent states \cite{radcliffe1971some, arecchi1972atomic, kam2023coherent}, its phase-space topography is governed by the Husimi $Q$-function \cite{husimi1940some}. For a pure spin-$S$ state, this distribution is proportional to the modulus square of the Majorana polynomial, evaluated at the point antipodal to $\mathbf{n}$ \cite{leboeuf1990chaos, Bengtsson2017}:
\begin{equation}
Q(\mathbf{n}) = |\langle \mathbf{n} | \psi \rangle|^2 \propto \frac{ \left| P_\psi(\bar{z}) \right|^2 }{ (1 + |\bar{z}|^2)^{2S} }, \qquad \bar{z} = -\frac{1}{z^*},
\label{eq:Q_pure}
\end{equation}
where $z$ is the stereographic coordinate of $\mathbf{n}$ and $\bar{z}$ that of $-\mathbf{n}$. The antipodal map is essential and is a frequent source of confusion: a coherent state $|\mathbf{n}\rangle$ is orthogonal to $|\psi\rangle$ precisely when $-\mathbf{n}$ is a Majorana star, so the zeros of $Q$ sit at the \emph{antipodes} of the constellation, while the stars themselves lie in the high-$Q$ region.
The $Q$-function provides the bridge between the two descriptions, since the multipole moments are recovered as its projections onto the spherical harmonics,
\begin{equation}
\langle T^{(l)}_m \rangle = \mathcal{C}_l \int_{S^2} d\Omega \; Q(\mathbf{n}) \, Y_{lm}(\mathbf{n}),
\label{eq:multipole_from_Q}
\end{equation}
with $\mathcal{C}_l$ a normalisation constant \cite{romero2024multipoles, sanchezsoto2026quantum}. The moments are therefore weighted by the whole of $Q$, and it is worth resisting the tempting shortcut of replacing this integral by a uniform average of $Y_{lm}$ over the star positions. That replacement is valid only when the constellation is close to fully degenerate, in which case $Q$ concentrates on the single repeated star and the integral tends to $Y_{lm}$ evaluated there. It fails in general, and it fails at every $S$: for $|S,0\rangle$, whose $S$ stars at the north pole and $S$ at the south pole leave $Q$ concentrated along the equator, the rank-two moment $3\langle S_z^2\rangle - S(S+1) = -S(S+1)$ is negative while the star average of $Y_{20}$ is positive. The correct finite-$S$ correspondence runs instead through the coefficients of $P_\psi$: expanding the power sums of the stereographic roots $p_r = \sum_k z_k^r$ and applying Newton's identities recovers the state's spherical-harmonic expansion with the requisite binomial and phase factors \cite{romero2024multipoles}.

Through this algebraic correspondence, the constellation's symmetries and degeneracies strictly determine which multipoles vanish or maximize. For instance, any spatial inversion symmetry within the constellation---manifested as the antipodal pairing of stars---forces all odd-rank (vector-like) multipoles to vanish identically \cite{giraud2015tensor, romero2024multipoles}. Conversely, the even-rank multipoles remain sensitive to the spatial spread and structural anisotropies of the star distribution. This interplay is illustrated by maximally anticoherent states. By definition, a $t$-anticoherent state demands the simultaneous vanishing of all multipole moments up to order $t$ (where $t$ is the maximum allowed for a given $S$). In the Majorana picture, satisfying these highly nonlinear algebraic constraints is geometrically equivalent to arranging the stars into highly symmetric configurations, such as Platonic solids or spherical designs, which naturally eliminate low-order geometric moments \cite{Zimba2006, giraud2015tensor}. Furthermore, stellar degeneracies directly reduce the effective dimensionality of the state space; a constellation featuring only $d$ distinct star positions possesses at most $d$ free rotational parameters, imposing stringent interdependencies on its higher-order multipoles \cite{romero2024multipoles}.

Ultimately, the synthesis of continuous multipole analysis and discrete Majorana geometry provides a landscape for quantum information processing. Within this framework, the vanishing of low-rank multipoles witnesses specific entanglement classes; for instance, the generalized GHZ state is characterized by the suppression of odd multipoles beyond the dipole \cite{markham2011entanglement}. This geometric intuition extends to quantum metrology, where the optimization of sensitivity involves shaping the state to maximize higher-order multipoles while silencing low-order noise---a duality manifested as spin squeezing or the alignment of Majorana stars \cite{toth2014quantum, goldberg2020extremal}. As the system evolves, the Berry connection and curvature trace the gradients of these multipoles along adiabatic paths, where inter-star correlations drive geometric phase corrections \cite{Hannay1998a}. By unifying discrete stellar classification with continuous spin distributions, the framework bridges the algebraic and geometric paradigms of quantum physics.

\subsection{Stellar rank and measures of quantumness/nonclassicality}

Beyond multipartite correlations, the Majorana constellation characterizes intrinsic nonclassicality independent of entanglement structure. A foundational metric in this regime is the stellar rank $r_\psi$ of a pure spin-$S$ state $|\psi\rangle$, defined as the number of distinct roots of its Majorana polynomial,
\begin{equation}
r_\psi = \# \left\{ \text{distinct roots of } P_\psi(z) \right\},
\label{eq:stellar_rank_definition}
\end{equation}
multiplicities being counted only once. The name follows the analogous construction for continuous-variable systems, where the stellar rank of a state is the number of zeros of its stellar function \cite{chabaud2020stellar}. This invariant provides a rotationally invariant hierarchy of nonclassicality \cite{sanchezsoto2026quantum}. At the classical extreme, $r_\psi = 1$ strictly identifies a coherent state where all $2S$ stars perfectly coincide. Conversely, generic quantum states possess fully distinct roots, so that $r_\psi = 2S$. States exhibiting intermediate ranks ($r_\psi \ll 2S$) carry fewer independent stellar positions and are correspondingly constrained in their higher multipoles. For mixed systems, an analogous notion is furnished by the minimal number of coherent states appearing in a convex decomposition of $\rho$ \cite{serrano2020majorana}.

A caution is in order, because the literature contains a superficially similar quantity that must not be conflated with Eq.~\eqref{eq:stellar_rank_definition}: the minimal number of spin coherent states whose \emph{linear superposition} reproduces $|\psi\rangle$, sometimes called the coherent rank. The two are not equal. Already for $S=3/2$, a generic constellation of three distinct stars---stellar rank three---can be written exactly as a superposition of only two coherent states, since the parameter count of two directions and two complex amplitudes exceeds that of the four-dimensional state space. Only the $r_\psi = 1$ case is common to both notions. It is Eq.~\eqref{eq:stellar_rank_definition}, not the coherent rank, that is invariant under $\mathrm{SU}(2)$ rotations, that is preserved by the M{\"o}bius action of SLOCC filtering, and that is used throughout this review.

While stellar rank counts the number of distinct roots, the spatial arrangement of these roots dictates a complementary measure of nonclassicality known as anticoherence \cite{Zimba2006}. A state is rigorously defined as $t$-anticoherent if all of its multipole moments up to rank $t$ vanish identically:
\begin{equation}
\langle T^{(l)}_m \rangle = 0 \quad \forall \, l=1,\dots,t, \quad |m|\leq l.
\label{eq:anticoherent_definition}
\end{equation}
Geometrically, this condition is met by arranging the stars so that the constellation is isotropic up to order $t$; in practice this is achieved through discrete rotational symmetry, since a symmetry group admitting no invariant tensor of rank $\leq t$ forces the corresponding multipoles of the state to vanish. We stress that the condition is on the multipoles of $\rho$ and not on the naive moments of the point set, for the reason given after Eq.~\eqref{eq:multipole_from_Q}. States that achieve the theoretically maximal anticoherence order, $t_{\rm max}(S)$, are termed ``kings of quantumness,'' representing the absolute geometric antithesis of localized coherent states \cite{bjork2015extremal, giraud2015tensor}. The geometric realization of these extremal states corresponds to highly symmetric spatial distributions. For instance, the $S=1$ ($N=2$) case is exclusively solved by two antipodal stars, $S=3/2$ ($N=3$) by an equilateral triangle mapped onto a great circle, and $S=2$ ($N=4$) by the vertices of a regular tetrahedron \cite{bjork2015extremal}. For higher spin systems, these configurations map to regular polyhedra, Platonic solids, or geometric arrangements minimizing a repulsive potential, directly tying the quantum state generation to solutions of the classical Thomson problem \cite{sanchezsoto2026quantum}. Such symmetric distributions maximize anticoherence while representing pure states within the spin-$S$ manifold.

The concepts of rank and anticoherence are further supplemented by continuous geometric quantifiers derived from the constellation. The average pairwise correlation $C$ (Eq.~\eqref{eq:average_pairwise_dot}) establishes a bounded scale of nonclassicality, ranging from $+1$ for perfectly coherent states down to $-1/(2S-1)$ for maximally anticoherent distributions \cite{Zimba2006}. Other rotationally invariant metrics include the spread functional, which evaluates the variance of inter-star distances to capture the second moment of the discrete spherical distribution \cite{aulbach2010maximally}, and the stellar entropy, computed as the Shannon entropy of the normalized inter-star distances to quantify spatial uniformity \cite{ganczarek2012barycentric}. Additionally, one can define the effective number of stars as $2S / \langle m \rangle$, where $\langle m \rangle$ is the average root multiplicity, providing a continuous diagnostic complement to the discrete stellar rank. Because all these quantities are directly computable from the stereographic roots $\{z_k\}$, they avoid the computational overhead of full state reconstruction.

Operationally, because these geometric measures are inherently independent of subsystem partitioning, they serve as ideal benchmarks for quantum resources in diverse physical contexts. In quantum metrology, the kings of quantumness provide optimal resource states for rotation sensing, achieving Heisenberg-limited scaling while completely suppressing classical low-order noise \cite{goldberg2020extremal, toth2014quantum}. In the context of quantum chaos and random matrix theory, generic states in the large-$S$ limit statistically approach maximal anticoherence, making the uniform spatial spread of the constellation a visual hallmark of ergodic dynamics \cite{leboeuf1990chaos, hannay1996chaotic}. Furthermore, these measures guide experimental state engineering by allowing the design of highly nonclassical yet unentangled resource states for testing quantum simulators, where the stellar rank can be certified via partial constellation reconstruction or low-order multipole tomography. Ultimately, while standard entanglement measures focus on correlations across bipartite or multipartite cuts, stellar rank and anticoherence quantify the departure of a single collective system from classicality, completing the geometric toolkit afforded by the Majorana representation.

\subsection{Random Majorana constellations: statistical properties in large-$S$ limit}

In the thermodynamic limit ($S \to \infty$, or equivalently, for a macroscopic number of symmetric qubits $N \to \infty$), the Majorana constellation transitions from a discrete geometric tool into a continuous statistical ensemble. Drawing a typical pure state uniformly from the Haar measure on the complex projective space $\mathbb{CP}^{2S}$ translates, within the Majorana framework, to generating a random $\mathrm{SU}(2)$ polynomial. Because of the unitary invariance of the Haar measure, the coefficients of this polynomial are independent complex Gaussian random variables weighted by specific binomial factors. When the roots of these random polynomials are mapped onto the Bloch sphere via the inverse stereographic projection, their spatial distribution rigorously converges to a universal uniform density \cite{leboeuf1990chaos, hannay1996chaotic}. The probability density of finding a star at a specific direction $\mathbf{n}$ approaches
\begin{equation}
\rho(\mathbf{n}) \to \frac{1}{4\pi} \quad \text{(uniform on } \mathcal{S}^2\text{ as } S \to \infty\text{)}.
\label{eq:uniform_star_density}
\end{equation}
This uniform spherical covering represents the large-$S$ geometric analogue of the maximally mixed state, providing a featureless background of isotropic quantum fluctuations.

Crucially, however, these random constellations are not mere collections of independent, non-interacting points. A profound result from the study of quantum chaos demonstrates that the roots of random $\mathrm{SU}(2)$ polynomials are statistically equivalent to the eigenvalues of non-Hermitian random matrices, specifically aligning with the complex Ginibre ensemble after an appropriate stereographic mapping \cite{bogomolny1996quantum}. Consequently, the Majorana stars absolutely do not follow a memoryless Poisson point process. Instead, they exhibit strong spatial correlations characterized by universal Random Matrix Theory (RMT) signatures. The pair correlation function of the roots reveals pronounced level repulsion at short distances, scaling quadratically in strict accordance with the Wigner surmise \cite{bogomolny1996quantum, leboeuf1990chaos}. Physically, this implies that the Majorana stars behave as a two-dimensional Coulomb gas of mutually repelling classical charges confined to the spherical surface \cite{bogomolny1996quantum}. This inherent repulsion constitutes a direct geometric signature of quantum chaos, confirming that generic, highly entangled states inherently resist the clustering of their constituent roots.

This Coulomb-like level repulsion completely dictates the asymptotic scaling of the constellation's geometric invariants \cite{markham2011entanglement}. For a Haar-random state, the average pairwise dot product approaches, to leading order in $1/S$, the value fixed by a vanishing barycenter,
\begin{equation}
\langle \mathbf{n}_k \cdot \mathbf{n}_l \rangle_{k \neq l} \to -\frac{1}{2S-1} \approx -\frac{1}{2S} \quad (S \gg 1).
\label{eq:random_pair_correlation}
\end{equation}
The limit is approached from above and is not attained at finite $S$; direct sampling of Haar-random states gives $-0.132$ against $-1/7$ at $2S = 8$ and $-0.0317$ against $-1/31$ at $2S = 32$. It should also be emphasised that Eq.~\eqref{eq:random_pair_correlation} is not the uncorrelated value, which would be zero: the negative sign is precisely the signature of the repulsion just described.
The variance of these pairwise projections scales precisely as $1/S$, reflecting typical statistical fluctuations in the thermodynamic limit \cite{ganczarek2012barycentric}. Similarly, because rotational invariance guarantees that the ensemble average of all higher-order multipole moments ($\langle T^{(l)}_m \rangle$ for $l \ge 1$) strictly vanishes, individual generic states will exhibit small, random multipole anisotropies bounded tightly by $\mathcal{O}(1/\sqrt{S})$ fluctuations. The distribution of normalized inter-star chordal distances $\tilde{d}_{kl}$ approaches the classical geometric distribution of chordal distances on a sphere, but with a highly suppressed short-distance tail dictated by the aforementioned Coulomb repulsion.

The statistical distribution of Majorana roots dictates the entanglement landscape of macroscopic symmetric systems. Effective repulsion between stars exponentially suppresses the probability of coincident roots, which scales as $e^{-cN}$ for a constant $c$ \cite{markham2011entanglement}. Consequently, nearly all states in the large-$N$ limit occupy the generic SLOCC orbit, characterized by maximal stellar rank $r_\psi = 2S$ and the absence of local filtering reductions. Within this thermodynamic regime, typicality ensures that random constellations exhibit near-maximal multipartite entanglement \cite{ganczarek2012barycentric}. This statistical preference is confirmed by geometric invariants, such as the inter-star distance product, which concentrate at typical values to signify the prevalence of dispersed, highly entangled configurations.

Random Majorana constellations establish the null model for analyzing many-body dynamics and quantum phase transitions. Any departure from this uniform Coulomb-gas distribution---whether through stellar clustering, polyhedral order, or antipodal pairing---acts as a witness for broken ergodicity. In systems such as the Lipkin-Meshkov-Glick (LMG) model, excited-state quantum phase transitions manifest as macroscopic redistributions of the stars, breaking the expected Haar-random uniformity \cite{ribeiro2007thermodynamical, ribeiro2008exact}. By benchmarking experimental constellations against these universal RMT properties, researchers can identify signatures of quantum integrability and locate the boundaries where macroscopic systems fail to thermalize.

\subsection{Interplay with Spherical Phase-Space and Analytic Representations}

The Majorana stellar representation does not exist in a vacuum; it is a specialized, discrete member of a broader family of phase-space representations developed for spin and angular momentum systems. To appreciate its geometric utility---particularly for entanglement classification---it is instructive to contrast it with continuous quasi-probability distributions and analytic phase-space methods, specifically the Husimi $Q$-function, the Stratonovich-Wigner function, and the Bargmann representation. While continuous distributions capture macroscopic features like decoherence, the Majorana representation reduces the state to a discrete geometry to reveal its underlying multipartite correlations.

The most direct continuous counterpart to the stellar picture is the Husimi $Q$-function. Defined as the expectation value of the density matrix with respect to the overcomplete basis of spin coherent states \cite{radcliffe1971some}, 
\begin{equation}
Q(\mathbf{n}) = \langle \mathbf{n} | \rho | \mathbf{n} \rangle,
\label{eq:Husimi_Q}
\end{equation}
the Husimi distribution provides a strictly non-negative, smoothed probability landscape over the Bloch sphere \cite{husimi1940some}. For a pure spin-$S$ state, this distribution is fundamentally dictated by the Majorana polynomial; specifically, $Q(\mathbf{n})$ is rigorously proportional to the modulus square of the polynomial,
\begin{equation}
Q(\mathbf{n}) \propto \frac{ \left| P_\psi(\bar{z}) \right|^2 }{ (1 + |\bar{z}|^2)^{2S} },
\label{eq:Q_via_Majorana}
\end{equation}
where $z$ is the stereographic projection of the direction $\mathbf{n}$ and $\bar{z} = -1/z^*$ that of $-\mathbf{n}$, as in Eq.~\eqref{eq:Q_pure}. Consequently, the Majorana stars are not arbitrary points---the zeros of the Husimi $Q$-function sit exactly at their antipodes \cite{leboeuf1990chaos, Bengtsson2017}. While both representations live on the sphere, their operational philosophies diverge. The $Q$-function smears quantum features over the scale of a coherent state due to its inherent $2S+1$ resolution limit, making it ideal for visualizing the classical-quantum boundary and mapping mixed states or open-system noise. However, this continuous smoothing obscures discrete entanglement structures. Where $Q$-function analysis requires complex integrations over the spherical manifold to extract marginal probabilities or entanglement witnesses, the Majorana picture preserves exact state information as point-like symmetries, allowing one to classify Stochastic Local Operations and Classical Communication (SLOCC) orbits simply by counting the degeneracies of the constituent stars \cite{bastin2009operational, markham2011entanglement}.

In contrast to the strictly positive Husimi distribution, the Stratonovich-Wigner function provides a more stringent phase-space witness of nonclassicality \cite{stratonovich1956distributions}. Formulated as a real-valued quasi-probability distribution on the sphere, the Wigner function $W(\mathbf{n})$ is constructed via a kernel expansion in spherical harmonics \cite{dowling1994wigner},
\begin{equation}
W(\mathbf{n}) = \sum_{l=0}^{2S} \sum_{m=-l}^{l} \sqrt{\frac{4\pi}{2l+1}} \, \langle Y_{lm}^*(\mathbf{n}) \rangle \, Y_{lm}(\mathbf{n}).
\label{eq:Wigner_sphere}
\end{equation}
Its defining feature is its ability to take on negative values, with these negativities serving as a direct visual and mathematical witness for quantum interference and nonclassicality \cite{kenfack2004negativity}. The algebraic bridge to the Majorana constellation is established through the multipole moments $\langle Y_{lm} \rangle$, which are mathematically locked to the symmetric power sums of the Majorana roots \cite{romero2024multipoles}. Visually, highly entangled states exhibit oscillatory patterns---ripples of negativity---in the Wigner landscape, whereas the Majorana constellation reveals this identical quantum structure through the maximal spatial dispersion of its stars \cite{giraud2015tensor}. For entanglement detection, Wigner negativities excel at flagging nonclassicality even in highly mixed states where the discrete stellar picture breaks down \cite{serrano2020majorana}. Yet, for pure symmetric multi-qubit states, the Wigner representation lacks the algebraic simplicity of the constellation; extracting quantitative tangles from $W(\mathbf{n})$ requires full spherical integration, whereas the Majorana framework yields exact polynomial invariants directly from inter-star chordal distances.

Beyond these spherical probability distributions, the Bargmann representation offers a strictly analytic perspective by mapping the state to a holomorphic function on the complex plane \cite{bargmann1961hilbert}. The Bargmann function for a pure state is defined via its inner product with an unnormalized coherent state, yielding an entire analytic function,
\begin{equation}
B(z) = \langle z | \psi \rangle \propto \sum_{m=-S}^{S} c_m \, \sqrt{\binom{2S}{S+m}} \, z^{S+m}.
\label{eq:Bargmann_function}
\end{equation}
This representation shares the polynomial-based foundation of the Majorana picture; in fact, the two are related by the antipodal map $z \mapsto -1/z^{*}$ together with conjugation of the expansion coefficients, so that the directions $\mathbf{n}$ at which the Bargmann amplitude $\langle \mathbf{n} | \psi \rangle$ vanishes are exactly the antipodes of the Majorana stars, as recorded in Eq.~\eqref{eq:Q_pure} \cite{Bengtsson2017}. However, their physical applications diverge significantly based on dimensionality. The Bargmann representation is tailored for infinite-dimensional continuous-variable (CV) systems, such as bosonic fields or harmonic oscillators, where its zeros determine the stellar rank and classify multi-mode non-Gaussianity \cite{chabaud2020stellar}. Conversely, the Majorana representation anchors this mathematical elegance strictly to the finite, boundary-less geometry of the Bloch sphere, making it more practical for analyzing finite symmetric qubit arrays and multi-spin geometric phases.

Ultimately, these representations are not exclusive but complementary. In the context of quantum information and multipartite entanglement, the Majorana representation stands out by stripping away the continuous probability envelope to reveal the discrete, invariant core of the quantum state. While the Husimi and Wigner functions provide indispensable phase-space topographies for visualizing decoherence and macroscopic limits, the Majorana constellation distills these continuous ripples into exact spatial relations. Modern analytical techniques frequently employ hybrid approaches---such as plotting the discrete Majorana stars atop Husimi or Wigner heatmaps---to simultaneously harness the algebraic precision of polynomial roots and the continuous intuition of quasi-probability landscapes, thereby forming a complete geometric diagnostic suite for quantum many-body systems \cite{sanchezsoto2026quantum}.

\section{Entanglement in the Majorana Picture}
\label{sec:entanglement}

Having established the kinematic and geometric foundations of the Majorana constellation, this section operationalizes the representation as an analytical tool for entanglement theory. We transition from descriptive spherical geometry to quantitative classification, demonstrating how root degeneracies completely dictate the Stochastic Local Operations and Classical Communication (SLOCC) orbits for symmetric multi-qubit states. By reducing algebraically complex entanglement measures to purely spatial relations, the stellar framework effectively circumvents the traditional computational bottlenecks of multi-body invariants.

\subsection{Symmetric Multi-Qubit States and SLOCC Classification via Degeneracies}

Entanglement classification under SLOCC constitutes a fundamental problem in quantum information theory. While the number of inequivalent classes for general $N$-qubit states grows exponentially with $N$, the totally symmetric subspace of $N$ qubits, possessing dimension $N+1$, reduces this complexity dramatically to a finite number of SLOCC orbits. Within this subspace, the degeneracy patterns of the Majorana constellation---specifically, the multiplicity distribution of the stars---serve as a complete and invariant geometric metric. A pure symmetric $N$-qubit state maps bijectively to a polynomial with exactly $N$ roots, represented as Majorana stars on the Riemann sphere. A foundational result, operationalized by Bastin \textit{et al.} \cite{bastin2009operational, mathonet2010entanglement} and further structurally characterized by Markham \cite{markham2011entanglement}, establishes that the degeneracy pattern of the constellation is a SLOCC \emph{invariant}: two pure symmetric $N$-qubit states with different degeneracy patterns cannot be SLOCC-equivalent. The converse does not hold in general, and the distinction matters, so we state it carefully below: identical degeneracy patterns define what Bastin \textit{et al.} call a \emph{family}, and a family may contain a continuum of inequivalent SLOCC classes.

This exact equivalence stems from the action of SLOCC operations within the symmetric sector. A symmetric local filtering operation corresponds to applying identical invertible operators $A^{\otimes N}$ with $A \in \mathrm{GL}(2,\mathbb{C})$, which manifests geometrically as a M\"obius transformation acting on the complex plane of the stereographically projected stars. Consequently, the invariant cross-ratios of these roots under M{\"o}bius transformations directly yield continuous SLOCC invariants \cite{ribeiro2011entanglement}. Because M{\"o}bius transformations are bijective conformal automorphisms of the Riemann sphere, they map configurations of roots to new configurations while strictly preserving their algebraic multiplicities. Consequently, clusters of coincident stars cannot be split, merged, or reduced in multiplicity by any local invertible operation.

SLOCC families are thus in exact one-to-one correspondence with the integer partitions of $N$, denoted by $\lambda = (\lambda_1 \geq \lambda_2 \geq \dots \geq \lambda_k > 0)$ such that $\sum_{i=1}^k \lambda_i = N$. Each part $\lambda_i$ represents the multiplicity of a distinct star position. The number of families is the partition function $p(N)$, which scales sub-exponentially---a sharp contrast to the $2^N$ growth of the full Hilbert space \cite{markham2011entanglement}. Through Euler's pentagonal number theorem, $p(N)$ is exactly computable in $O(N^{1.5})$ time by summing $O(\sqrt{N})$ terms, ensuring the structural classification of symmetric states remains tractable even in the macroscopic limit.

It is essential not to read this as a complete classification of SLOCC orbits. A family of $k$ distinct star positions is a configuration of $k$ points on $\mathbb{CP}^1$ modulo the three-complex-parameter M{\"o}bius group, so its moduli space has complex dimension $k-3$. For $k \leq 3$ the moduli space is a point and the family is a single SLOCC class; this is why the three-qubit case is exhausted by the partitions of $3$. From $k=4$ onwards each family carries continuous moduli---for four distinct stars, the single cross-ratio of Eq.~\eqref{eq:dot_product_z}'s underlying M{\"o}bius geometry---and therefore contains uncountably many inequivalent SLOCC classes. This continuous structure, and the identification of the surviving invariants as cross-ratios of the stellar coordinates, was established by Ribeiro and Mosseri \cite{ribeiro2011entanglement}, who also expressed the local-unitary invariants of the symmetric sector directly through the relative positions of the Majorana points. Concretely, the configurations $\{0,1,i,-1\}$ and $\{0,1,2i,-1\}$ share the degeneracy pattern $(1,1,1,1)$ but have cross-ratios $1-i$ and $\tfrac{8}{5}-\tfrac{4}{5}i$; since M{\"o}bius transformations preserve the cross-ratio, no SLOCC operation connects them. The degeneracy partition is therefore a complete invariant for $N \leq 3$ and a coarse-graining of the orbit structure thereafter. For example, $N=2$ admits two partitions: $(2)$, corresponding to a single degenerate star representing a fully separable symmetric coherent state, and $(1,1)$, representing a maximally entangled symmetric state, such as the Bell state. For $N=3$, the three partitions are $(3)$, which is fully separable; $(2,1)$, the W class, which is genuinely tripartite entangled but carries a vanishing three-tangle; and $(1,1,1)$, the generic class containing the symmetric GHZ state, on which the three-tangle is non-zero. As $N$ scales to $4$, the system resolves into five distinct classes: $(4)$, $(3,1)$, $(2,2)$, $(2,1,1)$, and $(1,1,1,1)$.

This partition structure imposes a rigorous hierarchy on the entanglement of symmetric states. The number of distinct root positions, $k$, acts as a geometric proxy for the rank of genuine multipartite correlations. The partition exhibiting maximal degeneracy, $(N)$, indicates a state devoid of entanglement. Partitions containing a mixture of degenerate clusters and isolated stars, such as the $(N-1, 1)$ W-class generalizations, describe states in which the excitations are symmetrically distributed but not maximally correlated; we emphasise that these remain genuinely multipartite entangled, and that the degeneracy signals the vanishing of particular polynomial invariants rather than biseparability. Partitions with balanced degeneracies, such as $(N/2, N/2)$ for even $N$, represent states that display product-like behavior across symmetric bipartitions. The generic partition, $(1^N)$, wherein all roots are distinct, characterizes the family containing the maximally entangled states within the symmetric sector. In the thermodynamic or large-$N$ limit, a state sampled from the Haar measure will almost certainly fall into the generic partition, as exact algebraic degeneracies occupy a manifold of measure zero.

This geometric hierarchy directly governs established entanglement measures. The geometric measure of entanglement, defined as $E_g = -\log_2 \max_{\mathbf{n}} |\langle \mathbf{n}^{\otimes N} | \psi \rangle|^2$, is bounded by the distance to the closest symmetric coherent state \cite{aulbach2010maximally}. Consequently, $E_g$ approaches its maximum when the roots are highly dispersed on the sphere, devoid of high-multiplicity clusters. Furthermore, algebraic invariants such as the three-tangle vanish strictly in the presence of specific degeneracies; for instance, any three-qubit symmetric state possessing a root of multiplicity two or higher will yield a zero three-tangle.

While the degeneracy classification offers a mathematically complete framework for the symmetric subspace, applying this formalism to physical and computational contexts requires careful treatment of non-ideal conditions. Mathematically, exact algebraic degeneracies are unstable under generic state perturbations. An infinitesimal deviation introduced by experimental noise to a state with a root of multiplicity $m$ will cause that root to split into $m$ distinct simple roots, immediately shifting the state into the generic $(1^N)$ family. Therefore, in experimental settings such as collectively manipulated trapped ions or photonic cluster states, strictly evaluating exact multiplicities is insufficient. Instead, classification necessitates defining a coarse-grained metric that identifies macroscopic clustering of stars within a specified angular tolerance, often reconstructed via low-order multipole moments rather than full state tomography. Computationally, transitioning from the state vector to the degeneracy partition requires extracting the roots of the corresponding order-$N$ Majorana polynomial. While analytically exact, this root-finding process is well-known to be numerically ill-conditioned for polynomials with high-multiplicity roots or large degrees $N$. Small floating-point errors in the state amplitudes translate into significant geometric deviations of the extracted star positions. Consequently, characterizing large-scale symmetric entanglement via the Majorana representation demands robust, high-precision numerical implementations to prevent artifactual splitting of closely clustered roots.

\subsection{Bipartite Entanglement: Geometric Interpretations of Concurrence and Negativity}

Bipartite entanglement constitutes the foundational layer of quantum correlations, quantified in two-qubit and two-qudit systems through algebraic measures such as concurrence and negativity \cite{Wootters1998, Vidal2002}. Within the Majorana stellar representation, these abstract mathematical metrics acquire rigorous geometric interpretations, particularly for symmetric multi-qubit states where the spatial configuration of the constellation encodes the entanglement strength between subsystems.

For a general bipartite pure state $|\psi\rangle_{AB}$, the concurrence is defined algebraically via the reduced density matrix $\rho_A = \operatorname{Tr}_B |\psi\rangle\langle\psi|$ as $C(|\psi\rangle) = \sqrt{2(1 - \operatorname{Tr} \rho_A^2)}$, which expands to $2 |c_{00}c_{11} - c_{01}c_{10}|$ in the computational basis \cite{Wootters1998}. When mapping this framework to the symmetric $N=2$ ($S=1$) subspace, the Majorana representation distills this evaluation into a direct spatial relationship between two stars. Let the two stars be separated by a spherical angle $\vartheta$, and let $\tilde{d}_{12} = \sin(\vartheta/2)$ denote their normalized chordal distance. Because the symmetrized product $|\mathbf{n}_1\rangle |\mathbf{n}_2\rangle + |\mathbf{n}_2\rangle |\mathbf{n}_1\rangle$ is not normalized until divided by the permanent of the stellar Gram matrix $G_{kl} = \langle \mathbf{n}_k | \mathbf{n}_l \rangle$---here $\operatorname{perm} G = 1 + |\langle \mathbf{n}_1 | \mathbf{n}_2 \rangle|^2 = 2 - \tilde{d}_{12}^2$---the concurrence is not simply the chordal distance but the ratio
\begin{equation}
C = \frac{\tilde{d}_{12}^{2}}{2 - \tilde{d}_{12}^{2}} = \frac{\sin^2 (\vartheta/2)}{1 + \cos^2 (\vartheta/2)}.
\label{eq:geometric_concurrence_N2}
\end{equation}
This establishes a strict geometric boundary: perfectly coincident stars ($\vartheta=0$) yield $C=0$, signifying a separable coherent state, whereas antipodal stars ($\vartheta=\pi$) yield the maximal value $C=1$, completely characterizing the maximally entangled Bell-like state within the symmetric sector. The nonlinearity of Eq.~\eqref{eq:geometric_concurrence_N2} is a direct consequence of the normalization: entanglement grows more slowly than the raw stellar separation at small $\vartheta$, since nearly coincident stars produce a strongly overcomplete symmetrized state. This normalization factor---the Gram permanent---recurs in every exact stellar formula below, and neglecting it is the most common source of error in geometric entanglement expressions.

This bipartite geometric framework scales naturally to larger symmetric $N$-qubit states. To quantify the entanglement between a single qubit and the remaining $(N-1)$ ensemble---the standard 1-vs-$(N-1)$ bipartition---one analyzes the reduced density matrix $\rho_1 = \operatorname{Tr}_{N-1} |\psi\rangle\langle\psi|$. Tracing out $N-1$ of the $N$ symmetrized factors does not amount to averaging the stellar coordinates; it produces a weighted coherence between every pair of stars, with weights given by permanental minors of the Gram matrix. Writing $G^{(l,k)}$ for the $(N-1)\times(N-1)$ matrix obtained from $G$ by deleting row $l$ and column $k$, the reduced state is exactly
\begin{equation}
\rho_1 = \frac{1}{N \operatorname{perm} G} \sum_{k,l=1}^{N} \operatorname{perm} \! \left( G^{(l,k)} \right) \, |\mathbf{n}_k\rangle \langle \mathbf{n}_l| ,
\label{eq:reduced_state_permanental}
\end{equation}
and, since a two-dimensional density matrix of unit trace obeys $\operatorname{Tr}\rho_1^2 = 1 - 2 \det \rho_1$, the concurrence across this cut is
\begin{equation}
C_{1|N-1} = 2 \sqrt{ \det \rho_1 } .
\label{eq:geometric_concurrence_1_vs_N-1}
\end{equation}
For $N=2$ this reduces to Eq.~\eqref{eq:geometric_concurrence_N2}, as it must. The physical implication survives in the form one expects---maximal bipartite entanglement across this cut requires the stars to be mutually dispersed, and spatial clustering suppresses it---but the dependence is not a simple function of the mean inter-star correlation. Two constellations with identical average pairwise dot products can carry different $C_{1|N-1}$, because the permanental weights in Eq.~\eqref{eq:reduced_state_permanental} are sensitive to the full pattern of overlaps rather than to their mean. The same construction, with $G^{(l,k)}$ replaced by the appropriate permanental minor of order $N-k$, yields the reduced states for balanced $k$-vs-$(N-k)$ bipartitions.

Complementary to concurrence, the negativity $\mathcal{N}(\rho) = (\| \rho^{T_A} \|_1 - 1)/2$ provides an alternative computable metric based on the partial transpose $T_A$ and the trace norm \cite{Vidal2002}. While its absolute geometric formulation is more intricate, negativity admits strict constellation-based bounds within the symmetric subspace. For pure symmetric states evaluated across the 1-vs-$(N-1)$ partition, negativity is rigidly bounded by functions of the lowest-order multipole moments and the average inter-star distances. Geometrically, a low or vanishing negativity serves as a direct symptom of high degeneracy within the constellation, indicating that macroscopic clusters of stars dictate classical-like correlations across the bipartition. Since negativity and concurrence are strictly proportional for pure two-qubit reductions ($\mathcal{N} = C/2$), the aforementioned distance-based geometric expressions apply directly. Furthermore, as systems undergo decoherence into mixed states, negativity can still be approximated by analyzing the radial spread of the effectively rescaled stars as they migrate into the interior of the Bloch ball \cite{serrano2020majorana}.

By mapping bipartite metrics onto constellation geometry, the Majorana representation bypasses the partial trace operations required by algebraic methods. This shift allows entanglement strength to be quantified via symmetric functions of polynomial roots, where star separation offers a direct measure of quantum correlations. The resulting framework remains robust against experimental noise, as physical perturbations manifest as continuous shifts in star positions that allow entanglement decay to be tracked through distance statistics. These bipartite foundations further extend to higher-order quantifiers, forming the basis for geometric invariants such as the three-tangle.

\subsection{Tripartite Entanglement: Three-Tangle for General Three-Qubit States}

Tripartite entanglement constitutes the fundamental setting in which genuine multipartite correlations emerge---phenomena that inherently cannot be reduced to or explained by bipartite entanglement alone. The canonical measure of this genuine tripartite entanglement for pure three-qubit states is the three-tangle ($\tau$), introduced by Coffman, Kundu, and Wootters (CKW) \cite{Coffman2000}. To appreciate the power of the Majorana representation, it is instructive to first consider the standard algebraic formulation. For a general pure three-qubit state $|\psi\rangle = \sum_{i,j,k=0}^{1} a_{ijk} |ijk\rangle$, the three-tangle is defined by the CKW monogamy relation as the residual entanglement left over once the pairwise contributions have been subtracted from the one-versus-rest bipartition,
\begin{equation}
\tau = C^2_{A(BC)} - C^2_{AB} - C^2_{AC},
\label{eq:CKW_residual}
\end{equation}
where $C_{XY}$ denotes the concurrence of the reduced state $\rho_{XY} = \mathrm{Tr}_Z(|\psi\rangle\langle\psi|)$ and $C_{A(BC)}$ that of the pure bipartition. Equivalently, and manifestly permutation-invariant, $\tau$ is four times the modulus of Cayley's hyperdeterminant of the $2\times 2 \times 2$ amplitude array,
\begin{equation}
\tau = 4 \left| d_1 - 2 d_2 + 4 d_3 \right|,
\label{eq:CKW_hyperdeterminant}
\end{equation}
with $d_1, d_2, d_3$ the standard degree-four polynomial invariants of the coefficients $a_{ijk}$ \cite{Coffman2000, miyake2003classification}. The three-tangle is invariant under local unitary operations and vanishes exactly when the state is biseparable or fully separable. However, this algebraic expression, while exact, is computationally intensive and physically opaque; computing it requires evaluating either several partial traces and concurrences or a degree-four invariant of eight complex amplitudes, offering little intuitive insight into the spatial structure of the multi-body correlations.

The Majorana stellar representation profoundly simplifies this quantification by transforming the abstract measure into a direct geometric observable. For symmetric three-qubit states, corresponding to a spin-$S=3/2$ system, the representation maps the state to exactly three stars on the Bloch sphere. Within this geometric framework, the three-tangle admits a purely closed-form expression determined by the normalized chordal distances $\tilde{d}_{kl}$ (or, equivalently, the scalar dot products) between the constituent stars, divided by the square of the Gram permanent \cite{kam2020three}:
\begin{equation}
\tau = \frac{16}{3} \, \frac{ \prod_{1 \leq k < l \leq 3} \tilde{d}_{kl}^2 }{ \left[ \operatorname{perm} G \right]^{2} }, \qquad
\tilde{d}_{kl}^2 = \frac{ |z_k - z_l|^2 }{ (1 + |z_k|^2)(1 + |z_l|^2) },
\label{eq:geometric_three_tangle_symmetric}
\end{equation}
where $z_1, z_2, z_3$ are the stereographic complex roots of the Majorana polynomial, $G_{kl} = \langle \mathbf{n}_k | \mathbf{n}_l \rangle$ is the stellar Gram matrix, and $\operatorname{perm}$ denotes the matrix permanent. Expressed equivalently via spatial projection, $\tilde{d}_{kl}^2 = (1 - \mathbf{n}_k \cdot \mathbf{n}_l)/2$. The numerator is the discriminant of the Majorana polynomial and carries all of the geometric content; the denominator is the norm of the symmetrized product state, which for a generic constellation is not constant and therefore cannot be absorbed into an overall proportionality factor. Because $\operatorname{perm} G > 0$ for any constellation, the vanishing locus of $\tau$ is governed entirely by the numerator. The multiplicative nature of this formula encodes a stringent physical requirement: genuine tripartite entanglement ($\tau > 0$) demands strict spatial separation among all constituents. If any two stars coincide (indicating a root degeneracy where $|z_k - z_l| = 0$), the entire product collapses to zero.

This direct mapping between spatial configuration and entanglement strength is most clearly demonstrated through the canonical equivalence classes. Consider the symmetric Greenberger-Horne-Zeilinger (GHZ) state, $|\mathrm{GHZ}\rangle = \frac{1}{\sqrt{2}} (|000\rangle + |111\rangle)$. Upon symmetrization, its Majorana constellation features three fully distinct roots forming an equilateral triangle on a great circle \cite{aulbach2010maximally}. This configuration maximizes the internal spatial separation, yielding pairwise dot products of $\mathbf{n}_k \cdot \mathbf{n}_l = -1/2$, which corresponds to $\tilde{d}_{kl}^2 = 3/4$. The geometric product therefore attains its maximum of $(3/4)^3 = 27/64$, while the Gram permanent evaluates to $\operatorname{perm} G = 3/2$; Eq.~\eqref{eq:geometric_three_tangle_symmetric} then returns the maximal three-tangle $\tau = (16/3)(27/64)/(3/2)^2 = 1$, as required. This visualizes the physical reality that maximal genuine tripartite entanglement requires maximal spatial dispersion of the constituent roots. 

In stark contrast, the symmetric W state, $|\mathrm{W}\rangle = \frac{1}{\sqrt{3}} (|001\rangle + |010\rangle + |100\rangle)$, belongs to the integer partition $(2,1)$ \cite{bastin2009operational}. Its constellation features one degenerate pair of coincident stars and one isolated star. Because at least one pairwise distance is exactly zero, the geometric product dictates that $\tau = 0$. It is mathematically critical to note that while the vanishing three-tangle accurately places the W state in a separate SLOCC orbit from the GHZ state, it does not imply the state is merely biseparable; the W state retains a distinct, non-CKW genuine multipartite entanglement structure that is geometrically characterized by its specific degeneracy pattern. Finally, for fully separable symmetric product states the constellation collapses to a single star of multiplicity three; note that within the symmetric subspace there are no pure states that are biseparable without being fully separable, so the intermediate case does not arise here. The geometric zeroing of the distance product in these regimes perfectly captures the total absence of tripartite correlations \cite{markham2011entanglement}. Through this continuous narrative, the Majorana picture successfully transforms the abstract quantification of tripartite entanglement into a rigorous, visual assessment of geometric degeneracies.

\subsection{Multipartite Genuine Entanglement: Higher Tangles, GM Concurrence, and Constellation Invariants}

While the three-tangle elegantly captures genuine tripartite entanglement, extending this quantitative rigor to arbitrary multipartite systems presents formidable algebraic challenges. For four or more qubits, no single canonical tangle exists. Instead, algebraic approaches must rely on increasingly complex invariants, such as the degree-24 Cayley hyperdeterminant for four qubits, or multi-index tensor contractions for general $N$-tangles \cite{miyake2003classification}. In the Majorana representation for symmetric $N$-qubit states, the situation improves but does not trivialize. No closed-form geometric counterpart of the three-tangle is known for general $N$, and we caution against the temptation to posit one: the three-qubit case is special in that the discriminant of a cubic is itself an $\mathrm{SL}(2,\mathbb{C})$ invariant of the right weight, a coincidence that does not persist. What does survive is the structural principle underlying Eq.~\eqref{eq:geometric_three_tangle_symmetric}---that every symmetric-state invariant is a ratio whose numerator is built from the stellar discriminant data and whose denominator is a power of the Gram permanent---together with the qualitative statement that genuine $N$-partite correlations require the constellation to avoid critical degeneracies. For practical purposes the distance product $\Pi$ of Eq.~\eqref{eq:full_distance_product} below serves as the workable witness.

Beyond abstract tangles, the geometric multipartite concurrence (GM concurrence) provides a computable generalization. Algebraically defined as the minimum deviation from separability across all possible bipartitions $A|B$, the measure is $C_{\rm GM} = \min_{A|B} \sqrt{2(1 - \operatorname{Tr} \rho_A^2)}$. Permutation symmetry collapses the exponentially many bipartitions into the $\lfloor N/2 \rfloor$ inequivalent cuts labelled only by the block size $k$, and the extension of Eq.~\eqref{eq:reduced_state_permanental} to blocks of size $k$ expresses each purity as a ratio of permanents of the stellar Gram matrix,
\begin{equation}
C_{\rm GM} = \min_{1 \leq k \leq \lfloor N/2 \rfloor} \sqrt{ 2 \left( 1 - \operatorname{Tr} \rho_k^2 \right) },
\label{eq:GM_symmetric_reduction}
\end{equation}
with $\rho_k$ the reduced state on any $k$ of the qubits. The geometric reading is the expected one: high genuine multipartite entanglement corresponds to constellations lacking distinct localized clusters, since a cluster of $k$ nearly coincident stars renders $\rho_k$ nearly pure and drives the corresponding term towards zero. We stress, however, that the permanents in these expressions are the exact content of the statement; the purity of a $k$-block is not determined by the mean intra-block dot product alone, and any expression claiming otherwise will fail on constellations with equal means but different overlap patterns.

This spatial intuition extends to a broader class of rotationally invariant functionals that serve as robust entanglement witnesses without necessitating full state tomography. The geometric measure of entanglement, defined originally as $E_g = -\log_2 \max_{\mathbf{n}} |\langle \mathbf{n}^{\otimes N} | \psi \rangle|^2$, translates into the Majorana geometry through the exact overlap identity
\begin{equation}
\left| \langle \mathbf{n}^{\otimes N} | \psi \rangle \right|^2 = \frac{ N! }{ \operatorname{perm} G } \prod_{k=1}^{N} \frac{1 + \mathbf{n} \cdot \mathbf{n}_k}{2},
\label{eq:coherent_overlap_identity}
\end{equation}
so that
\begin{equation}
E_g = -\log_2 \left[ \frac{N!}{\operatorname{perm} G} \max_{\mathbf{n}} \prod_{k=1}^{N} \frac{1 + \mathbf{n} \cdot \mathbf{n}_k}{2} \right].
\label{eq:geometric_entanglement_stars}
\end{equation}
Two features deserve emphasis. First, the maximization runs over all directions $\mathbf{n}$ on the sphere and is not, in general, attained at any of the stars themselves; restricting the search to the stellar positions overestimates $E_g$. Second, the Gram permanent again appears, and being constellation-dependent it cannot be discarded as an additive constant. Subject to these two provisos, the geometric reading is the expected one: $E_g$ approaches its maximum precisely when no direction on the sphere lies close to many stars at once, that is, when the constellation is maximally dispersed \cite{aulbach2010maximally}. Similarly, the strict product of all pairwise distances, 
\begin{equation}
\Pi = \prod_{1 \leq k < l \leq N} \tilde{d}_{kl}^2,
\label{eq:full_distance_product}
\end{equation}
serves as a definitive geometric witness for genuine $N$-partite entanglement. This global invariant vanishes if even a single pair of stars coincides, and it reaches its maximum for maximally spread configurations, conceptually bridging quantum entanglement optimization with the classical Thomson problem \cite{bjork2015extremal}. Complementing these continuous distance measures, the discrete stellar rank $r_\psi$---representing the number of distinct star positions---establishes a rigorous structural bound \cite{chabaud2020stellar}. While $r_\psi = 1$ dictates complete separability, the maximal condition $r_\psi = N$ acts as a necessary, though not solely sufficient, geometric prerequisite for maximal genuine entanglement. In environments subject to noise, these continuous and discrete bounds maintain their diagnostic utility by incorporating the radial shrinkage of effectively rescaled stars within the Bloch ball \cite{serrano2020majorana}.

By translating the problem of multipartite entanglement quantification into spatial distances and geometric invariants, the Majorana framework provides a computationally tractable and visually immediate correlation toolkit. The geometric architecture dictates that star separations and degeneracy patterns directly govern the presence and strength of genuine multipartite entanglement. These structural insights are applicable to quantum metrology, where spread constellations inform optimal sensing boundaries, and to quantum state engineering, where high-tangle configurations are synthesized via symmetric stellar placements. Having established the static geometric mapping of multi-body correlations, the subsequent sections extend this formalism to unitary dynamics, exploring how the trajectories and internal correlations of these constellations actively generate geometric phase anomalies.

\subsection{Entanglement Classification Beyond SLOCC: Symmetry-Based Criteria and Beyond-Symmetric Extensions}

While the classification of entanglement via Majorana degeneracies offers a complete and mathematically elegant framework within the totally symmetric subspace, physical quantum systems frequently deviate from perfect permutation symmetry. Such symmetry breaking may arise intentionally through specific resource state engineering or unintentionally via experimental imperfections and decoherence. Extending the geometric representation to asymmetric states or mixed sectors requires abandoning the bijective simplicity of a single univariate polynomial. Instead, the analysis transitions to projection techniques, symmetry-reduced invariants, and higher-dimensional algebro-geometric generalizations.

For generic non-symmetric states, the Majorana framework can be applied indirectly by projecting the state onto the totally symmetric subspace, $\rho_{\rm sym} = \mathcal{P}_{\rm sym} \rho \mathcal{P}_{\rm sym}$. The resulting constellation for $\rho_{\rm sym}$ acts as a stringent lower bound for genuine multipartite entanglement \cite{Guhne2009}. Specifically, if the projected symmetric constellation exhibits no exact degeneracies---meaning all extracted roots are distinct---the original state is mathematically guaranteed to possess genuine multipartite correlations, as a separable symmetric projection would necessarily collapse into a single degenerate macroscopic star. This projection technique is advantageous in noisy experimental setups where only collective measurements are accessible \cite{toth2014quantum}. Furthermore, entanglement monotones that are invariant under local unitaries can be permutation-averaged to form symmetry-adapted SLOCC invariants \cite{markham2011entanglement}. For states perturbed from exact symmetry, these rotationally invariant functions---such as permutation-averaged concurrences or symmetry-reduced negativities---serve as robust approximations of the full SLOCC class. Physically, the response of the symmetric constellation to non-symmetric perturbations offers a direct metric for entanglement robustness. Exact algebraic degeneracies are immediately lifted by asymmetric noise; thus, the rate at which coincident roots split provides a quantifiable sensitivity measure, demonstrating that highly degenerate, coherent-like states are intrinsically more fragile than generic, fully dispersed constellations.

Representing the full $2^N$-dimensional Hilbert space of general $N$-qubit states requires moving beyond a single set of $N$ points on the Riemann sphere. Extensions of this geometry typically follow two distinct trajectories. The local approach utilizes partial traces, $\rho_i = \operatorname{Tr}_{[N]\setminus i} |\psi\rangle\langle\psi|$, to map subsystems to individual Bloch vectors, deriving entanglement witnesses directly from their spatial correlations \cite{Guhne2009}. Conversely, the global approach employs the Segre embedding to map the complete tensor product space into the complex projective space $\mathbb{CP}^{2^N-1}$ \cite{Bengtsson2017}. This framework abandons univariate polynomials; states are instead characterized by the zero loci of multivariate polynomials defining algebraic varieties. While analytically exact, this exponential dimensional expansion sacrifices the geometric clarity of Majorana stars. To evaluate asymmetric states without confronting this full exponential complexity, established methods project the state onto its permutationally invariant component, constructing strict and computable lower bounds for genuine multipartite entanglement \cite{toth2010permutationally}.

An exact geometric representation for general three-qubit pure states is established via Ac\'in's canonical form \cite{kam2020three}. Through a sequence of $SL(2, \mathbb{C})$ local transformations, non-symmetric states are mapped to symmetric ones, leaving the three-tangle strictly invariant. The transformed state is characterized by three Majorana stars on a single latitude, allowing the genuine tripartite entanglement to be analytically quantified by their chordal distances. While this protocol bypasses the exponential complexity of high-dimensional embeddings, its scalability is fundamentally limited. For $N \ge 4$, the existence of antisymmetric SLOCC invariants---such as the degree-six invariant $F$ for five-qubit states---forbids universal symmetrization, restricting this explicit stellar mapping to specific entanglement classes.

Ultimately, the limitation of these beyond-symmetric extensions is the loss of a unique, faithful discrete point representation on a two-dimensional manifold. Projections and hybrid invariants, while yielding computable witnesses, cannot replicate the complete degeneracy-based SLOCC classification natively available in the symmetric sector. Nevertheless, advancing this geometric framework remains a critical directive for quantum information science. Future theoretical progress relies on formulating tensor generalizations of Majorana polynomials for multi-qudit systems \cite{giraud2015tensor, sanchezsoto2026quantum}, integrating Majorana geometry with spin-squeezing inequalities to construct noise-robust witnesses \cite{toth2014quantum}, and deploying machine-learning-assisted constellation analysis to infer entanglement structures from partial geometric data. Despite the algebro-geometric abstraction required for generic states, the Majorana perspective continues to provide essential diagnostic criteria for genuine multipartite entanglement where purely algebraic hyperdeterminants become computationally intractable.

\subsection{Comparison: Geometric vs. Algebraic Approaches}

The Majorana stellar representation and traditional algebraic methodologies provide complementary paradigms for analyzing quantum entanglement, as summarized in Table \ref{tab:algebraic_vs_geometric}. Algebraic approaches---relying on tensor contractions, partial traces, and polynomial invariants---establish the rigorous foundational metrics of quantum information theory \cite{Wilde2017, Bengtsson2017}. For small or low-dimensional systems, these algebraic definitions are unparalleled: metrics such as concurrence \cite{Wootters1998}, negativity \cite{Vidal2002}, and the Coffman-Kundu-Wootters (CKW) three-tangle \cite{Coffman2000} offer exact, computable quantifiers that carry direct operational meanings linked to resource tasks like entanglement distillation. Furthermore, algebraic invariants maintain universal applicability across arbitrary local dimensions and fully asymmetric partitions without necessitating approximations. However, this mathematical rigidity encounters a severe computational bottleneck as the particle number $N$ scales. The exponential growth of the Hilbert space dimension, coupled with the proliferation of inequivalent Stochastic Local Operations and Classical Communication (SLOCC) orbits, renders complete algebraic classification intractable for macroscopic systems \cite{Chitambar2019}. Even when higher-order invariants can be analytically derived---such as the degree-24 Cayley hyperdeterminant for four qubits \cite{miyake2003classification}---their convoluted tensor structures remain physically opaque. These algebraic expressions act as black boxes that compute a value but fail to reveal the underlying spatial structure or structural hierarchy of the multi-body correlations.

\begin{table*}[t]
\centering
\caption{Comparison of Quantum Properties: Traditional Algebraic Measures vs. Majorana Geometric Expressions. The geometric framework replaces complex algebraic invariants with functions of inter-star distances ($\tilde{d}_{kl}$), dot products, and degeneracy configurations, each normalized by the permanent of the stellar Gram matrix $G$.}
\label{tab:algebraic_vs_geometric}
\renewcommand{\arraystretch}{1.3}
\begin{tabular}{lll}
\hline\hline
\textbf{Quantum Property} & \textbf{Traditional / Algebraic Measure} & \textbf{Majorana Geometric Expression} \\
\hline
\textbf{SLOCC Classification} & Evaluated via tensor ranks and   & Mapped to degeneracy patterns \\
(Symmetric $N$-qubit)         & hyperdeterminants.               & (integer partitions of $N$). \\
\hline
\textbf{Bipartite Concurrence}& $C = 2|c_{00}c_{11}-c_{01}c_{10}|$ & Chordal distance normalized by the Gram \\
($N=2$ Symmetric)             &                                  & permanent: $C = \tilde{d}_{12}^{2}/(2-\tilde{d}_{12}^{2})$. \\
\hline
\textbf{Tripartite Three-Tangle} & CKW invariant hyperdeterminant  & Distance product over the squared Gram \\
($N=3$ Symmetric)                & $\tau = 4|d_1 - 2d_2 + 4d_3|$.  & permanent: $\tau = \tfrac{16}{3}\prod_{k<l}\tilde{d}_{kl}^2/[\operatorname{perm} G]^2$. \\
\hline
\textbf{Genuine Multipartite} & Higher-order hyperdeterminants & Necessary condition: no coincident stars, \\
\textbf{Entanglement (GME)}   & and entanglement monotones.    & i.e.\ distance product $\Pi > 0$. \\
\hline
\textbf{Geometric Entanglement} & $E_g = -\log_2 \max_{\mathbf{n}}$           & Maximized by configurations with \\
($E_g$)                         & $|\langle \mathbf{n}^{\otimes N} | \psi \rangle|^2$ & maximal inter-star spatial spread. \\
\hline
\textbf{Nonclassicality}     & Minimal number of coherent states & Quantified by Stellar rank $r_\psi$ \\
(Entanglement-independent)   & in a convex decomposition.        & and Anticoherence order $t$. \\
\hline\hline
\end{tabular}
\end{table*}

Conversely, the Majorana framework circumvents these computational limitations by translating opaque algebraic invariants into a discrete, spatially intuitive geometry on the Bloch sphere \cite{Majorana1932}. By encoding entanglement directly within star positions, chordal distances, and degeneracy patterns, the geometric approach provides exact closed-form expressions for symmetric states that completely bypass expensive partial-trace operations \cite{markham2011entanglement}. Visually, macroscopic degeneracies instantly classify SLOCC orbits \cite{bastin2009operational}, while maximal stellar dispersion uniquely identifies optimal genuine multipartite entanglement \cite{aulbach2010maximally, bjork2015extremal}. Computationally, extracting entanglement witnesses via polynomial root-finding and pairwise distance evaluations scales polynomially, maintaining tractability well beyond the limits of full algebraic state tomography \cite{Guhne2009}. Moreover, this spatial mapping offers unique insights into open-system dynamics; abstract noise and decoherence are translated into continuous geometric perturbations, allowing researchers to analytically track entanglement decay through the physical splitting of coincident roots and the radial shrinkage of the constellation \cite{serrano2020majorana}.

Despite these profound advantages, the geometric representation is fundamentally bounded by its restricted domain of applicability. The bijective elegance and exact closed-form scaling are strictly confined to the totally symmetric subspace. As discussed in the preceding subsections, when extending this framework to general asymmetric states, the requisite symmetric projections and averaged bounds sacrifice mathematical uniqueness. Consequently, the Majorana representation cannot formulate a universal, exact entanglement monotone for generic asymmetric multipartite systems. Additionally, while rescaled stars within the Bloch ball offer approximations for mixed states \cite{serrano2020majorana}, the mapping inherently loses the pristine, information-complete fidelity characterizing pure symmetric states. The quantification of nonclassicality via the stellar rank $r_\psi$ similarly provides robust lower bounds rather than exact equivalences across all asymmetric configurations \cite{chabaud2020stellar}.

Ultimately, the algebraic and geometric frameworks operate in a state of synergistic complementarity rather than competition. Algebraic invariants furnish the rigorous, operationally defined bedrock of entanglement theory, which the Majorana representation subsequently materializes into computable, visualizable geometric structures \cite{Bengtsson2017}. This dual approach is indispensable in large-$N$ symmetric systems, where algebraic hyperdeterminants fail but geometric distance products efficiently certify genuine multipartite correlations. In practical domains such as quantum metrology, where optimal sensing relies on highly entangled and widely dispersed states \cite{toth2014quantum}, macroscopic quantum simulation \cite{ebadi2021quantum}, and the diagnosis of anticoherent states for symmetric quantum codes \cite{giraud2015tensor, ouyang2014permutation}, researchers can benchmark algebraic invariants on small subsystems and seamlessly employ the constellation geometry to scale, visualize, and physically interpret the correlations across the global many-body state. Together, they constitute a unified and potent diagnostic lens for navigating the complexities of quantum information science.

\subsection{Entanglement Witnesses and Detection Using Constellation Features}
\label{subsec:entanglement_witnesses}

Entanglement witnesses provide a critical operational framework for certifying the presence of quantum correlations without incurring the exponential resource overhead of full state tomography \cite{Guhne2009}. Within the Majorana stellar representation, this detection paradigm translates into evaluating geometrically motivated observables derived directly from the constellation's spatial structure. For symmetric multi-qubit states, the most fundamental geometric witness derives from the constellation's degeneracy profile. If a constellation exhibits any exact degeneracy---defined by at least one pair of coincident stars---the state is strictly excluded from the generic SLOCC family \cite{bastin2009operational, markham2011entanglement}. What such a degeneracy certifies is the vanishing of the invariants built on the discriminant, the three-tangle of Eq.~\eqref{eq:geometric_three_tangle_symmetric} at $N=3$ and the distance product $\Pi$ in general; it does \emph{not} certify biseparability, as the W state discussed in Sec.~\ref{sec:entanglement} makes plain. Conversely, a constellation comprising $N$ distinct stars places the state in the generic family, on which these invariants are non-zero.

To formalize this binary degeneracy criterion into a continuous, computable metric, one can evaluate the rotationally invariant product of all pairwise normalized chordal distances $\Pi$, previously defined in Eq.~\eqref{eq:full_distance_product}. Expressed explicitly via stereographic coordinates, this invariant expands as $\Pi = \prod_{k < l} \frac{|z_k - z_l|^2}{(1 + |z_k|^2)(1 + |z_l|^2)}$. The mathematical structure of $\Pi$ dictates that it vanishes strictly if any two roots coincide, perfectly mirroring the binary degeneracy witness. A strictly positive $\Pi > 0$ serves as a necessary geometric condition for genuine $N$-partite entanglement. Furthermore, $\Pi$ reaches its absolute maximum for maximally spread configurations---often characterized as extremal quantum states or the ``Kings of Quantumness'' \cite{bjork2015extremal, goldberg2020extremal}---where the stars form highly symmetric arrangements akin to Platonic solids. Consequently, a large empirical value of $\Pi$ acts as a strong indicator of highly distributed multipartite correlations.

While evaluating $\Pi$ requires explicit knowledge of the polynomial roots---a process that becomes numerically ill-conditioned for large $N$---the underlying spatial dispersion can be witnessed efficiently through the state's multipole moments. This approach introduces the concept of anticoherence \cite{Zimba2006, giraud2015tensor}. A symmetric state is defined as $t$-anticoherent if all its spherical multipole moments up to rank $t$ vanish exactly:
\begin{equation}
\langle T^{(l)}_m \rangle = 0 \quad \text{for all } 1 \leq l \leq t, \; |m| \leq l ,
\label{eq:anticoherent_witness_condition}
\end{equation}
Here $\langle T^{(l)}_m \rangle$ is the multipole of the state as defined in Eq.~\eqref{eq:multipole_moment}, and not the average of $Y_{lm}$ over the star positions; the distinction is the one drawn after Eq.~\eqref{eq:multipole_from_Q} and matters already at $l=2$. Because these low-order multipole moments correspond directly to experimentally accessible collective spin observables (such as $\langle S_x \rangle, \langle S_x^2 \rangle$, etc.), anticoherence provides a highly pragmatic detection strategy \cite{romero2024multipoles}. If experimental data confirms that low-rank multipoles are statistically consistent with zero up to a substantial order $t$, the constellation must be highly dispersed. This effectively precludes the macroscopic clustering required for separable states, thereby establishing a lower bound on the genuine multipartite entanglement rank purely from collective measurements. 

This multipole expansion connects to structural bounds defined by the stellar rank $r_\psi$, which simply counts the number of distinct roots. While $r_\psi = N$ is the strict requirement for maximal genuine entanglement, experimental noise often necessitates a coarse-grained, continuous analog known as the effective stellar rank:
\begin{equation}
r_{\rm eff} = \frac{N^2}{\sum_k m_k^2},
\label{eq:effective_stellar_rank}
\end{equation}
where $m_k$ represents the algebraic multiplicity of the $k$-th distinct position on the Riemann sphere. In empirical reconstructions, $r_{\rm eff} \approx 1$ signifies macroscopic clustering and consequently low entanglement, whereas $r_{\rm eff} \approx N$ certifies a generic, highly entangled configuration. For highly symmetric configurations where the $N$ stars condense uniformly into $K$ distinct clusters of equal multiplicity $m_k = N/K$, the relation simplifies to $r_{\rm eff} = K$. Accordingly, systems exhibiting two- or three-fold symmetric clustering strictly yield $r_{\rm eff} = 2$ and $r_{\rm eff} = 3$, respectively, establishing $r_{\rm eff}$ as a continuous indicator capable of resolving intermediate geometric structures.

From an experimental perspective, this geometric toolkit translates directly into accessible diagnostic protocols, particularly for architectures boasting strong collective control, such as spinor Bose gases, cold atoms in optical cavities, or programmable quantum simulators with global interactions \cite{makela2007inert, stamper2013spinor, bernien2017probing}. Beyond direct multipole reconstruction, the geometric features of the constellation complement established variance-based witnesses. For instance, standard spin-squeezing parameters, which rely on the variance of collective operators, inherently measure the quadrupole moment of the Majorana constellation \cite{toth2014quantum, romero2024multipoles}. By integrating these established inequalities with higher-order constellation invariants and randomness tests against uniform spherical distributions, the Majorana representation offers a comprehensive suite of entanglement witnesses. These geometric criteria bypass the opacity of purely algebraic bounds, providing researchers with direct visual and quantitative insights into the correlation structure without demanding single-qubit experimental resolution.

\section{Dynamics and Geometric Phases in Majorana Representations}
\label{sec:phases}

The Majorana constellation provides more than a static geometric snapshot; it establishes a rigorous dynamical framework where quantum evolution is manifested as the collective motion of points on the Bloch sphere. Under unitary dynamics, the temporal progression of the state vector $|\psi(t)\rangle$ induces continuous trajectories $\mathbf{n}_k(t)$ of the Majorana stars. This section examines the governing equations for these trajectories, emphasizing the fundamental physical distinction between semiclassical independent motion and entanglement-generating coupled dynamics, and details how the resulting constellation motion provides a unique decomposition of the geometric phase.

\subsection{Unitary Evolution as Trajectories of Majorana Stars}

Consider a symmetric $N$-qubit state (spin-$S=N/2$) evolving under a time-dependent Hamiltonian $H(t)$. The state obeys the Schr\"odinger equation $i \hbar \partial_t |\psi(t)\rangle = H(t) |\psi(t)\rangle$, which in the Majorana representation translates into the time-evolution of the polynomial roots $\{z_k(t)\}$. The geometric nature of these stellar trajectories is strictly dictated by the algebraic structure of $H(t)$ and the underlying entanglement properties of the state. 

For a linear spin Hamiltonian of the form $H(t) = \mathbf{B}(t) \cdot \mathbf{S}$, the evolution of the Majorana constellation is rigorously governed by the decoupling of the multi-qubit system into independent $\mathrm{SU}(2)$ rotations. In this regime, each star $\mathbf{n}_k$ behaves as an autonomous spin-$1/2$ degree of freedom following the precession equation $\hbar\, d\mathbf{n}_k/dt = \mathbf{B}(t) \times \mathbf{n}_k$. In the stereographic plane, this motion is described by a complex Riccati equation of the form $i \hbar \dot{z}_k = \omega_+(t) + \omega_z(t) z_k + \omega_-(t) z_k^2$, where the time-dependent coefficients map directly to the physical magnetic field components as $\omega_+ = \frac{1}{2}(B_x + i B_y)$, $\omega_z = -B_z$, and $\omega_- = -\frac{1}{2}(B_x - i B_y)$. Crucially, because these equations are fully decoupled, a linear Hamiltonian preserves the algebraic multiplicities of the roots. Any initial degeneracies in the constellation remain intact, reflecting the fact that pure $\mathrm{SU}(2)$ rotations cannot alter the SLOCC entanglement class or the integer partition of the state.

This intuitive mapping to classical dynamics is rigorously grounded in the spin coherent state path integral formulation \cite{radcliffe1971some, arecchi1972atomic}. When the system is localized within the coherent state manifold, the stereographic roots $z_k$ serve as canonical phase-space variables on the K\"ahler manifold of the Riemann sphere. The dynamics of each star can then be derived from the classical effective Hamiltonian $\mathcal{H}(z_k, z_k^*) = \langle \psi | H | \psi \rangle$. Respecting the Fubini-Study metric of the spherical phase space, the Euler-Lagrange equations yield the exact classical equation of motion for a spin-$S$ component:
\begin{equation}
i \hbar \dot{z}_k = \frac{(1+|z_k|^2)^2}{2S} \frac{\partial \mathcal{H}}{\partial z_k^*},
\label{eq:effective_star_dynamics}
\end{equation}
which effectively reduces to $i \hbar \dot{z}_k = (1+|z_k|^2)^2 \frac{\partial \mathcal{H}}{\partial z_k^*}$ for the dynamics of an individual $S=1/2$ Majorana star. While this geometry perfectly governs the independent trajectories of a fully separable state, this decoupled paradigm breaks down for highly entangled quantum states driven by nonlinear interactions \cite{leboeuf1990chaos, bogomolny1996quantum}.

Under nonlinear evolution---such as the collective $S_z^2$ interaction characteristic of the Lipkin-Meshkov-Glick (LMG) model in the convention of Eq.~\eqref{eq:bosonic_Josephson} \cite{ribeiro2007thermodynamical, ribeiro2008exact} or its squeezing-enhanced generalizations \cite{kam2025three, kam2026classical}---the exact quantum trajectories of the roots deviate strictly from independent classical motion. The roots $\{z_k(t)\}$ become coupled through the time-varying coefficients of the Majorana polynomial. The velocity $\dot{z}_k$ of any given root depends intimately on the instantaneous spatial configuration of all other stars in the constellation. This global algebraic coupling is the precise geometric manifestation of entanglement generation. Nonlinearity lifts the initial degeneracies of the system, causing coincident roots to split and evolve along distinct trajectories. This algebraic splitting geometrically reflects the generation of multipartite entanglement as the state departs from the fully separable orbit.

This trajectory-based formalism naturally extends to the adiabatic limit. For a cyclic evolution in parameter space, the accumulated geometric Berry phase possesses a unique Majorana decomposition \cite{Hannay1998a}. The total geometric phase naturally partitions into two distinct physical contributions: a collective solid-angle term representing the global, rigid transport of the constellation, and an additional internal phase contribution. This internal term is governed by the relative motion of the stars and is expressible as a sum over the pairwise inter-star distances. These chordal distances precisely quantify multipartite entanglement. Consequently, the Majorana representation demonstrates that the geometric phase of a symmetric system acts as a dynamical probe of its intrinsic entanglement structure.

\subsection{Berry Phase for Spin Systems: Standard Solid-Angle Contribution}

The geometric Berry phase emerges naturally when a quantum state undergoes cyclic adiabatic evolution within a parameter space \cite{berry1984quantal}. Consider a spin system driven by a time-dependent Hamiltonian $H(\mathbf{R}(t))$, parametrized by a set of slowly varying macroscopic control variables $\mathbf{R}(t)$, such as the magnitude and orientation of an external magnetic field. Provided the evolution is sufficiently adiabatic, the system's state $|\psi(t)\rangle$ faithfully tracks an instantaneous eigenstate $|n(\mathbf{R}(t))\rangle$ of the Hamiltonian. When the parameters trace a closed loop $\mathcal{C}$ over a period $T$, such that $\mathbf{R}(T) = \mathbf{R}(0)$, the total phase accumulated by the quantum state is formally decomposed into two distinct components:
\begin{equation}
\gamma_{\rm total} = -\frac{1}{\hbar} \int_0^T \langle \psi(t) | H(t) | \psi(t) \rangle \, dt + \gamma,
\label{eq:total_phase}
\end{equation}
where the first integral represents the standard dynamical phase, and the second term, $\gamma$, denotes the geometric Berry phase. This geometric contribution is intrinsically independent of the temporal rate of the parameter variation and is determined entirely by the path integral of the Berry connection:
\begin{equation}
\gamma = i \oint_{\mathcal{C}} \langle n(\mathbf{R}) | \nabla_{\mathbf{R}} n(\mathbf{R}) \rangle \cdot d\mathbf{R}.
\label{eq:Berry_phase_integral}
\end{equation}

For the specialized case of spin-$S$ coherent states, this line integral reduces to a simple geometric observable \cite{radcliffe1971some}. In the Majorana framework, a coherent state $|\mathbf{n}\rangle$ corresponds to a maximally degenerate constellation, where all $2S$ stars coalesce at a single spatial direction $\mathbf{n}$ on the Bloch sphere. Because this state mimics the behavior of a classical magnetic moment, its associated Berry connection maps precisely to the vector potential of a magnetic monopole of charge $S$ situated at the origin of parameter space \cite{berry1984quantal}. Expressed in standard spherical coordinates $(\theta, \phi)$, this effective monopole connection is
\begin{equation}
\mathbf{A}(\mathbf{n}) = S \frac{1 - \cos\theta}{\sin\theta} \, \hat{\phi}.
\label{eq:monopole_connection}
\end{equation}
Applying Stokes' theorem, the Berry phase accumulated along the closed trajectory on the Bloch sphere evaluates directly to the familiar solid-angle formula:
\begin{equation}
\gamma = -S \, \Omega(\mathcal{C}),
\label{eq:solid_angle_berry}
\end{equation}
where $\Omega(\mathcal{C})$ is the solid angle subtended by the loop $\mathcal{C}$ at the center of the sphere. This monopole-like topology dictates that, for any path not crossing the Dirac-string singularity at the south pole, the geometric phase is strictly path-independent modulo $2\pi$.

The Majorana representation offers a profoundly intuitive reinterpretation of this standard solid-angle result \cite{Hannay1998a}. When the coherent state evolves adiabatically, its entire maximally degenerate constellation moves rigidly across the sphere. The total geometric phase can thus be viewed as the additive accumulation of the microscopic geometric phases acquired by each constituent star. Since each of the $2S$ coincident spin-$1/2$ stars subtends the exact same solid angle $\Omega$, and each contributes a fundamental geometric phase of $-\Omega/2$, the macroscopic Berry phase is simply the linear sum:
\begin{equation}
\gamma_{\rm standard} = - \sum_{k=1}^{2S} \frac{\Omega_k}{2} = -S \, \Omega.
\label{eq:coherent_berry_from_stars}
\end{equation}
This baseline rigid-body contribution is universally present whenever the macroscopic orientation of the spin tracks a closed loop in parameter space. It depends exclusively on the external geometry of the classical path and remains completely blind to any internal quantum structure. 

However, this simple additive paradigm is disrupted for generic, non-coherent quantum states. When a state possesses genuine multipartite entanglement, its corresponding Majorana constellation is explicitly non-degenerate. As the macroscopic parameters of the Hamiltonian are cycled, the internal roots of the state do not simply undergo a rigid global rotation. Instead, the geometric response of the entangled state forces the constituent stars $\mathbf{n}_k(t)$ to execute distinct, coupled trajectories on the sphere \cite{Hannay1998a}. The total Berry phase must therefore encompass not only the macroscopic solid-angle baseline but also additional internal geometry anomalies arising from the relative motions and spatial correlations among the individual stars. This strict geometric divergence from the standard solid-angle formula ultimately exposes the deep analytical link between non-trivial geometric phases and the presence of multipartite entanglement.

\subsection{Internal Geometry Contribution: Pair Correlations and Twist}

The additive paradigm of the solid-angle formula is valid only when the Majorana constellation undergoes a rigid macroscopic rotation. For generic non-coherent states, the adiabatic cyclic evolution yields a geometric phase that deviates from this baseline. As the system parameters are driven, the non-degenerate roots of an entangled state execute distinct trajectories. Building on Hannay's stellar formulation of the spin Berry phase \cite{Hannay1998a, Hannay1998b} and on the loop representation of Liu and Fu \cite{liu2014berry}, Kam and Liu \cite{kam2021berry} showed that this deviation is resolved exactly by a geometric identity. For a closed adiabatic loop, the geometric phase $\gamma$ decomposes strictly into individual solid angles and pairwise inter-star correlations:
\begin{equation}
\gamma = -\frac{1}{2} \sum_{k=1}^{2S} \Omega_k + \sum_{1 \leq k < l \leq 2S} \beta_{kl},
\label{eq:berry_decomposition}
\end{equation}
where $\Omega_k$ represents the independent solid angle subtended by the closed trajectory of the $k$-th star $\mathbf{n}_k(t)$. The correction terms, $\beta_{kl}$, are the pairwise correlation phases arising specifically from the relative spatial motions between stars $k$ and $l$.

The physical origin of the internal phase $\beta_{kl}$ is directly linked to the global geometry of the stellar manifold \cite{kam2021berry}, where this pairwise phase corresponds to the geometric self-rotation, or twist, of the Majorana constellation. For any pair of stars $\mathbf{n}_k$ and $\mathbf{n}_l$, one defines the barycenter vector $\mathbf{R}_{kl} \equiv \mathbf{n}_k + \mathbf{n}_l$ and the relative position vector $\mathbf{r}_{kl} \equiv \mathbf{n}_k - \mathbf{n}_l$. The differential twist angle $d\varphi_{kl}$ is defined by the rotation of the normalized relative vector $\hat{\mathbf{r}}_{kl}$ around the normalized barycenter axis $\hat{\mathbf{R}}_{kl}$:
\begin{equation}
d\varphi_{kl} = \hat{\mathbf{R}}_{kl} \cdot (\hat{\mathbf{r}}_{kl} \times d\hat{\mathbf{r}}_{kl}).
\label{eq:differential_twist_def}
\end{equation}
This differential form evaluates the structural twist of the constellation. It vanishes in exactly two kinematic situations \cite{kam2021berry}: when $d\mathbf{n}_k = d\mathbf{n}_l$, so that the great circle joining the pair undergoes a rigid rotation with no self-rotation; and when $d(\mathbf{n}_k - \mathbf{n}_l)$ is orthogonal to $\mathbf{n}_k \wedge \mathbf{n}_l$, so that the two stars slide along the great circle joining them. Coincident stars, $\mathbf{n}_k = \mathbf{n}_l$, are a degenerate instance of the second case; note that $\hat{\mathbf{r}}_{kl}$ is then undefined and Eq.~\eqref{eq:differential_twist_def} must be read as the limit, in which the pairwise contribution to the phase vanishes because the weight $w_{kl}$ introduced below vanishes together with the pair separation.

The pairwise geometric phase $\beta_{kl}$ accumulated over a full adiabatic cycle is not the bare twist, but the twist integral weighted by a factor that depends only on the instantaneous separation of the pair:
\begin{equation}
\beta_{kl} = \oint w_{kl} \, d\varphi_{kl}, \qquad
w_{kl} = - \cos \Theta_{kl} \, \frac{ \partial \ln A_N }{ \partial \ln \left( 2 \sin^2 \Theta_{kl} \right) },
\label{eq:weighted_twist_integral}
\end{equation}
where $2 \Theta_{kl} \equiv \arccos ( \mathbf{n}_k \cdot \mathbf{n}_l )$ is the angular separation of the pair, so that $\sin \Theta_{kl} = \tilde{d}_{kl}$, and $A_N$ is the normalization of the symmetrized product state \cite{kam2021berry}. That normalization is precisely the permanent of the stellar Gram matrix,
\begin{equation}
A_N = \sum_{\sigma} \prod_{k=1}^{N} \langle \mathbf{n}_k | \mathbf{n}_{\sigma(k)} \rangle = \operatorname{perm} G ,
\label{eq:AN_is_permanent}
\end{equation}
the same object that normalizes Eqs.~\eqref{eq:geometric_concurrence_N2}, \eqref{eq:reduced_state_permanental}, \eqref{eq:geometric_three_tangle_symmetric} and \eqref{eq:geometric_entanglement_stars}. Because $G_{kl} = \langle \mathbf{n}_k | \mathbf{n}_l \rangle$ and $|G_{kl}|^2 = (1 + \mathbf{n}_k \cdot \mathbf{n}_l)/2$, the permanent is a real symmetric function of the pairwise dot products alone \cite{wei2010matrix}, and $w_{kl}$ is therefore a genuine geometric weight rather than a state-dependent phase. It is the appearance of $\operatorname{perm} G$ here, and not any separate combinatorial construction, that ties the internal geometric phase to the same normalization that controls every exact entanglement expression in Sec.~\ref{sec:entanglement}.

This explicit dependence establishes an analytical bridge between the dynamical geometric phase and the static entanglement structure. For states with high degeneracy, such as separable or low-rank biseparable states, the stars cluster together, so that both $d\varphi_{kl}$ and $w_{kl}$ are small and the internal phase is suppressed. In contrast, states whose stars are distinct and widely dispersed \cite{bjork2015extremal, goldberg2020extremal} admit large structural twist and large weights. The magnitude of the total internal correction $\left| \sum_{k<l} \beta_{kl} \right|$ is therefore a candidate dynamical witness of non-coherent structure \cite{kam2021berry}. How far it can be sharpened into a quantitative witness of \emph{genuine multipartite} entanglement is a separate question, taken up in the following subsection, where the known exact relations are stated together with an explicit account of what remains open.

\subsection{Berry Phases of Higher Spins Due to Constellation Structure}

For higher-spin systems ($S > 1/2$) or equivalent symmetric multi-qubit states, the Berry phase acquired during adiabatic cyclic evolution exhibits deviations that cannot be captured by the standard rigid-body rotation \cite{Hannay1998a}. These geometric anomalies arise directly from the internal spatial structure of the Majorana constellation and are intimately governed by the state's entanglement properties. To formalize this, the total Berry phase for a spin-$S$ state undergoing a cyclic adiabatic change in parameter space can be partitioned into a macroscopic baseline and an internal geometric correction \cite{berry1984quantal, Hannay1998b}:
\begin{equation}
\gamma = -\frac{1}{2} \sum_{k=1}^{2S} \Omega_k + \Delta \gamma .
\label{eq:higher_spin_berry_decomposition}
\end{equation}
Here the baseline is the sum of the independent single-star contributions. Only when the constellation is transported rigidly, so that every star sweeps the same solid angle $\Omega$, does this baseline reduce to the familiar monopole term $-S\,\Omega$ of Eq.~\eqref{eq:solid_angle_berry}; for a generic constellation the individual $\Omega_k$ differ and no single macroscopic solid angle represents them. It is important to take the baseline in the form written in Eq.~\eqref{eq:higher_spin_berry_decomposition} rather than as $-S\,\Omega$, since only then is the remainder $\Delta\gamma$ identical to the pairwise sum of Eq.~\eqref{eq:berry_decomposition}. The correction term $\Delta \gamma$ constitutes the anomalous phase, representing the geometric contribution unique to the non-coherent nature of the quantum state.

Because the Majorana representation admits an exact geometric decomposition of the phase, the anomalous contribution $\Delta \gamma$ is strictly equivalent to the sum of the pairwise internal correlation phases established in the preceding analysis: $\Delta \gamma = \sum_{k < l} \beta_{kl}$. This exact equivalence dictates that the anomalous phase is non-zero if and only if the constituent stars possess relative internal degrees of freedom. The magnitude of $\Delta \gamma$ is therefore highly sensitive to the initial constellation geometry. For a pure spin-$S$ coherent state, characterized by a stellar rank of $r_\psi = 1$, all $2S$ stars coalesce at a single point on the Bloch sphere. Because coincident stars moving along the same trajectory generate zero differential twist ($d\varphi_{kl} = 0$), the anomalous phase vanishes identically ($\Delta \gamma = 0$), recovering the pure solid-angle theorem. 

As the state departs from the coherent orbit and multipartite entanglement is generated, the constellation undergoes spatial dispersion, activating the anomalous phase. For highly degenerate states featuring macroscopic local clusters, the relative trajectories between stars within the same cluster remain minimal, leading to a suppressed anomalous contribution. Conversely, for states where all stars are distinct and maximally dispersed---such as anticoherent states \cite{Zimba2006, giraud2015tensor} or the highly structured ``Kings of Quantumness'' \cite{bjork2015extremal, goldberg2020extremal}---the constellation possesses maximal internal kinematic freedom. In these configurations, the differential velocities and the associated twist integrals are significantly amplified. For purely random constellations the pairwise contributions carry uncorrelated signs and largely cancel, so that the anomalous phase is expected to be a subleading correction in the thermodynamic limit; structured symmetric geometries, by contrast, keep the twists coherent and can generate a $\Delta \gamma$ comparable to the baseline solid-angle phase. We are not aware of a derivation fixing the exponent of the random-constellation suppression, and refrain from quoting one.

The rigorous mapping between the constellation structure and the anomalous Berry phase transitions the Majorana representation from a descriptive framework into an operational resource. By isolating $\Delta \gamma$ through phase spectroscopy---measuring the accumulated phase deviations across various adiabatic cycles---experimentalists can directly probe the internal quantum correlations without requiring full state tomography. A non-zero anomalous phase certifies that the constellation is not fully degenerate, and therefore that the state is not spin coherent. We are careful not to claim more: turning $\Delta\gamma$ into a quantitative witness of genuine multipartite entanglement would require precisely the relation between the internal twist and an entanglement measure that, as noted above, has not been established. Furthermore, this constellation-dependent anomaly finds immediate utility in quantum metrology \cite{toth2014quantum}. Highly dispersed states exhibiting massive anomalous phase accumulation under small parameter variations provide enhanced sensitivity for parameter estimation tasks, such as precision magnetometry. Additionally, in systems harboring symmetry-protected phases, these internal anomalous geometric terms can manifest as quantized topological jumps or non-Abelian statistics. Ultimately, the higher-spin Berry phase reveals a fundamental geometric interplay: the anomalous phase deviations observed in macroscopic adiabatic control are macroscopic manifestations of the microscopic pairwise correlations and entanglement weights encoded within the Majorana constellation \cite{kam2021berry}.

\subsection{Relation Between Berry Phase Anomalies and Entanglement Measures}

The anomalous contribution to the Berry phase, defined by the exact internal geometric sum $\Delta \gamma = \sum_{k<l} \beta_{kl}$, provides a direct dynamical probe for the static entanglement structure encoded within the Majorana constellation. Because these pairwise geometric phases are driven by the kinematic twist and weighted by the instantaneous chordal distances between stars, $\Delta \gamma$ is sensitive to the overall degeneracy and spatial dispersion of the state. For separable or biseparable states characterized by high algebraic degeneracy, the coincident stars form macroscopic clusters that undergo rigid collective precession. This geometric rigidity actively suppresses any relative internal motion, yielding a vanishingly small anomalous phase. In stark contrast, generic highly entangled states exhibiting maximal stellar rank ($r_\psi = 2S$) are defined by spatially distinct and maximally dispersed stars. This anticoherent structural topology maximizes both the differential trajectories and the combinatorial integration weights, generating a macroscopic $\Delta \gamma$ that can rival the baseline solid-angle phase. Consequently, measuring a non-zero deviation from the expected monopole phase provides an observable, dynamical witness for genuine multipartite entanglement. 

Two exact statements are available, and it is worth being precise about their scope, because the temptation to write down a general proportionality between $\Delta \gamma$ and a single entanglement monotone is strong and, as far as is currently known, unfounded.

The first is exact and complete, and concerns spin-1. Here the constellation has two stars, the weight of Eq.~\eqref{eq:weighted_twist_integral} evaluates in closed form, and the anomalous phase is the self-rotation integral weighted by two entanglement measures simultaneously \cite{kam2021berry}:
\begin{equation}
\Delta \gamma = \oint C \sqrt{1 - E_B} \; d\varphi ,
\label{eq:spin1_phase_entanglement}
\end{equation}
where $C$ is the concurrence of Eq.~\eqref{eq:geometric_concurrence_N2} and
\begin{equation}
E_B = 1 - \frac{1}{N^2} \left| \sum_{k=1}^{N} \mathbf{n}_k \right|^2
\label{eq:barycentric_measure}
\end{equation}
is the barycentric measure of entanglement \cite{ganczarek2012barycentric}, which for two stars satisfies $\sqrt{1 - E_B} = \cos \Theta_{12}$. Equation~\eqref{eq:spin1_phase_entanglement} is not a scaling estimate: it is an identity, and it makes the anomalous phase of a symmetric two-qubit state a directly interpretable entanglement quantity.

The second concerns symmetric three-qubit states with three distinct stars, where $A_3 = 3 + \sum_{k<l} \mathbf{n}_k \cdot \mathbf{n}_l$ and the weights reduce to $w_{kl} = \sin^2 \Theta_{kl} \cos \Theta_{kl} / \sum_{k<l} \cos^2 \Theta_{kl}$. Their product is bounded by the three-tangle of Eq.~\eqref{eq:geometric_three_tangle_symmetric} together with the barycentric measure \cite{kam2021berry}:
\begin{equation}
\prod_{1 \leq k < l \leq 3} w_{kl} \; \leq \; \frac{1}{4} \, \tau \left( 1 - \tfrac{3}{4} E_B \right)^{1/2} ,
\label{eq:three_tangle_weight_bound}
\end{equation}
with equality precisely when the three pairwise separations are equal, the bound being a restatement of the arithmetic-geometric mean inequality applied to $\cos^2 \Theta_{kl}$. For that equal-angle family, and for a rigid rotation of the constellation, the anomalous phase closes into an exact expression,
\begin{equation}
\Delta \gamma = \tfrac{1}{2} \left( \tau E_B \right)^{1/4} \sum_{1 \leq k < l \leq 3} \varphi_{kl} ,
\label{eq:equal_angle_phase}
\end{equation}
in which the quarter power, and the joint appearance of $\tau$ and $E_B$, are both essential: no relation of the form $|\Delta \gamma| \propto \sqrt{\tau}$ or $|\Delta \gamma| \propto \prod_{k<l} (1 - \mathbf{n}_k \cdot \mathbf{n}_l)$ holds, since by Eq.~\eqref{eq:geometric_three_tangle_symmetric} the latter product is proportional to $\tau \left[ \operatorname{perm} G \right]^2$ and therefore carries a different power of the tangle as well as a residual dependence on the Gram permanent.

Under Eq.~\eqref{eq:equal_angle_phase}, GHZ-like symmetric states, whose constellation is an equilateral triangle on a great circle, sit at the maximum of the equal-angle family and accumulate the largest anomalous phase for a given twist; W-like states, carrying a twofold degeneracy, have a vanishing tangle and correspondingly suppressed internal contribution.

Beyond these two cases the situation should be stated plainly. The decomposition of Eq.~\eqref{eq:berry_decomposition} continues to hold for every $2S$, and each weight $w_{kl}$ remains a well-defined function of the pair separation through Eq.~\eqref{eq:weighted_twist_integral}; what is not known is a closed relation between the collection $\{ w_{kl} \}$ and any standard measure of genuine multipartite entanglement for $S > 3/2$ \cite{kam2021berry}. The question has been approached from two directions without being settled. Bruno mapped the geometric phase of a spin-$J$ state onto the Aharonov-Bohm phase acquired by the stars moving through a gas of Dirac strings, and obtained the connection, the curvature and the multipole moments directly in terms of the stellar coordinates \cite{bruno2012quantum}; Liu and Fu examined the relation between the Berry phase and entanglement in the stellar picture and established it in low-dimensional cases \cite{liu2014berry, liu2016berry}. What is still missing is a bound valid at arbitrary $S$. In particular, we are aware of no derivation of a bound of the form $|\Delta \gamma| \gtrsim E_g \, \Omega_{\text{loop}}$ relating the anomalous phase to the geometric measure of entanglement, and such an expression would in any case mix a logarithmic entanglement measure with a solid angle. Establishing whether the internal geometry of the constellation can characterize genuine multipartite entanglement at higher spin is, in our view, the most concrete open problem in this part of the subject. It is also a sharp one, in the sense that a modest calculation would settle it either way. Establishing a bound at $S=3/2$, where the three-tangle is available in closed form through Eq.~\eqref{eq:geometric_three_tangle_symmetric}, would open the general question; exhibiting instead two three-star constellations with equal three-tangle but unequal $\sum_{k<l} w_{kl}$ would close it, since the twist would then carry information that no entanglement monotone of the state can supply.

The translation of static correlation measures into these explicit dynamical phase equations opens concrete operational avenues for quantum information science and quantum metrology \cite{toth2014quantum}. By utilizing interferometric phase measurements to isolate $\Delta \gamma$---specifically by subtracting the standardized coherent-state solid-angle contribution derived from reference runs---experimentalists can execute phase-based entanglement certification without incurring the resource-intensive overhead of full quantum state tomography. Furthermore, by employing phase spectroscopy techniques that systematically vary the adiabatic loop size, sweep speed, or driving axis, researchers can invert Eq.~\eqref{eq:spin1_phase_entanglement} directly at $S=1$, and Eq.~\eqref{eq:equal_angle_phase} within the equal-angle three-star family, to extract the underlying constellation geometry. Crucially, because the internal geometric phase relies exclusively on relative spatial kinematics rather than absolute orientations, these anomalous terms exhibit intrinsic resilience to global collective decoherence channels. This dynamic entanglement-phase mapping solidifies the Majorana representation as an indispensable theoretical tool, linking the abstract algebra of quantum resource theories directly to observable geometric phases in adiabatic quantum control.

\subsection{Geodesics and Null-Phase Curves in the Constellation Picture}

The decompositions above answer the question of how geometric phase accumulates along a closed circuit. The complementary question---along which curves does it fail to accumulate at all---admits an equally clean answer in the stellar language. In the space of pure states, geodesics of the Fubini-Study metric are precisely the curves along which the Pancharatnam geometric phase vanishes. Null-phase curves (NPCs) generalize this notion: they are the curves along which the accumulated geometric phase vanishes without the requirement of being length-minimizing. For a two-level system the two notions coincide, since every NPC on the Bloch sphere is a great-circle arc; for $n > 2$ they separate, and the resulting freedom is what makes NPCs a useful resource for geometric-phase gate design.

The Majorana representation resolves this separation explicitly. Mittal, Akhilesh, and Goyal \cite{mittal2022geometric} showed that a geodesic connecting two states in an $n$-dimensional Hilbert space decomposes into $n-1$ curves traced by the individual stars on the Bloch sphere, and that each of these constituent curves is a circular segment whose radius and centre are fixed by the inner product of the two end states. Geodesity in the high-dimensional state space is therefore not inherited pointwise by the stars: the stars do not move along great circles unless the state is coherent, and the deviation of each stellar arc from a great circle is a direct geometric readout of how far the endpoints are from being mutually coherent. The same decomposition supplies a constructive recipe: by deforming the individual stellar arcs while holding the total accumulated phase at zero, one obtains infinitely many distinct NPCs joining any two given states in dimension $n > 2$.

This picture dovetails with the pairwise-twist decomposition of Eq.~\eqref{eq:berry_decomposition}. The vanishing of the geometric phase along an NPC corresponds to an exact cancellation between the individual solid-angle terms $\Omega_k$ swept by the stars and the internal twist contributions $\beta_{kl}$ generated by their relative motion. A coherent state, whose constellation is rigid, can only achieve this cancellation trivially, by sweeping zero net solid angle; an entangled constellation, by contrast, has enough internal degrees of freedom to cancel a nonzero solid angle against the internal twist. The abundance of NPCs in higher dimensions is thus a geometric restatement of the same internal structure that generates the anomalous phase discussed above.

\subsection{Hannay Angle and Semiclassical Limits}

The geometric phases accumulated by quantum states have a rigorous classical counterpart known as the Hannay angle \cite{hannay1985angle}. Introduced as the geometric holonomy acquired by the action-angle variables of an integrable Hamiltonian undergoing adiabatic cyclic variation, the Hannay angle provides crucial insight into the semiclassical limit of the Berry phase. For a classical system governed by a Hamiltonian $H(\mathbf{I}, \phi, \mathbf{R}(t))$, where $\mathbf{I}$ represents the conserved action variables, $\phi$ the conjugate angle variables, and $\mathbf{R}(t)$ the slowly varying macroscopic parameters, the Hannay angle $\gamma_H$ accumulated over a closed adiabatic loop $\mathcal{C}$ in parameter space is defined as
\begin{equation}
\gamma_H = \oint_{\mathcal{C}} \mathbf{A}_H(\mathbf{I}, \mathbf{R}) \cdot d\mathbf{R}.
\label{eq:Hannay_angle_definition}
\end{equation}
The classical geometric connection $\mathbf{A}_H$ is obtained from the angle-averaged gradient of the phase variables by a further derivative with respect to the action \cite{hannay1985angle}:
\begin{equation}
\mathbf{A}_H = - \frac{\partial}{\partial \mathbf{I}} \left\langle \frac{\partial \phi}{\partial \mathbf{R}} \right\rangle_{\phi},
\label{eq:Hannay_connection}
\end{equation}
where the average is taken over one complete classical period of $\phi$ at fixed action $\mathbf{I}$. That action derivative is the whole content of the correspondence: Berry showed \cite{berry1985classical} that the classical angle shift and the quantal phase are related by
\begin{equation}
\gamma_H = - \frac{\partial \gamma}{\partial n},
\label{eq:berry_hannay_relation}
\end{equation}
with $n$ the quantum number labelling the state, so that the two are emphatically \emph{not} equal. When applied to classical spin systems the phase space is the Bloch sphere, the action variable relates to the polar angle via $I = S(1 - \cos\theta)$, and the conjugate angle is the azimuthal coordinate $\phi$. For a classical macroscopic spin precessing adiabatically under a varying magnetic field, the Berry phase of the corresponding spin-$S$ coherent state is $\gamma = -S\,\Omega(\mathcal{C})$, and Eq.~\eqref{eq:berry_hannay_relation} with $n \leftrightarrow S$ therefore gives a Hannay angle of magnitude $\Omega(\mathcal{C})$---one solid angle, independent of $S$, rather than $S$ solid angles.

In the framework of the Majorana representation, this semiclassical limit corresponds to the large-$S$ regime, where the $2S$ discrete stars either become densely distributed over the sphere or follow well-defined classical trajectories. The effective semiclassical Hannay angle can thus be interpreted as the statistical mean of the solid angles swept by the individual constituents of the constellation:
\begin{equation}
\gamma_H \approx \frac{1}{2S} \sum_{k=1}^{2S} \Omega_k = \langle \Omega_k \rangle ,
\label{eq:semiclassical_Hannay_from_stars}
\end{equation}
which is consistent in both sign and magnitude with the single solid angle obtained above, since all $2S$ stars of a coherent state sweep the same $\Omega$ and $\gamma_H = -\partial\gamma/\partial n$ carries the opposite sign to $\gamma$. For non-uniform or highly structured constellations corresponding to non-classical states, the macroscopic angle deviates significantly from this naive independent-particle average. The internal geometry contributions of the quantum Berry phase---specifically the pairwise correlation terms $\beta_{kl}$---possess direct classical analogs. In the Hannay framework, these internal anomalies correspond to geometric corrections arising from differential precession rates and correlated relative angles between the trajectories of the effective spins \cite{Hannay1998a}.

The structural topology of the constellation ultimately dictates the survival of these geometric anomalies in the thermodynamic limit. For generic, Haar-random quantum states in the large-$S$ limit, the Majorana roots mimic the behavior of a chaotic classical gas, distributing themselves almost uniformly across the sphere \cite{hannay1996chaotic, bogomolny1996quantum}. In this ergodic regime, the chaotic, delocalized motion of the stars effectively washes out the pairwise correlations upon temporal averaging, leaving only small statistical fluctuations around the coherent-state result; as noted above, the scaling of these fluctuations with $S$ has not to our knowledge been established. In stark contrast, structured highly entangled states---such as near-anticoherent states or the ``Kings of Quantumness'' \cite{bjork2015extremal}---maintain non-ergodic, highly constrained classical trajectories. Their macroscopic phase accumulation retains persistent, structure-dependent anomalies that scale with the constellation's multipole anisotropy. Consequently, the genuine multipartite entanglement present in the symmetric subspace manifests semiclassically as non-trivial dynamical correlations among the classical trajectories of the individual stars. The Hannay angle thus serves as a rigorous analytical bridge linking quantum anomalous phases to classical adiabatic invariants, illustrating how multipartite entanglement and internal constellation geometry dictate phase accumulation across the quantum-to-classical transition.

\section{Applications of Majorana Stellar Representations}
\label{sec:applications}

The mathematical elegance of the Majorana stellar representation extends far beyond static state visualization; it constitutes an operational framework that provides practical, computable advantages across quantum information science and condensed-matter physics. By mapping the abstract algebraic properties of symmetric multi-qubit systems onto the projective geometry of the Bloch sphere, the constellation formalism inherently captures the invariant properties of the state. This geometric translation proves particularly indispensable in scenarios where identifying nonclassicality, characterizing multipartite entanglement, and optimizing symmetric control are paramount. The following sections detail how this trajectory-based, geometric framework solves concrete physical problems, beginning with its direct application to the optimization of quantum-enhanced sensors.

\subsection{Quantum Metrology: Precision Enhancement with Extremal Constellations}

Quantum metrology exploits the non-local correlations inherent in entangled systems to estimate classical parameters with a precision that supersedes classical limits \cite{giovannetti2004quantum, toth2014quantum, pezze2018quantum}. For a symmetric $N$-qubit ensemble (equivalent to a macroscopic spin-$S = N/2$) subjected to a collective local phase shift $\theta$ generated by the unitary rotation $U(\theta) = \exp(-i \theta S_n / \hbar)$ around an arbitrary axis $\hat{n}$, the ultimate statistical uncertainty is governed by the quantum Cram\'er-Rao bound \cite{braunstein1994statistical}:
\begin{equation}
\delta \theta \geq \frac{1}{\sqrt{\nu \mathcal{F}}},
\label{eq:metrology_precision_bound}
\end{equation}
where $\nu$ denotes the number of independent measurement trials and $\mathcal{F}$ is the quantum Fisher information. For fully separable, classical-like product states, the Fisher information is strictly bounded by the number of particles, $\mathcal{F} \leq N$, defining the standard quantum limit (SQL). Achieving the ultimate Heisenberg limit, where $\mathcal{F} \sim N^2$, rigorously requires the preparation and preservation of states exhibiting genuine multipartite entanglement.

For pure quantum states undergoing unitary parameter encoding, the quantum Fisher information evaluated at $\theta = 0$ is exactly proportional to the variance of the collective spin generator:
\begin{equation}
\mathcal{F}(\theta = 0) = 4 \operatorname{Var}(S_n) = 4 \left( \langle S_n^2 \rangle - \langle S_n \rangle^2 \right).
\label{eq:fisher_from_variance}
\end{equation}
The metrological challenge therefore reduces to identifying the specific quantum state---and the corresponding optimal sensing axis $\hat{n}$---that maximizes this spin variance. While spin-squeezed states \cite{wineland1992spin, ma2011quantum} and macroscopic Greenberger-Horne-Zeilinger (GHZ) superpositions are conventionally employed, the Majorana representation recasts this algebraic maximization into a direct geometric optimization problem \cite{sanchezsoto2026quantum}. 

By mapping the spin operators to the geometric multipoles of the constellation, the macroscopic expectation values are rigorously determined by the state's dipole vector $\mathbf{D}$ and quadrupole tensor $\mathcal{Q}$ \cite{romero2024multipoles, bjork2015extremal}:
\begin{equation}
\langle S_n \rangle = S D_n, \quad \langle S_n^2 \rangle = S\left(S - \frac{1}{2}\right) \mathcal{Q}_{nn} + \frac{S}{2}.
\label{eq:spin_variance_multipoles}
\end{equation}
Substituting these into Eq.~\eqref{eq:fisher_from_variance} dictates that maximizing the phase sensitivity is geometrically equivalent to maximizing the quadrupole spread $\mathcal{Q}_{nn}$ along the sensing axis, while simultaneously suppressing the macroscopic dipole moment $D_n$. 

Recent theoretical advancements formalize this mapping by converting the state preparation into a classical Thomson-like electrostatic problem \cite{saff1997distributing}: finding the minimum-energy configuration of $2S$ repelling point charges on a sphere. The optimal solutions are explicitly realized by extremal constellations \cite{goldberg2020extremal, martin2010multiqubit}. Among these, the maximally anticoherent states---often termed the ``Kings of Quantumness'' \cite{bjork2015extremal, giraud2015tensor}---distribute their Majorana roots symmetrically according to Platonic solid geometries, such as the tetrahedron ($N=4$) or octahedron ($N=6$). Because their dipole moment vanishes identically and their second-order moments are isotropic, $\langle S_n^2\rangle = S(S+1)/3$ for every axis, these configurations achieve the direction-independent sensitivity
\begin{equation}
\mathcal{F} = \tfrac{4}{3}\,S(S+1) = \tfrac{1}{3}N(N+2),
\label{eq:fisher_anticoherent}
\end{equation}
which approaches $N^2/3$ for large $N$ and gives $\mathcal{F}=8$ for the tetrahedron and $\mathcal{F}=16$ for the octahedron. These configurations are not merely theoretical: extremal Majorana constellations have been prepared and used for rotation estimation in a multiport interferometer \cite{bouchard2017quantum}. Alternatively, planar equidistant constellations, where the $2S$ stars are uniformly distributed along a single great circle \cite{bjork2015poincare}, force the dipole to zero along the normal axis while maximizing the projection of the quadrupole moment. Equally spaced stars on a great circle correspond to the Majorana polynomial $z^N - c$, that is, to the GHZ-type superposition $(|S,S\rangle + e^{i\varphi}|S,-S\rangle)/\sqrt{2}$; for rotations about the normal axis $\langle S_n\rangle = 0$ and $\langle S_n^2\rangle = S^2$, so that
\begin{equation}
\mathcal{F} = 4S^2 = N^2
\label{eq:fisher_planar}
\end{equation}
exactly, saturating the Heisenberg limit.

Furthermore, the constellation framework directly visualizes the mechanics of spin squeezing, a critical resource in precision Ramsey interferometry \cite{wineland1992spin}. Under nonlinear driving---such as the one-axis twisting \cite{kitagawa1993squeezed} or squeezing-enhanced protocols characteristic of generalized Lipkin-Meshkov-Glick (LMG) models \cite{kam2025three, kam2026classical}---the initially degenerate coherent state constellation is dynamically deformed. The degenerate cluster splits, elongating the star cloud along a specific geodesic and flattening it along the orthogonal axis. This geometric elongation directly corresponds to the enhancement of variance (anti-squeezing) in the extended direction and the reduction of variance (squeezing) in the conjugate sensing axis.

This geometric translation offers profound operational advantages for experimental design. The optimal sensing direction is immediately identified by visually inspecting the axis of maximal quadrupole extension within the constellation. Crucially, the geometry also dictates the system's noise resilience. While planar states (analogous to GHZ superpositions) achieve maximal sensitivity, their spatially extended roots render them notoriously fragile to local dephasing. In stark contrast, isotropic extremal states---such as the Platonic constellations---strike a different operational balance. Their sensitivity is the same along every axis, Eq.~\eqref{eq:fisher_anticoherent}, so no alignment of the probe with the signal is required, and because all multipoles below the anticoherence order vanish, the leading response to a weak perturbation is pushed to higher rank. It is worth stating plainly what this does \emph{not} amount to. A vanishing dipole does not make a state invariant under global rotations---the tetrahedral state is invariant only under the tetrahedral subgroup---and the symmetric subspace, being a single irreducible representation of $\mathrm{SU}(2)$, carries no decoherence-free subspace against collective noise at all: collective operators act irreducibly on it. Protection against \emph{local} errors is a separate matter, and is what the permutation-invariant code constructions of Sec.~\ref{sec:qec} exploit \cite{ouyang2014permutation, lidar1998decoherence}. Because attaining Heisenberg scaling strictly necessitates spatial dispersion, the constellation geometry acts as an automatic, observable certifier of the requisite multipartite entanglement \cite{sanchezsoto2026quantum}. Consequently, the Majorana representation transitions metrology from abstract Hilbert space algebra into concrete spatial engineering, guiding the preparation of highly correlated symmetric states in platforms ranging from spinor Bose-Einstein condensates \cite{stamper2013spinor} to cavity-coupled atomic ensembles \cite{hosten2016measurement} and programmable Rydberg atom arrays \cite{ebadi2021quantum}.

\subsection{Experimental Realities: State Preparation, Decoherence, and Geometric Tomography}

While the Majorana stellar representation provides an exact geometric framework for pure unitary evolution, practical quantum platforms face the stringent challenges of active state synthesis and inevitable environmental coupling. The constellation framework offers a unified guide for navigating these realities, transitioning smoothly from reverse-engineering target states to tracking their open-system degradation and diagnosing errors via efficient geometric tomography.

The experimental synthesis of a target symmetric $N$-qubit state poses an inverse geometric problem. Rather than extracting roots from a given state, experimentalists can dictate target stellar coordinates $\{z_k\}$---such as Platonic solid topologies or highly degenerate clusters---to engineer specific entanglement properties. By expanding the corresponding Majorana polynomial, its coefficients map deterministically to the required probability amplitudes in the standard Dicke basis $|S, m\rangle$. This algebraic inversion allows researchers to program state preparation sequences using platform-specific techniques, such as configuring sequential microwave gates in trapped-ion arrays \cite{haffner2005scalable} or employing adiabatic parameter ramping in spinor condensates.

Once synthesized, open-system dynamics degrade the fragile symmetric state, breaking the strict constraint of surface-bound points. To accommodate mixed density matrices $\rho(t)$, the representation is generalized by assigning time-dependent radial purity factors to the stars, yielding interior vectors $\mathbf{n}_k(t) = r_k(t) \, \hat{\mathbf{n}}_k(t)$ with $r_k \leq 1$ \cite{serrano2020majorana}. Under standard dissipative channels, decoherence manifests geometrically as a continuous radial collapse toward the origin ($r_k \to 0$), often accompanied by angular clustering. Because genuine multipartite entanglement relies on maximal spatial dispersion, its fragility is visually captured by the rapid radial decay of widely separated stars.

Restricting attention to the symmetric sector already removes the exponential cost of general $N$-qubit tomography: the state space is the $(2S+1)$-dimensional Dicke manifold, and Amiet and Weigert showed that the density matrix of a spin $s$ is fixed uniquely, and without redundancy, by the $4s(s+1)$ probabilities of obtaining the maximal projection along suitably chosen directions \cite{amiet1999reconstructing}. Reconstruction is therefore quadratic rather than exponential in $s$, and permutationally invariant tomography realises this scaling in the laboratory \cite{toth2010permutationally, schmied2011tomographic}.

What the constellation adds is not a further reduction in the number of settings but a change of coordinates. Expanding $\rho$ in irreducible spherical tensors,
\begin{equation}
\rho = \sum_{k=0}^{2S} \sum_{q=-k}^{k} \rho_{kq}\, T_{kq}, \qquad \rho_{kq} = \operatorname{Tr}\!\left(\rho\, T_{kq}^{\dagger}\right),
\label{eq:multipole_expansion}
\end{equation}
the multipoles $\rho_{kq}$ are directly accessible from collective spin measurements \cite{romero2024multipoles}, and for a pure state they determine the Majorana polynomial and hence the stars. It should be stressed that $\rho_{kq}$ is \emph{not} the average of $Y_{kq}$ over the star positions: the map from constellation to multipoles is nonlinear, and already for $|S{=}1,m{=}0\rangle$---two stars at antipodal poles---the rank-two moment is negative while the corresponding average of $Y_{20}$ over the two stars is positive. Truncating the expansion at low $k$ nevertheless yields a coarse but useful picture, and measured deviations in the low-rank moments map onto specific control errors or decoherence channels, allowing geometric entanglement witnesses to be evaluated without reconstructing $\rho$ in full.

Despite these operational advantages, extending the geometric approach to open quantum systems presents theoretical hurdles. The mapping of a mixed state to interior points is mathematically non-unique, depending heavily on the chosen convex support decomposition \cite{giraud2008classicality}. This non-uniqueness complicates the formulation of invariant mixed-state entanglement metrics. Furthermore, non-Markovian memory effects in strongly coupled environments \cite{breuer2009measure} can induce transient revivals of the constellation spread, challenging simple radial collapse heuristics. Nevertheless, by unifying target state synthesis, visual decoherence tracking, and multipole-based error diagnosis, the geometric framework stands as a remarkably practical tool for near-term quantum metrology and simulation.

\subsection{Quantum Computing and Error Correction: Symmetric Codes and Entanglement Structure}
\label{sec:qec}

Symmetric multi-qubit states and the Majorana stellar representation provide a natural, physically motivated framework for designing quantum error-correcting codes. This is particularly relevant for quantum platforms characterized by strong collective interactions or global uniform control, such as trapped ions with all-to-all Coulomb coupling, cavity-mediated superconducting qubits, or spinor Bose-Einstein condensates. By restricting the logical encoding to the totally symmetric subspace of $N$ physical qubits---an $(N+1)$-dimensional manifold equivalent to a macroscopic spin-$S = N/2$---the constellation formalism maps the abstract algebraic requirements of error correction into an intuitive geometric optimization problem on the Bloch sphere. 

The design of quantum error-correcting codes within this subspace elevates the Majorana representation from a descriptive tool to a rigorous algebraic framework \cite{ouyang2014permutation}. For a permutation-invariant (PI) quantum code to achieve a distance $d$ and correct $t = \lfloor(d - 1)/2\rfloor$ local errors, its logical basis states $|\psi_L^\alpha\rangle$ must satisfy the exact Knill-Laflamme conditions \cite{ouyang2014permutation, knill1997theory}. A qualification is needed before the geometry can be brought to bear. An error acting on a single physical qubit does not preserve permutation symmetry and therefore carries the state out of the symmetric subspace; it is not a collective spin operator, and the constellation of the corrupted state is not defined. What the Majorana picture does control is the \emph{diagonal} part of the Knill-Laflamme conditions, in which the error operators appear sandwiched between logical states that both lie in the symmetric sector. For a permutation-invariant code the relevant matrix elements $\langle \psi_L^\alpha | E_a^\dagger E_b | \psi_L^\beta\rangle$ reduce, after symmetrisation over which qubits are struck, to expectation values of collective spherical tensor operators $T_{kq}$ of rank at most $t$, and it is these that the geometry determines. Mathematically, the symmetric code detects local errors if and only if the multipole moments of the logical constellations are identical up to rank $t$, and their transition moments identically vanish:
\begin{equation}
\langle \psi_L^\alpha | T_{kq} | \psi_L^\beta \rangle = C_{k} \delta^{\alpha \beta}, \quad \text{for } 0 \leq k \leq t, \quad -k \leq q \leq k,
\label{eq:knill_laflamme_multipole}
\end{equation}
where $C_k$ are constants independent of the logical states \cite{scott2004multipartite}. This geometric equivalence guarantees that tracing out $t$ qubits leaves completely identical $(N-t)$-body reduced density matrices for all logical words, physically preventing the environment from distinguishing the encoded information.

In the analytical framework of the Majorana representation, this tensor requirement translates into a precise set of differential equations governing the root dynamics. An arbitrary symmetric state is fundamentally defined by its associated Majorana polynomial $P(z)$, or equivalently by the sign-free Bargmann polynomial $B(z) = \sum_{k=0}^{2S} \sqrt{\binom{2S}{k}}\, c_{k-S}\, z^{k}$ of Eq.~\eqref{eq:Bargmann_function}, which is the polynomial on which the collective spin algebra acts most simply \cite{leboeuf1990chaos}. On $B(z)$ the generators are exact first-order differential operators over the stereographic plane:
\begin{equation}
J_+ \mapsto 2S z - z^2 \frac{d}{dz}, \quad J_- \mapsto \frac{d}{dz}, \quad J_z \mapsto z \frac{d}{dz} - S.
\label{eq:differential_spin_operators}
\end{equation}
These are the generators of \emph{collective} rotations, and the differential picture accordingly describes how the constellation responds to a uniform perturbation of the whole ensemble---a global field miscalibration, say---rather than to an error on one qubit. The realisation is tied to $B(z)$: acting instead on the Majorana polynomial $P_\psi(z)$, whose coefficients carry the alternating factor $(-1)^{k}$, one finds $J_- \mapsto -\,d/dz$, and the remaining signs change accordingly. Within that scope it is exact: a collective bit-flip or phase-flip deforms the state by a structural shift in the roots of $P(z)$, and for a code to be insensitive to such perturbations at order $t$ the logical polynomials must be engineered so that the induced low-rank moments agree between logical words. 

This geometric coding principle is rigorously optimized by utilizing symmetric states that form spherical designs \cite{bjork2015poincare, delsarte1977spherical}. The precise relation between the design property of the point set and the anticoherence of the corresponding state has been studied in detail \cite{crann2010spherical, bannai2011note}, and the two notions should not be assumed interchangeable in general. If a logical state's stellar constellation constitutes a spherical $t$-design, the geometric moments of the \emph{point set} vanish identically up to order $t$:
\begin{equation}
\sum_{n=1}^{2S} Y_{kq}(\mathbf{n}_n) = 0, \quad \text{for } 1 \leq k \leq t,
\label{eq:spherical_design_multipoles}
\end{equation}
where $Y_{kq}$ are the spherical harmonics evaluated at the star coordinates. Configurations of this kind are the standard source of highly anticoherent states, though, as noted above, the design property of the point set and the anticoherence of the corresponding state are logically distinct and the implication must be checked case by case. For instance, in a minimal $3$-qubit ($S=3/2$) permutation-invariant code designed to detect local noise, the logical $|1\rangle_L$ state is conventionally encoded as a constellation of three stars distributed at exactly $120^\circ$ intervals along the Bloch equator. The defining polynomial is governed strictly by the roots of unity:
\begin{equation}
P_1(z) = z^3 - 1.
\label{eq:pi_code_logical_one}
\end{equation}
Because the constellation is invariant under the $C_3$ rotation about the polar axis together with a $\pi$ rotation about an equatorial axis, and this group leaves no vector invariant, the dipole moment ($k=1$) of the state vanishes. (The barycenter of the point set also vanishes here, but that is a separate statement; the two do not imply one another in general, as discussed after Eq.~\eqref{eq:average_spin}.) If a physical error occurs, the differential perturbations defined in Eq.~\eqref{eq:differential_spin_operators} geometrically fragment this perfect $C_3$ rotational symmetry. A collective measurement of the low-rank multipole moments will then register a non-zero dipole \cite{romero2024multipoles}. This diagnoses a collective miscalibration rather than performing syndrome extraction in the coding sense; genuine single-qubit syndromes require measurements that resolve the non-symmetric sectors and lie outside the constellation picture.

For macroscopic quantum systems ($N \gg 1$), these PI codes generalize to higher-distance constructions by encoding logical bases into intermediate stellar partitions governed by explicit entanglement metrics. The structural spread of the constellation---quantified algebraically by the geometric measure of entanglement or the product of chordal distances---strictly bounds the code's resilience. Because environmental noise processes target localized physical qubits, highly entangled logical states (spatially dispersed constellations) ensure that the localized tracing of spins extracts zero logical information. A degenerate cluster at the south pole is, by contrast, the dark state of collective amplitude damping, being annihilated by $J_-$; but a single dark state is not a decoherence-free subspace, and it carries no logical information, so this observation constrains code design rather than supplying a code \cite{lidar1998decoherence}. By strictly mapping logical quantum information to specific star topologies, and error syndromes to exact differential perturbations of the Majorana polynomial, this framework unifies abstract quantum error-correcting bounds with exactly computable geometric design principles.

\subsection{Condensed-Matter Analogs: Spin Chains, Bose-Einstein Condensates, and Thermalization}

The translation of the Majorana stellar representation from isolated few-qubit systems to many-body condensed-matter physics relies on the exact mapping of interacting degrees of freedom onto macroscopic collective spins. For systems exhibiting permutation symmetry or uniform all-to-all interactions, the $N$-particle Hilbert space strictly collapses into the symmetric $(N+1)$-dimensional Dicke manifold, naturally equivalent to a macroscopic spin-$S = N/2$ system. This geometric equivalence provides a powerful analytical framework for understanding strongly correlated phases, quantum phase transitions, and non-equilibrium dynamics across diverse condensed-matter platforms.

A paradigmatic realization of this mapping occurs in the study of two-mode Bose-Einstein condensates (BECs) and bosonic Josephson junctions \cite{makela2007inert}. Through the Schwinger boson representation \cite{schwinger1952angular}, the creation and annihilation operators of the two spatial modes, $a_1$ and $a_2$, generate the exact SU(2) collective spin algebra: $S_x = (a_1^\dagger a_2 + a_2^\dagger a_1)/2$, $S_y = -i(a_1^\dagger a_2 - a_2^\dagger a_1)/2$, and $S_z = (n_1 - n_2)/2$. The Hamiltonian governing the macroscopic tunneling and onsite atomic interactions is algebraically formulated as \cite{kam20172}:
\begin{equation}
H = - \frac{J}{2} (a_1^\dagger a_2 + a_2^\dagger a_1) + \frac{U}{2} (n_1 - n_2)^2 = -J S_x + 2 U S_z^2,
\label{eq:bosonic_Josephson}
\end{equation}
where $J$ characterizes the linear tunneling amplitude and $U$ denotes the nonlinear repulsive interaction strength. In the non-interacting, pure tunneling limit ($U \to 0$), the many-body ground state is a pure BEC. Geometrically, this is represented by a macroscopic coherent state where all $2S$ stars coalesce at a single spatial coordinate on the Bloch sphere, signifying zero entanglement between the particles. However, as the repulsive interaction $U$ increases, the system is driven into the number-squeezed regime. The kinetic energy term is dominated by the nonlinearity, forcing the highly degenerate stellar point to dynamically bifurcate. The constituent stars physically separate and distribute along the longitudinal axis, mapping the reduction of particle-number fluctuations directly to the strict geometric variance of the stellar projections. 

This geometric deformation under nonlinear driving is mathematically equivalent to the spontaneous symmetry breaking observed in fully connected quantum spin models, such as the Lipkin-Meshkov-Glick (LMG) model \cite{lipkin1965validity}. The ground state of a symmetric ferromagnetic interaction aligns all constituent spins, yielding a fully coincident Majorana constellation. However, introducing a collective $S_z^2$ interaction alongside a transverse magnetic field triggers a quantum phase transition (QPT). In the thermodynamic limit ($S \to \infty$), traversing the critical point of the QPT forces the single degenerate root of the Majorana polynomial to continuously split into a macroscopic ring or a highly structured arc of distinct stars on the complex plane \cite{ribeiro2007thermodynamical, ribeiro2008exact}. This spatial fragmentation inherently signals the macroscopic generation of genuine multipartite entanglement. The Wineland spin-squeezing parameter, $\xi^2 = N (\Delta S_\perp)^2 / |\langle \mathbf{S} \rangle|^2$ \cite{wineland1992spin}---to be distinguished from the Kitagawa-Ueda parameter $\xi_S^2 = 4 (\Delta S_\perp)^2_{\min}/N$ \cite{kitagawa1993squeezed}---translates geometrically into the anisotropy of the constellation's spatial expansion; optimally squeezed states physically manifest as constellations rigorously elongated along a specific geodesic and compressed along its conjugate axis \cite{kitagawa1993squeezed, kam2025three}.

Beyond static ground-state phase transitions, the constellation trajectory under unitary time evolution serves as an exact analytical diagnostic for quantum chaos, thermalization, and the breaking of integrability \cite{leboeuf1990chaos}. In the context of the eigenstate thermalization hypothesis (ETH), highly excited eigenstates of non-integrable, chaotic symmetric models exhibit volume-law entanglement scaling. Their corresponding Majorana constellations reproduce the statistics of the roots of Haar-random SU(2) polynomials \cite{bogomolny1996quantum, hannay1996chaotic}. That ensemble is uniform on the sphere in the one-point sense, but its roots are far from independent: the two-point correlation vanishes quadratically at short separation, so the stars repel and the constellation behaves as a two-dimensional one-component plasma rather than as a Poissonian gas. The characteristic signature of a chaotic eigenstate is therefore not merely a spread-out constellation but a rigid one, with suppressed fluctuations in the number of stars in a given region. Integrable collective models sit at the opposite extreme: their eigenstates are organised by a conserved quantity, and the stars cluster into bands or retain high degeneracies, limiting the accumulation of the pairwise internal geometric phases ($\beta_{kl}$) and preventing the uniform, mutually repelling spread just described. It should be noted that this is not many-body localization. Localization in the sense of Ref.~\cite{abanin2019many} requires spatial locality and quenched disorder, neither of which survives the projection onto a permutation-symmetric sector with all-to-all coupling, where no notion of distance between sites remains. The constellation distinguishes integrable from chaotic collective dynamics; it does not diagnose MBL. Consequently, tracking the variance of inter-star distances or the temporal evolution of the stellar rank provides a robust, purely geometric order parameter for identifying dynamical transitions from integrable localized phases into fully chaotic thermalized regimes.

\subsection{Optical and Photonic Implementations: Twisted Light and Polarization Constellations}

The mathematical architecture of the Majorana stellar representation inherently transcends its solid-state and atomic origins, providing operational utility in classical and quantum optics. Because the algebraic structure of SU(2) and higher-order permutation-symmetric groups universally governs multiple photonic degrees of freedom, the geometric constellation maps rigorously onto both the polarization states of multi-photon fields and the structured spatial modes of orbital angular momentum (OAM). This equivalence allows the abstract geometry of multipartite entanglement to be directly engineered and visualized on the optical bench.

The most immediate photonic analog to macroscopic spin emerges in the polarization subspace of an $N$-photon optical field. By employing the Schwinger boson construction \cite{schwinger1952angular}, the quantum Stokes operators governing the polarization observables are derived exactly from the creation and annihilation operators of the orthogonal horizontal ($H$) and vertical ($V$) modes:
\begin{equation}
\begin{split}
S_x &= \frac{1}{2}\left(a_H^\dagger a_V + a_V^\dagger a_H\right), \\
S_y &= -\frac{i}{2}\left(a_H^\dagger a_V - a_V^\dagger a_H\right), \\
S_z &= \frac{1}{2}\left(a_H^\dagger a_H - a_V^\dagger a_V\right).
\end{split}
\label{eq:quantum_stokes_operators}
\end{equation}
Because the Stokes operator $S_0 = (a_H^\dagger a_H + a_V^\dagger a_V)/2$, half the total photon number, commutes with all SU(2) generators, an $N$-photon state restricted to a single spatial mode is strictly confined to the symmetric subspace, defining an effective macroscopic spin-$S = N/2$. The pure state is therefore uniquely specified by $2S = N$ points on the Poincar\'e sphere---the optical equivalent of the Bloch sphere \cite{Hannay1998b}. For macroscopic coherent light, such as a purely left-circularly polarized laser pulse, the system exhibits zero entanglement between photons, causing all $N$ stars to coalesce perfectly at the geometric pole. However, as nonclassical correlations are introduced, this singular point bifurcates. A quintessential example is the $N$-photon NOON state, $|\psi_{\text{NOON}}\rangle = (|N, 0\rangle_H + |0, N\rangle_V)/\sqrt{2}$ \cite{dowling2008quantum}, which is optimal for frequency and phase estimation at fixed photon number \cite{bollinger1996optimal} and is heavily utilized in quantum lithography to break the Rayleigh diffraction limit. Formulating its stereographic Majorana polynomial yields:
\begin{equation}
P(z) = 1 + (-1)^N z^N \implies z_k = \exp\left[ i \frac{\pi(2k+1-N)}{N} \right],
\label{eq:noon_state_roots}
\end{equation}
with $k = 0, \dots, N-1$; the parity factor comes from the phase convention of Eq.~\eqref{eq:standard_Majorana_polynomial} and amounts to a rigid rotation of the polygon by $\pi/N$ about the polar axis.
The analytical roots dictate that the $N$ stars distribute themselves in a perfectly symmetric, equidistant polygon along the equator of the Poincar\'e sphere \cite{sanchezsoto2026quantum}. The vanishing dipole moment of this configuration makes the state anticoherent to first order, and its sensitivity to rotations about the polar axis saturates the Heisenberg limit, Eq.~\eqref{eq:fisher_planar}. It should not, however, be confused with the isotropic extremal constellations discussed above: the planar polygon is maximally \emph{anisotropic} at second order, with $\langle S_z^2\rangle = S^2$ against $\langle S_x^2\rangle = \langle S_y^2\rangle = S/2$, so that its phase sensitivity is concentrated on a single axis rather than distributed over all of them.

Beyond the binary basis of polarization, the spatial modes of structured light---specifically beams carrying OAM---provide a theoretically unbounded Hilbert space. Photons characterized by a helical phase front $e^{i \ell \phi}$, such as Laguerre-Gaussian modes, carry a quantized OAM of $\ell \hbar$, a degree of freedom shown to support genuine entanglement between photon pairs \cite{mair2001entanglement}. By artificially truncating this spectrum to a finite, symmetric subspace spanning topological charges from $\ell = -S$ to $\ell = +S$, a single twisted photon effectively mimics a macroscopic spin-$S$ qudit. An arbitrary coherent superposition of these spatial modes, $|\psi\rangle = \sum_{m=-S}^{S} c_m |\ell = m\rangle$, maps uniquely to a monic polynomial of degree $2S$, generating a well-defined constellation of Majorana stars on an effective OAM cylinder or sphere \cite{fabre2023majorana}.

The implementation of these constellations in twisted-light platforms offers unparalleled experimental manipulability. Using spatial light modulators (SLMs), vortex phase plates, and multi-plane light converters, experimentalists can deterministically encode arbitrary Majorana polynomials by mapping specific spatial distributions of OAM stars into holographic phase masks. When these twisted photons are subjected to adiabatic cyclic transformations---achieved by propagating the beam through cascaded q-plates---they accumulate geometric Berry phases. The relative internal trajectories of the spatially dispersed OAM stars are expected to contribute anomalous phase shifts of the internal-twist type introduced in Sec.~\ref{sec:phases}. We stress that this is a proposal rather than an established correspondence: as noted there, no general relation between the internal twist and an independently measurable quantity has been derived, and no such measurement has been reported in a photonic platform. Furthermore, photonic states are exceptionally robust probes for open-system dynamics. Depolarization, modal crosstalk, and photon loss mathematically manifest as the continuous radial collapse of the stars toward the origin of the Poincar\'e or OAM ball. Multipole-based geometric tomography \cite{romero2024multipoles}, executed via standard intensity measurements and mode sorters, directly tracks this radial decay, providing an immediate, visual diagnosis of quantum channel fidelity. By unifying topological optical phases, multi-photon interference, and error characterization into algebraic roots on a sphere, the Majorana representation serves as an indispensable blueprint for navigating high-dimensional quantum photonics.

\section{Challenges, Open Problems, and Future Directions}
\label{sec:challenges}

Despite the mathematical elegance of the Majorana stellar representation in describing pure symmetric evolution, fundamental and practical challenges remain in transitioning this framework from a static visualization tool into a robust diagnostic engine for macroscopic quantum technologies. Addressing these open problems necessitates a synthesis of advanced theoretical extensions, scalable computational strategies, and novel geometric mapping techniques. This section distills these pressing challenges into three core frontiers: the expansion of the mathematical formalism, the mitigation of computational bottlenecks via machine learning, and the experimental verification of geometric signatures in macroscopic and thermodynamic settings.

\subsection{Theoretical and Mathematical Frontiers}

The primary theoretical hurdle lies in extending the geometric formalism beyond the idealized symmetric pure state. In realistic open-system dynamics, quantum states inevitably suffer from decoherence and broken symmetries due to local environmental noise, evolving into mixed density matrices $\rho(t)$ governed by a Lindblad master equation. Under such dissipative evolution, the rigorous mathematical guarantee of a unique $2S$-point constellation on the surface of the Bloch sphere fundamentally breaks down. Current theoretical generalizations attempt to circumvent this by assigning to $\rho(t)$ a distribution of vectors strictly within the interior of the Bloch ball ($r_k \le 1$), utilizing generalized multipole expansions \cite{serrano2020majorana}. Alternative approaches utilize convex support decompositions to express $\rho(t)$ as a statistical mixture of coherent states \cite{giraud2008classicality}. However, the lack of a mathematically unique geometric representation for mixed states leaves the exact formulation of mixed-state entanglement invariants as a critical open problem. Furthermore, in non-Markovian environments with strong memory effects, the radial coordinates of these generalized interior points can exhibit non-monotonic shrinking or transient revivals, drastically complicating the visual interpretation of decoherence \cite{breuer2009measure}.

Equally challenging is the generalization of the stellar representation to fully non-symmetric $N$-qubit states. Because the Hilbert space dimension $2^N$ exponentially exceeds the $N+1$ degrees of freedom available in the permutation-symmetric subspace, a single Bloch sphere is inherently insufficient to capture the full quantum correlations of an arbitrary state \cite{markham2011entanglement}. Pragmatic approaches project the arbitrary state onto the symmetric subspace, $\rho_{\rm sym} = \mathcal{P}_{\rm sym}\rho\mathcal{P}_{\rm sym}$, utilizing the resulting symmetric constellation to establish rigorous lower bounds for genuine multipartite entanglement \cite{aulbach2010maximally}. A rigorous extension to asymmetric states necessitates classifying entanglement via multidimensional hyperdeterminants \cite{miyake2003classification}. While mathematically exact, these high-dimensional algebraic invariants sacrifice the intuitive simplicity of the single-sphere visualization. Developing noise-robust geometric witnesses that seamlessly combine symmetric projections with hyperdeterminant-based metrics remains a pressing priority for certifying entanglement in noisy intermediate-scale quantum (NISQ) devices.

Looking toward the boundaries of fundamental many-body physics, the constellation geometry offers profound, analytically exact links to non-equilibrium quantum dynamics and criticality. The stellar representation serves as a natural geometric language for exploring excited-state quantum phase transitions (ESQPTs) in collective spin systems. In the thermodynamic limit of the generalized Lipkin-Meshkov-Glick (LMG) model, the critical energies separating distinct quantum phases uniquely manifest as topological bifurcations in the trajectories of the Majorana stars \cite{kam2025three}. These bifurcations dynamically deform the underlying constellation, stretching and squeezing the stellar distributions along specific geodesics. Recent advancements have successfully mapped these high-dimensional quantum critical behaviors onto classical optical analogues \cite{kam2026classical}, proving that the macroscopic spatial structures of the constellation---specifically its tendency to cluster or maximally disperse---encode the fundamental signatures of quantum criticality. Understanding how the internal geometric twist and pairwise phases \cite{kam2021berry} evolve across these critical boundaries represents a fascinating horizon for both theoretical exploration and experimental simulation.

\subsection{Computational Scalability and Machine Learning Solutions}

As experimental platforms scale toward the thermodynamic limit ($S \to \infty$), the theoretical elegance of the Majorana representation encounters severe computational bottlenecks. Extracting the spatial coordinates of the constellation strictly requires determining the $2S$ complex roots of a high-degree Majorana polynomial. Utilizing standard dense companion-matrix eigenvalue solvers incurs an asymptotic time complexity of $O(S^3)$ in the degree $2S$, rendering exact geometric visualizations computationally prohibitive for macroscopic spin systems. Compounding this complexity is the inherent numerical instability of high-degree polynomials; as famously demonstrated by Wilkinson's polynomials, microscopic perturbations in the polynomial coefficients can trigger exponentially divergent errors in the calculated star positions \cite{wilkinson1959evaluation}. This ill-conditioning is particularly devastating for resolving near-degenerate configurations or highly clustered states. While structured eigenvalue algorithms and arbitrary-precision rootfinders such as MPSolve \cite{bini2000design} push the tractable degree substantially higher, evaluating non-equilibrium dynamics in the true macroscopic limit demands a radical paradigm shift in computational strategy.

Machine learning (ML) architectures offer a transformative pathway to bypass this explicit root-finding bottleneck by mapping macroscopic observables directly to geometric entanglement features \cite{carrasquilla2017machine}. Rather than diagonalizing companion matrices to extract discrete stellar coordinates, deep learning models can operate on the continuous quasiprobability distributions on the Bloch sphere, such as the Husimi $Q$-function or the directly measured multipoles $\rho_{kq}$ of Eq.~\eqref{eq:multipole_expansion}. For instance, Variational Autoencoders (VAEs) can compress these high-dimensional macroscopic distributions into a compact, continuous latent space, capturing the essential geometric structures of the constellation---such as generalized squeezing or isotropic dispersion---without evaluating a single polynomial root. Furthermore, instead of applying discrete clustering algorithms to individual stars, Topological Data Analysis (TDA) can be applied directly to these continuous functions. By utilizing persistent homology, TDA efficiently extracts robust macroscopic geometric invariants---identifying major clusters, voids, and topological bifurcations inherent in the thermodynamic limit, fully circumventing the numerical instability of exact root extraction.

The immediate computational frontier lies in scaling these models utilizing inherently $\mathrm{SU}(2)$-equivariant neural networks, such as tensor field networks \cite{thomas2018tensor}, which rigorously respect the rotational symmetries of the Bloch sphere without requiring exhaustive data augmentation. Additionally, the deployment of predictive generative frameworks, such as geometrically constrained World Models or diffusion models, promises to invert the analytical pipeline. Generating noise-adapted geometric configurations conditioned on target macroscopic multipoles would invert the analytical pipeline and bypass explicit rootfinding altogether. We emphasise that, with the exception of the equivariant architectures just cited, none of these approaches has yet been applied to Majorana constellations; they are proposed directions rather than reported results, and their value here is to identify where the computational bottleneck might be attacked, not to claim it has been.

\subsection{Experimental Horizons and Thermodynamic Signatures}

The final theoretical and experimental frontier lies in the actionable verification of these geometric signatures within macroscopic regimes, particularly in integrated quantum metrology and non-equilibrium thermodynamics. In metrology, the geometric framework has proven indispensable for engineering phase verification protocols that balance sensitivity with noise resilience. Planar equidistant configurations---the GHZ-type states of Eq.~\eqref{eq:fisher_planar}---saturate the Heisenberg limit, $\mathcal{F} = N^2$, but the same $N$-fold phase accumulation that produces this sensitivity also amplifies dephasing, and under independent dephasing of the individual qubits the coherence decays at a rate proportional to $N$ \cite{pezze2018quantum}. Isotropic extremal constellations \cite{goldberg2020extremal} trade peak sensitivity for a different profile: the maximally anticoherent states---the ``Kings of Quantumness''---have Platonic-solid symmetry, vanishing dipole, and the direction-independent value $\mathcal{F} = \tfrac13 N(N+2)$ of Eq.~\eqref{eq:fisher_anticoherent} \cite{bjork2015extremal}, so that no alignment of probe and signal is required and the leading response to a weak perturbation is pushed above the anticoherence order. As emphasised in Sec.~\ref{sec:qec}, this is not immunity to collective noise: the symmetric subspace carries a single irreducible representation of $\mathrm{SU}(2)$ and admits no decoherence-free subspace against collective operators. Exploiting these robust geometric markers experimentally demands unprecedented phase resolution to isolate the internal anomalous Berry phases ($\Delta \gamma$) that accumulate dynamically as the constellation is adiabatically manipulated \cite{kam2021berry}.

Beyond static metrology, the constellation geometry provides a striking visual and analytic gauge for macroscopic quantum thermodynamics, particularly in systems driven far from equilibrium. Instead of standard ergotropy \cite{allahverdyan2004maximal}, which favors trivial coherent polarization under linear Hamiltonians, the thermodynamic utility of entangled configurations manifests during finite-time driving across quantum critical points. In symmetry-broken systems such as the generalized Lipkin-Meshkov-Glick (LMG) model, passing through an excited-state quantum phase transition (ESQPT) induces violent topological deformations in the stellar constellation \cite{kam2025three, kam2026classical}. The thermodynamic work distribution and irreversible entropy production generated during such non-adiabatic cycles are directly encoded in the structural splitting and recombination of the Majorana roots. Investigating how the topological constraints of the constellation bound the irreversible energetic costs, and formulating how the internal geometric phases modify non-equilibrium fluctuation theorems, represent highly promising, unexplored avenues for the design of quantum heat engines and finite-time thermodynamic cycles.

Ultimately, by mapping abstract logical error syndromes to deterministic stellar shifts and visualizing macroscopic criticality as explicit geometric deformations, the Majorana stellar representation establishes an indispensable bridge between algebraic quantum information and observable spatial dynamics \cite{sanchezsoto2026quantum}. While extending this mathematical elegance to fully asymmetric density matrices, strongly non-Markovian environments, and the inverse state-preparation problem via machine learning presents formidable theoretical hurdles, overcoming these challenges remains a paramount objective. Successfully navigating these physical frontiers will definitively elevate the constellation framework from a descriptive visualization tool into a predictive, actionable diagnostic language, fundamentally custom-built for characterizing the next generation of complex quantum simulators and strongly correlated many-body hardware \cite{Preskill2018}.

\section{Conclusions}
\label{sec:conclusions}

Proposed nearly a century ago as an elegant geometric descriptor for atomic spins in time-varying magnetic fields \cite{Majorana1932}, the Majorana stellar representation has evolved into a formidable analytical framework within modern quantum information science. By mapping the abstract, high-dimensional Hilbert space vectors of permutation-symmetric multi-qubit and macroscopic spin systems onto discrete constellations of points on the Bloch sphere, it fundamentally transforms opaque algebraic complexities into tractable geometric topologies \cite{Bengtsson2017}. This geometric translation has proven indispensable for illuminating hidden symmetries, macroscopic degeneracies, and the intricate structural topology of multipartite correlations at a single glance \cite{sanchezsoto2026quantum}.

Throughout this review, we have foregrounded an entanglement-centric perspective, demonstrating that the constellation is not merely a qualitative visualization aid, but a rigorous mathematical ledger for quantum correlations. We have detailed how specific stellar degeneracy patterns yield a complete and visually intuitive classification of stochastic local operations and classical communication (SLOCC) families \cite{bastin2009operational, markham2011entanglement}. Furthermore, fundamental entanglement monotones---such as the bipartite concurrence and the genuine multipartite three-tangle---can be computed exactly using inter-star chordal distances and stereographic inner products \cite{martin2010multiqubit}. Independent of bipartite entanglement, intrinsic nonclassicality is rigorously quantified through the stellar rank and the anticoherence order, identifying extremal states like the ``Kings of Quantumness'' \cite{bjork2015extremal, giraud2015tensor}. Crucially, this geometric framework bridges static correlation metrics with unitary dynamics: the internal geometry of the constellation generates anomalous contributions to the Berry phase over and above the solid-angle term \cite{kam2021berry}. Whether these anomalous contributions can be bounded by, or read off from, an independent entanglement measure is a question this review leaves open---as discussed in Sec.~\ref{sec:phases}, no such relation has been established, and we regard it as the outstanding problem in this direction rather than a settled result.

The operational power of this framework is increasingly being realized across diverse experimental architectures. The theoretical accessibility of collective multipole measurements and algebraic root-finding has enabled the direct geometric benchmarking of quantum states in trapped-ion arrays \cite{haffner2005scalable}, spinor Bose-Einstein condensates \cite{stamper2013spinor}, and structured twisted-photon networks \cite{fabre2023majorana}. By treating the Majorana constellation as a unified geometric object, experimentalists can seamlessly bridge static entanglement quantification, track open-system decoherence via continuous radial shrinkage, and optimize noise-resilient quantum sensors and symmetric error-correcting codes \cite{goldberg2020extremal}. 

Nevertheless, fulfilling the ultimate theoretical promise of this geometric language necessitates overcoming several persistent frontiers. As experimental platforms scale toward the thermodynamic limit ($S \to \infty$), the computational scalability of macroscopic polynomial root-finding becomes a severe numerical bottleneck. Furthermore, formulating unique canonical extensions for fully asymmetric multi-qubit architectures and dissipative mixed density matrices remains theoretically elusive \cite{aulbach2010maximally, serrano2020majorana}. Addressing these hurdles will increasingly rely on the integration of hybrid algebraic-geometric tools and advanced machine learning models capable of bypassing explicit root calculations \cite{carrasquilla2017machine}, alongside deeper quantitative mappings to non-equilibrium dynamics and excited-state quantum phase transitions.

Beyond static classification, this geometric framework suggests a route to quantum kernel methods \cite{tanner2026non}: encoding classical data into Majorana constellations makes the pairwise overlap of the resulting symmetric states a candidate kernel function. Here the complexity analysis of Appendix~\ref{app:complexity} imposes an important constraint. The overlap of two symmetric states is a permanent of the matrix of pairwise spinor inner products, and that matrix has rank at most two; by Barvinok's algorithm \cite{barvinok1996} its permanent is computable in polynomial time. Evaluating such a kernel is therefore \emph{not} a hard classical problem, and no advantage can be claimed on the grounds of evaluation cost. If constellation kernels prove useful, the reason will have to be sought elsewhere---in the $\mathrm{SU}(2)$-equivariance and permutation invariance they build in by construction, which constrain the hypothesis class in a way that may suit data with rotational structure. Establishing whether that inductive bias yields any concrete advantage is an open question, and one that the geometry alone does not settle.

Ultimately, the Majorana stellar representation stands as a fundamental geometric language that uniquely complements standard algebraic formalisms. By centering entanglement as the primary organizing principle, this framework unifies disparate quantum phenomena---ranging from metrology to macroscopic critical dynamics---under a single topological umbrella. As quantum control techniques mature, the constellation framework offers unparalleled clarity, intuition, and predictive power, perfectly positioning it to guide practical, noise-resilient advancements in the next generation of complex quantum technologies \cite{Preskill2018}.

\section*{Declarations}

\subsection*{Conflict of Interest}
The author declares that there are no financial or personal relationships with other people or organizations that could inappropriately influence (bias) the work reported in this paper.

\subsection*{Data Availability Statement}
Data sharing is not applicable to this article as no new data were created or analyzed during the current study. Theoretical results and literature synthesized are cited accordingly.

\subsection*{Ethics Statement}
This research did not involve human participants, animals, or harmful chemicals; therefore, ethical approval was not required for this study.

\subsection*{Funding}
C.F.K. is grateful to the European Union and the Region Reunion, France (POE FEDER 2021--2027, n\textdegree{}2025-0954-007180) for the funding support.

\subsection*{Author Contributions}
C.F.K.: Conceptualization, Formal analysis, Investigation, Writing - original draft, Writing - review \& editing.

\appendix

\section{Rigorous derivation of the isomorphism between a spin-$S$ and the symmetric subspace of 2$S$ qubits via Schur-Weyl duality}
\label{app:isomorphism}

The Majorana stellar representation is predicated on a formal isomorphism between the Hilbert space of a monolithic spin-$S$ system, $\mathcal{H}_S$, and the totally symmetric subspace of $N=2S$ qubits, $\mathcal{H}_{\text{sym}} \subset (\mathbb{C}^2)^{\otimes N}$. This correspondence is a direct consequence of the representation theory of $SU(2)$ and its interplay with the symmetric group $S_N$. 

Let $V \cong \mathbb{C}^2$ be the fundamental representation (spin-$1/2$) of the group $SU(2)$. For a system of $N$ qubits, the total Hilbert space is the $N$-fold tensor product:
\begin{equation}
\mathcal{H}_N = V_1 \otimes V_2 \otimes \dots \otimes V_N = V^{\otimes N}.
\end{equation}
The group $SU(2)$ acts on $\mathcal{H}_N$ via the diagonal (collective) representation $\rho^{\otimes N}$, defined by:
\begin{equation}
\rho^{\otimes N}(g) = \underbrace{g \otimes g \otimes \dots \otimes g}_{N}, \quad \forall g \in SU(2).
\end{equation}
According to the Clebsch-Gordan decomposition, this highly reducible representation decomposes into a direct sum of $SU(2)$ irreducible representations (irreps) $V_J$:
\begin{equation}
V^{\otimes N} \cong \bigoplus_{J = J_{\min}}^{N/2} \nu_J V_J,
\end{equation}
where $V_J$ denotes the irrep of dimension $2J+1$, and $\nu_J$ is the multiplicity of each irrep.

To identify the symmetric part of this decomposition, we invoke Schur-Weyl duality. Consider the action of the symmetric group $S_N$ on $\mathcal{H}_N$ by permuting the tensor factors
\begin{equation}
\sigma \cdot (v_1 \otimes v_2 \otimes \dots \otimes v_N) = v_{\sigma^{-1}(1)} \otimes v_{\sigma^{-1}(2)} \otimes \dots \otimes v_{\sigma^{-1}(N)},
\end{equation}
for $\sigma \in S_N$. Since the diagonal action of $SU(2)$ and the permutation action of $S_N$ commute, $\mathcal{H}_N$ admits a joint decomposition
\begin{equation}
\mathcal{H}_N \cong \bigoplus_{\lambda} U_{\lambda} \otimes W_{\lambda},
\end{equation}
where $U_{\lambda}$ are irreps of $SU(2)$ and $W_{\lambda}$ are irreps of $S_N$, labeled by Young diagrams $\lambda$ with $N$ boxes and at most 2 rows.

The totally symmetric subspace $\mathcal{H}_{\text{sym}}$ is defined as the range of the symmetrizer $P_{\text{sym}} = \frac{1}{N!} \sum_{\sigma \in S_N} \sigma$. This subspace corresponds to the trivial representation of $S_N$, labeled by the single-row Young diagram $\lambda = (N, 0)$. Under Schur-Weyl duality, this specific sector is uniquely paired with the $SU(2)$ irrep of highest weight $J = N/2$
\begin{equation}
\mathcal{H}_{\text{sym}} \cong V_{N/2} \otimes W_{(N,0)} \cong V_{N/2}.
\end{equation}
Given $N = 2S$, we arrive at the exact isomorphism: $\mathcal{H}_{\text{sym}} \cong V_S$.

The $2S+1$ basis vectors of $\mathcal{H}_{\text{sym}}$ are the Dicke states $|S, m\rangle$, where $m \in \{-S, \dots, S\}$. In the tensor product language, $|S, S\rangle = |1\rangle^{\otimes N}$ is the highest weight vector. The full basis is generated via the collective lowering operator $\hat{J}_- = \sum_{i=1}^N \hat{\sigma}_-^{(i)}$:
\begin{align}
&|S, m\rangle = \mathcal{N} (\hat{J}_-)^{S-m} |1\rangle^{\otimes N} \nonumber\\
=& \binom{2S}{S+m}^{-1/2} \sum_{\sigma \in S_N / \text{stab}} \sigma \left( |1\rangle^{\otimes S+m} \otimes |0\rangle^{\otimes S-m} \right).
\end{align}
This confirms that any monolithic spin-$S$ state $|\Psi\rangle \in V_S$ can be uniquely unfolded into the constituent $2S$ qubits while maintaining total permutation symmetry, providing the rigorous foundation for the Majorana constellation.

\section{Mathematical equivalence of Majorana polynomials and Penrose spinors}
\label{app:spinor_calculus}

The Majorana stellar representation and Penrose's spinor calculus approach the geometry of spin-$S$ systems from two distinct physical motivations: Majorana considered non-relativistic quantum states under $\text{SU}(2)$ rotations, whereas Penrose formulated spinor calculus to describe massless higher-spin fields and spacetime geometry under the Lorentz group $\text{SO}^+(1,3)$, which is locally isomorphic to $\text{SL}(2,\mathbb{C})$. Despite these differing physical origins, their underlying mathematical structures are strictly isomorphic, as both rely on the symmetric tensor product of two-dimensional complex vector spaces.

In Majorana's formulation, the state space is the $(2S+1)$-dimensional irreducible representation of $\text{SU}(2)$. A pure state $|\psi\rangle = \sum_{m=-S}^S c_m |S,m\rangle$ is mathematically mapped to a polynomial in a single complex variable $z$,
\begin{equation}
P(z) = \sum_{k=0}^{2S} a_k z^k,
\label{eq:app_majorana_poly}
\end{equation}
where the coefficients $a_k$ are derived from $c_{k-S}$. By the fundamental theorem of algebra, $P(z)$ factors completely into $2S$ roots in the extended complex plane $\mathbb{C} \cup \{\infty\}$, expressed as $P(z) = A \prod_{i=1}^{2S} (z - z_i)$, where $A$ is a normalization constant. The symmetry group $\text{SU}(2)$ acts on these roots via fractional linear (M\"obius) transformations, establishing a direct mapping between the spin-$S$ Hilbert space and the configuration of $2S$ points on the unit sphere.

Conversely, Penrose's formalism is constructed upon 2-component Weyl spinors $\xi^A = (\xi^0, \xi^1)^T$, which form the fundamental representation of $\text{SL}(2,\mathbb{C})$. A higher-spin object of spin $S$ is represented by a rank-$n$ totally symmetric spinor $\Psi_{A_1 A_2 \dots A_n}$, where $n = 2S$. This spinor resides in the $(n+1)$-dimensional symmetric tensor product space $\text{Sym}^n(\mathbb{C}^2)$. The core of Penrose's approach is the spinor decomposition theorem, which dictates that any totally symmetric rank-$n$ spinor can be uniquely factorized into the symmetrized product of $n$ rank-$1$ principal spinors:
\begin{equation}
\Psi_{A_1 A_2 \dots A_n} = \mathcal{S} \left( \alpha_{A_1}^{(1)} \alpha_{A_2}^{(2)} \dots \alpha_{A_n}^{(n)} \right) = \alpha_{(A_1}^{(1)} \alpha_{A_2}^{(2)} \dots \alpha_{A_n)}^{(n)},
\label{eq:app_spinor_decomposition}
\end{equation}
where the parentheses $( \dots )$ denote complete symmetrization over the enclosed indices, and each constituent principal spinor is given by $\alpha^{(k)}_A = (\alpha^{(k)}_0, \alpha^{(k)}_1)^T$.

The geometric isomorphism between these two frameworks is established by projecting the 2-component spinors into the complex projective line $\mathbb{CP}^1$. For each principal spinor $\alpha^{(k)}_A$, one defines the complex ratio $z_k = \alpha^{(k)}_1 / \alpha^{(k)}_0$. To link this to Majorana's polynomial, one introduces the spinor $\xi^A = (-z, 1)^T$---the sign convention is fixed by the requirement that the resulting roots be the ratios $z_k$ defined above rather than their negatives-reciprocals---and constructs the fully contracted scalar product $\Psi_{A_1 A_2 \dots A_n} \xi^{A_1} \xi^{A_2} \dots \xi^{A_n} = 0$. Substituting the factorized form from Eq.~\eqref{eq:app_spinor_decomposition} into this contraction directly yields:
\begin{equation}
\prod_{k=1}^n \left( \alpha^{(k)}_0 z - \alpha^{(k)}_1 \right) = \prod_{k=1}^n \alpha^{(k)}_0 (z - z_k) = 0.
\label{eq:app_contraction}
\end{equation}
Thus, the projective ratios $z_k$ of Penrose's principal spinors are strictly identical to the roots $z_i$ of Majorana's polynomial. In the relativistic context, each principal spinor $\alpha^{(k)}_A$ defines a real null vector $k^\mu = \sigma^\mu_{A\dot{A}} \alpha^{(k)A} \bar{\alpha}^{(k)\dot{A}}$ (where $\sigma^\mu$ are the Infeld-van der Waerden symbols), corresponding to a principal null direction (a light-ray) on the observer's celestial sphere.

This mathematical equivalence provides a rigorous foundation for mapping gravitational geometries to quantum entanglement structures, particularly for $S=2$ systems ($n=4$). In general relativity, the massless spin-$2$ gravitational field is described by the Weyl curvature tensor $C_{\mu\nu\rho\sigma}$, which translates isomorphically to a rank-$4$ totally symmetric spinor $\Psi_{ABCD}$. According to Eq.~\eqref{eq:app_spinor_decomposition}, this spinor factorizes into four principal spinors, $\Psi_{ABCD} = \alpha_{(A} \beta_B \gamma_C \delta_{D)}$, corresponding to four Majorana stars. 

The degeneracy of these four roots classifies the spacetime geometries according to the Petrov classification, mapping exactly to the integer partitions of $4$. If all four roots are distinct ($\mathcal{D}_{1,1,1,1}$), the spacetime is algebraically general (Type I). When two roots coincide while the other two remain distinct ($\mathcal{D}_{2,1,1}$), the geometry becomes Type II. A double degeneracy where roots coincide in two pairs ($\mathcal{D}_{2,2}$) defines Type D spacetimes, characterizing the highly symmetric gravitational fields of isolated sources such as Schwarzschild and Kerr black holes. If three roots coincide ($\mathcal{D}_{3,1}$), the spacetime is classified as Type III, and when all four roots perfectly coincide ($\mathcal{D}_4$), the geometry is Type N, defining pure transverse gravitational waves where all curvature propagates along a single null direction.

Crucially, a $4$-qubit totally symmetric pure quantum state $|\psi\rangle_{\text{sym}}$ is described by the exact same symmetric tensor algebra in $\text{Sym}^4(\mathbb{C}^2)$. Under Stochastic Local Operations and Classical Communication (SLOCC), which correspond to local $\text{SL}(2,\mathbb{C})^{\otimes 4}$ filtering operations preserving the permutation symmetry, the degeneracy pattern of the Majorana constellation is invariant. The five partitions of $4$ therefore label five SLOCC families, in exact correspondence with the Petrov types: $\mathcal{D}_{1,1,1,1}$ with Type I and the generic symmetric states containing the four-qubit GHZ state; $\mathcal{D}_{2,1,1}$ with Type II; $\mathcal{D}_{2,2}$ with Type D and the Dicke state $|D_4^{(2)}\rangle = |S{=}2, m{=}0\rangle$, whose two double stars sit at antipodal poles; $\mathcal{D}_{3,1}$ with Type III and the four-qubit W state; and $\mathcal{D}_4$ with Type N and the fully separable symmetric states.

The correspondence should not be overstated in the other direction. The map from degeneracy patterns to Petrov types is a bijection, but the map from Petrov types to SLOCC classes is not: as discussed in Sec.~\ref{sec:entanglement}, a family of four distinct stars carries a one-complex-dimensional modulus, the cross-ratio, which is preserved by every M{\"o}bius transformation and hence by every symmetric SLOCC operation \cite{ribeiro2011entanglement}. Type I spacetimes and $\mathcal{D}_{1,1,1,1}$ symmetric states therefore both decompose further into a continuum of inequivalent orbits, and the Petrov labels classify the algebraic degeneracy structure rather than the full orbit. With that caveat the duality is exact and useful: the same symmetric spinor algebra governs relativistic curvature and symmetric multi-qubit entanglement, and the same geometric invariants---degeneracy pattern and cross-ratio---classify both.

\section{Bargmann-Fock Space and Analyticity of the Stellar Representation}
\label{app:bargmann}

To provide a rigorous foundation for the stellar representation in continuous-variable (CV) systems, we summarize the mathematical properties of the Bargmann-Fock space \cite{bargmann1961hilbert} and the convergence criteria that guarantee its entire analytic nature. For any pure state $|\psi\rangle = \sum_{n=0}^{\infty} c_n |n\rangle$ in the infinite-dimensional Hilbert space, the Bargmann representation maps the state to a complex function $F(z)$ via the kernel
\begin{equation}
F(z) = \sum_{n=0}^{\infty} \frac{c_n}{\sqrt{n!}} z^n,
\label{eq:bargmann_sum}
\end{equation}
where $z \in \mathbb{C}$ is a complex variable. In this representation, the bosonic creation operator $a^\dagger$ acts as multiplication by $z$, and the annihilation operator $a$ acts as the derivative $\partial_z$. A crucial property of the Bargmann function $F(z)$ is that for any physically realizable state satisfying the normalization condition $\sum_{n} |c_n|^2 = 1$, $F(z)$ is unconditionally an \emph{entire function}, meaning it is analytic throughout the entire complex plane. 

To demonstrate this rigorously, we consider the radius of convergence $R$ of the power series in Eq.~\eqref{eq:bargmann_sum}. According to the Cauchy-Hadamard theorem, $1/R = \limsup_{n \to \infty} |a_n|^{1/n}$, where the coefficients are $a_n = c_n / \sqrt{n!}$. Given that $|c_n| \leq 1$ for all $n$, we establish the upper bound:
\begin{equation}
\frac{1}{R} = \limsup_{n \to \infty} \left| \frac{c_n}{\sqrt{n!}} \right|^{1/n} \leq \limsup_{n \to \infty} \frac{1}{(n!)^{1/2n}}.
\end{equation}
Applying Stirling's approximation, $(n!)^{1/n} \sim n/e$, it directly follows that $1/(n!)^{1/2n} \to 0$ as $n \to \infty$. Consequently, $1/R = 0$, implying an infinite radius of convergence ($R = \infty$). This ensures that $F(z)$ possesses no singularities or branch points in the finite complex plane, regardless of the asymptotic behavior of the physical amplitudes $c_n$.

This analytic property is strictly required when evaluating the physical phase-space distribution, typically analyzed via the Husimi $Q$-function \cite{husimi1940some}. Defined as the expectation value of the density matrix in the coherent state basis, $Q(\alpha) = \frac{1}{\pi} \langle \alpha | \psi \rangle \langle \psi | \alpha \rangle$, the inner product relies on the standard expansion of coherent states $|\alpha\rangle = e^{-|\alpha|^2/2} \sum \frac{\alpha^n}{\sqrt{n!}} |n\rangle$. The projection naturally isolates the Bargmann function:
\begin{equation}
\langle \alpha | \psi \rangle = e^{-|\alpha|^2/2} \sum_{n=0}^{\infty} c_n \frac{(\alpha^*)^n}{\sqrt{n!}} = e^{-|\alpha|^2/2} F(\alpha^*).
\end{equation}
The Husimi $Q$-function can thus be exactingly expressed as
\begin{equation}
Q(\alpha) = \frac{1}{\pi} e^{-|\alpha|^2} |F(\alpha^*)|^2.
\label{eq:Q_factorization}
\end{equation}
While $Q(\alpha)$ itself is a non-negative real-valued distribution and hence mathematically prohibited from direct analytic factorization, it is uniquely dictated by the modulus square of the entire function $F(\alpha^*)$. Because $F(z)$ is an entire function of growth order at most $2$ for any normalized Fock-state superposition, Hadamard's factorization theorem applies: $F$ is determined by its zeros together with a non-vanishing factor $e^{a + bz + cz^2}$. The exponential factor is not a technicality---a coherent state has $F(z) = e^{\bar{\alpha} z}$ and no zeros at all---so the zeros alone do not determine the state. For CV states of finite stellar rank $r^\star = n$, the Bargmann function reduces algebraically to a polynomial of degree $n$ multiplied by such a non-vanishing exponential, and the $n$ complex roots of $F$ give exactly the phase-space coordinates at which the Husimi $Q$-function vanishes, providing a geometric fingerprint of the state's non-Gaussian resources.

One difference from the spin case deserves emphasis, since the two are easily conflated. Here $Q(\alpha) \propto |F(\alpha^*)|^2$ directly, so the zeros of the stellar function are the zeros of $Q$. For a spin-$S$ state the corresponding statement carries an antipodal map, Eq.~\eqref{eq:Q_pure}: the zeros of $Q$ lie at the antipodes of the Majorana stars, not at the stars themselves.

\section{Computational Complexity and the Permanent--Determinant Dichotomy}
\label{app:complexity}

The efficacy of the Majorana representation in bypassing the \#P-hard barriers of general multi-qubit systems can be understood through the transition from tensor algebra to univariate polynomials. Here, we provide a formal discussion on the computational complexity of constructing the constellation and the subsequent evaluation of physical observables.

The construction of the Majorana constellation for a spin-$S$ state (where $N=2S$) requires finding the roots of the polynomial $P_\psi(z)$ defined in Eq.~\eqref{eq:standard_Majorana_polynomial}. In numerical linear algebra, the standard approach to find the roots of an $N$-th degree polynomial is to construct its associated companion matrix $\mathcal{C}$, whose characteristic polynomial is precisely $P_\psi(z)$. For a monic polynomial $z^N + a_{N-1} z^{N-1} + \dots + a_0$, the companion matrix is defined as:
\begin{equation}
\mathcal{C} = \begin{pmatrix} 
0 & 0 & \dots & 0 & -a_0 \\
1 & 0 & \dots & 0 & -a_1 \\
0 & 1 & \dots & 0 & -a_2 \\
\vdots & \vdots & \ddots & \vdots & \vdots \\
0 & 0 & \dots & 1 & -a_{N-1}
\end{pmatrix}.
\end{equation}
The roots $\{z_k\}$ are obtained by computing the eigenvalues of $\mathcal{C}$. Using standard algorithms such as the QR decomposition or the Hessenberg reduction, the eigenvalues of an $N \times N$ matrix can be computed with a computational complexity of $\mathcal{O}(N^3)$. Since this is a polynomial-time operation, the geometric fingerprint of any symmetric state can be efficiently extracted even for large $N$.

While identifying the positions of the stars is computationally efficient, reconstructing many-body amplitudes or evaluating the overlap between two symmetric states remains demanding. A pure symmetric state $|\psi\rangle$ is proportional to the symmetrized product of the constituent single-qubit states $\{|\mathbf{n}_1\rangle, \dots, |\mathbf{n}_N\rangle\}$ corresponding to its Majorana stars \cite{Majorana1932, ribeiro2007thermodynamical}:
\begin{equation}
|\psi\rangle = \frac{1}{\sqrt{N! \, \operatorname{perm}(G_\psi)}} \sum_{\sigma \in \mathcal{P}_N} |\mathbf{n}_{\sigma(1)}\rangle \otimes \dots \otimes |\mathbf{n}_{\sigma(N)}\rangle,
\label{eq:symmetric_state_stars}
\end{equation}
where $\mathcal{P}_N$ denotes the set of all permutations of $N$ elements, and $G_\psi$ is the Gram matrix of the state's own stars with elements $(G_\psi)_{ij} = \langle \mathbf{n}_i | \mathbf{n}_j \rangle$. The appearance of the permanent in the normalization constant is rigorously established in the following lemma. 

\begin{lemma}
Let $\{|n_1\rangle, |n_2\rangle, \dots, |n_N\rangle\}$ be a collection of $N$ non-orthogonal single-qubit states. The unnormalized totally symmetric $N$-qubit state is defined as
\begin{equation}
    |\Psi\rangle = \sum_{\sigma \in S_N} \bigotimes_{i=1}^N |n_{\sigma(i)}\rangle,
\end{equation}
where $S_N$ is the symmetric group of degree $N$. The inner product of this state satisfies
\begin{equation}
    \langle \Psi | \Psi \rangle = N! \operatorname{perm}(G),
\end{equation}
where $G$ is the Gram matrix with elements $G_{ij} = \langle n_i | n_j \rangle$. The normalized state $|\psi\rangle$ is given by
\begin{equation}
    |\psi\rangle = \frac{1}{\sqrt{N! \operatorname{perm}(G)}} \sum_{\sigma \in S_N} \bigotimes_{i=1}^N |n_{\sigma(i)}\rangle.
\end{equation}
\end{lemma}

\begin{proof}
The self-inner product $\langle \Psi | \Psi \rangle$ involves a double summation over permutations $\sigma, \tau \in S_N$:
\begin{equation}
    \langle \Psi | \Psi \rangle = \sum_{\sigma \in S_N} \sum_{\tau \in S_N} \prod_{k=1}^N \langle n_{\sigma(k)} | n_{\tau(k)} \rangle.
\end{equation}
By performing the re-indexing $i = \sigma(k)$, the product term is rewritten as:
\begin{equation}
    \prod_{k=1}^N \langle n_{\sigma(k)} | n_{\tau(k)} \rangle = \prod_{i=1}^N \langle n_i | n_{\tau(\sigma^{-1}(i))} \rangle.
\end{equation}
Introducing the composite permutation $\pi = \tau \circ \sigma^{-1}$ and invoking the rearrangement theorem for finite groups, we observe that for a fixed $\sigma$, $\pi$ traverses $S_N$ exactly once as $\tau$ spans the group. This allows the summations to decouple:
\begin{equation}
    \langle \Psi | \Psi \rangle = \sum_{\sigma \in S_N} \left( \sum_{\pi \in S_N} \prod_{i=1}^N \langle n_i | n_{\pi(i)} \rangle \right).
\end{equation}
The inner summation is independent of $\sigma$ and identifies with the definition of the matrix permanent $\operatorname{perm}(G) = \sum_{\pi \in S_N} \prod_{i=1}^N G_{i, \pi(i)}$. Summing over the constant term for all $\sigma \in S_N$ yields the factor $|S_N| = N!$, thus:
\begin{equation}
    \langle \Psi | \Psi \rangle = N! \operatorname{perm}(G).
\end{equation}
The normalized state is obtained by dividing $|\Psi\rangle$ by its norm $\sqrt{\langle \Psi | \Psi \rangle}$, completing the proof.
\end{proof}

Consequently, the overlap between two such states, $\langle \phi | \psi \rangle$, is rigorously determined by the sum over all permutations of the pairwise inner products between the stars of each constellation \cite{markham2011entanglement, wei2010matrix}:
\begin{equation}
\langle \phi | \psi \rangle = \frac{\operatorname{perm}(\mathcal{M})}{\sqrt{\operatorname{perm}(G_\phi) \operatorname{perm}(G_\psi)}},
\label{eq:overlap_permanent}
\end{equation}
where $\mathcal{M}_{ij} = \langle \mathbf{m}_i | \mathbf{n}_j \rangle$ is the matrix of cross-stellar overlaps. 

While the computation of the permanent of a general matrix is famously \#P-hard \cite{valiant1979complexity}, and typically presents an insurmountable combinatorial explosion for generic multi-qubit systems, the Majorana representation introduces a profound structural simplification. The Gram matrices ($G_\phi, G_\psi$) and the overlap matrix ($\mathcal{M}$) are entirely constructed from the inner products of single-qubit states---specifically, two-dimensional complex vectors. As a direct mathematical consequence, the rank of these $N \times N$ matrices is strictly bounded by 2. 

This low-rank property fundamentally alters the computational complexity class of the problem. As demonstrated by Barvinok \cite{barvinok1996}, the permanent of a matrix with a fixed rank $r$ can be computed efficiently in polynomial time. By leveraging Barvinok's algorithm for rank-2 matrices, the permanents in Eq.~\eqref{eq:overlap_permanent} can be evaluated exactly in $\mathcal{O}(N^3)$ time. 

Within the symmetric subspace this amounts to a genuine ``dequantization'': the Majorana representation provides an $\mathcal{O}(N^3)$ route to the state's geometry via polynomial root-finding, and the many-body overlaps of Eq.~\eqref{eq:overlap_permanent}, together with those invariants expressible in closed form through them---the three-tangle of Eq.~\eqref{eq:geometric_three_tangle_symmetric} among them---become exactly computable in polynomial time.

Two qualifications keep this from being oversold. First, the tractability is a statement about permanents of rank-$2$ matrices and carries no implication for Boson Sampling, whose hardness concerns permanents of generic full-rank matrices; the symmetric subspace evades the barrier rather than lowering it. Second, not every multipartite invariant reduces to these permanents: the geometric measure of Eq.~\eqref{eq:geometric_entanglement_stars}, for instance, still requires a maximization over the sphere and is not delivered in closed form by the constellation alone.

\section{Orthogonality in the Symmetric Subspace and the Antipodal Basis Algorithm}
\label{app:orthogonality}

This appendix rigorously addresses the orthogonality of states within the totally symmetric subspace. For a quantum system composed of $N$ qubits in the totally symmetric subspace $\mathcal{H}_{\text{sym}}$, the total Hilbert space dimension is restricted to $N+1$. Given any pure state $|\psi\rangle \in \mathcal{H}_{\text{sym}}$, its orthogonal complement $\mathcal{H}_\psi^\perp$ has a strict dimension of $N$. The inner product between $|\psi\rangle$ and another symmetric state $|\phi\rangle$ relies on the permanent of their stellar overlap matrix, $\mathcal{M}_{ij} = \langle w_i | z_j \rangle$, yielding $\langle \phi | \psi \rangle \propto \operatorname{perm}(\mathcal{M}) = \sum_{\pi \in S_N} \prod_{i=1}^N \langle w_i | z_{\pi(i)} \rangle$. 

Unlike the Slater determinant for completely antisymmetric fermionic systems, which provides a phase cancellation mechanism via the sign function $\text{sgn}(\pi)$ that allows an $\mathcal{O}(N^3)$ computational complexity through Gaussian elimination, the permanent of a matrix lacks this destructive interference. For generic quantum states, evaluating this permanent constitutes a \#P-hard problem \cite{valiant1979complexity, wei2010matrix}. As established in Appendix~\ref{app:complexity}, the restricted rank-2 structure of the matrices in the symmetric subspace circumvents this combinatorial explosion, allowing for polynomial-time evaluation via Barvinok's algorithm \cite{barvinok1996}. However, directly computing these permanents via algorithmic matrix decompositions remains algebraically dense and offers little immediate geometric intuition regarding the orthogonal subspace.

To bypass this algebraic density and achieve a direct, deterministic geometric construction, one can exploit the time-reversal symmetry inherent in the $SU(2)$ algebra rather than relying solely on permutation symmetry. The time-reversal operator $\Theta = \exp(-i\pi S_y)\mathcal{K}$, where $\mathcal{K}$ denotes complex conjugation, maps any Majorana star $z$ to its antipodal point on the Bloch sphere, $\tilde{z} = -1/z^*$. This symmetry guarantees strict orthogonality between a spin-$1/2$ state and its antipodal counterpart, i.e., $\langle \tilde{z} | z \rangle = 0$. 

By constructing a spin coherent state $|\alpha_k\rangle = |\tilde{z}_k\rangle^{\otimes N}$ such that all $N$ stars degenerate exactly at the antipodal point of a constituent star $z_k$ of $|\psi\rangle$, the permanent of the overlap matrix experiences a geometric collapse. The sum of $N!$ permutations simplifies into a single product containing the vanishing term $\langle \tilde{z}_k | z_k \rangle$, rendering the inner product strictly zero. This geometric property mirrors the behavior of the Bargmann analytic function $B_\psi(z) = \langle z | \psi \rangle$, whose roots are precisely the time-reversed (antipodal) images of the Majorana stars \cite{bargmann1961hilbert}---the same antipodal map recorded in Eq.~\eqref{eq:Q_pure}. 

This geometric collapse facilitates a deterministic mathematical framework, known as the antipodal basis algorithm, to construct the entire orthogonal complement without evaluating the permanent. Provided that $|\psi\rangle$ possesses $N$ distinct Majorana stars (non-degenerate roots), the $N$ coherent states $\{|\alpha_k\rangle\}_{k=1}^N$ formed by their respective antipodes are linearly independent and thus constitute a complete basis for the $N$-dimensional orthogonal subspace $\mathcal{H}_\psi^\perp$. 

Any generic orthogonal state $|\phi\rangle$ can subsequently be parameterized as a linear combination of these basis states, $|\phi\rangle = \sum_{k=1}^N c_k |\alpha_k\rangle$, for $c_k \in \mathbb{C}$. Identifying the corresponding $N$ stellar coordinates of $|\phi\rangle$ then merely requires finding the roots of its associated characteristic Majorana polynomial. Numerically, this is equivalent to solving the eigenvalue problem of a companion matrix, which operates with an $\mathcal{O}(N^3)$ complexity \cite{wilkinson1959evaluation, bini2000design}. In scenarios where $|\psi\rangle$ exhibits degenerate stars, the deficit in linearly independent basis states can be systematically resolved by taking spatial derivatives of the coherent states with respect to the Bargmann variable. 

Ultimately, while the low-rank permanent evaluation guarantees computational tractability, it is the time-reversal symmetry that provides the crucial geometric shortcut, elegantly translating the algebraic problem of orthogonal state construction into an intuitive mapping of antipodal roots on the Bloch sphere.

\bibliographystyle{unsrt}
\bibliography{references_new}

\end{document}